\documentclass[10pt,reqno,a4paper]{amsart}
\usepackage[british]{babel}
\usepackage{amssymb,amstext,amsthm,eucal,bbm,mathrsfs,amscd,bbold}
\usepackage[margin=1in]{geometry}
\usepackage{array}
\usepackage[svgnames,table]{xcolor}
\usepackage[T1]{fontenc}
\usepackage[utf8]{inputenc}
\usepackage{enumitem}
\usepackage[small]{eulervm}
\usepackage{tgpagella}
\usepackage{verbatim}
\usepackage{booktabs} 
\usepackage{caption}
\usepackage{tikz,pgfplots}
\pgfplotsset{compat=1.14}
\usepgfplotslibrary{fillbetween}
\usetikzlibrary{calc,arrows,cd,shapes.geometric,positioning,through}
\tikzset{carrollian/.style={draw,regular polygon,regular polygon sides=3,fill=DarkOrange,color=DarkOrange,minimum width=1pt,scale=0.25}}
\tikzset{lorentzian/.style={draw,regular polygon,regular polygon
    sides=4,fill=red!70!black,color=red!70!black,minimum width=1pt,scale=0.3}}
\usepackage[unicode]{hyperref}
\hypersetup{%
  pdftitle   = {Kinematical superspaces},
  pdfkeywords = {Lie superalgebra, kinematical, aristotelian, superspace, Klein supergeometry},
  pdfauthor  = {José Figueroa-O'Farrill, Ross Grassie},
  pdfcreator = {\LaTeX\ with package \flqq hyperref\frqq},
  linkcolor=NavyBlue,
  citecolor=ForestGreen,
  urlcolor=OrangeRed,
  anchorcolor=OrangeRed,
  colorlinks=true
}
\theoremstyle{plain}
\newtheorem{lemma}{Lemma}

\theoremstyle{definition}
\newtheorem{definition}{Definition}

\newcommand{\g}{\mathfrak{g}}
\newcommand{\h}{\mathfrak{h}}
\renewcommand{\a}{\mathfrak{a}}

\renewcommand{\r}{\mathfrak{r}}

\newcommand{\z}{\mathfrak{z}}
\newcommand{\so}{\mathfrak{so}}

\renewcommand{\sp}{\mathfrak{sp}}
\newcommand{\osp}{\mathfrak{osp}}

\renewcommand{\k}{\mathfrak{k}}
\newcommand{\s}{\mathfrak{s}}
\newcommand{\sa}{\mathfrak{sa}}

\renewcommand{\v}{\mathfrak{v}}
\newcommand{\m}{\mathfrak{m}}
\renewcommand{\H}{H}
\newcommand{\be}{\boldsymbol{e}}

\renewcommand{\v}{\boldsymbol{v}}

\newcommand{\J}{\boldsymbol{J}}
\newcommand{\B}{\boldsymbol{B}}
\renewcommand{\P}{\boldsymbol{P}}
\newcommand{\Q}{\boldsymbol{Q}}

\newcommand{\bc}{\mathbb{c}}
\newcommand{\sJ}{\mathsf{J}}
\newcommand{\sB}{\mathsf{B}}
\newcommand{\sP}{\mathsf{P}}
\newcommand{\sQ}{\mathsf{Q}}
\newcommand{\sV}{\mathsf{V}}
\newcommand{\cV}{\mathscr{V}}
\newcommand{\cJ}{\mathscr{J}}
\newcommand{\cS}{\mathscr{S}}
\newcommand{\Agr}{\mathcal{A}}
\newcommand{\Kgr}{\mathcal{K}}
\newcommand{\Hgr}{\mathcal{H}}
\newcommand{\eE}{\mathcal{C}^\infty}
\newcommand{\eN}{\mathcal{N}}
\newcommand{\eO}{\mathcal{O}}
\newcommand{\Cl}{C\ell}
\renewcommand{\Re}{\operatorname{Re}}
\renewcommand{\Im}{\operatorname{Im}}

\newcommand{\Ad}{\operatorname{Ad}}
\newcommand{\vol}{\operatorname{vol}}
\newcommand{\id}{\mathbb{1}}
\newcommand{\BB}{\mathbb{B}}
\newcommand{\PP}{\mathbb{P}}
\newcommand{\JJ}{\mathbb{J}}
\newcommand{\ii}{\mathbb{i}}
\newcommand{\jj}{\mathbb{j}}
\newcommand{\kk}{\mathbb{k}}
\newcommand{\hh}{\mathbb{h}}
\newcommand{\hhbar}{\overline{\mathbb{h}}}
\newcommand{\bb}{\mathbb{b}}
\newcommand{\pp}{\mathbb{p}}
\newcommand{\ppbar}{\overline{\mathbb{p}}}
\newcommand{\qq}{\mathbb{q}}
\newcommand{\uu}{\mathbb{u}}
\newcommand{\xx}{\mathbb{x}}
\newcommand{\qqbar}{\overline{\mathbb{q}}}
\newcommand{\sbar}{\overline{s}}
\newcommand{\ubar}{\overline{u}}
\newcommand{\HH}{\mathbb{H}}
\newcommand{\RP}{\mathbb{RP}}
\newcommand{\RR}{\mathbb{R}}

\newcommand{\ZZ}{\mathbb{Z}}

\newcommand{\Aut}{\operatorname{Aut}}
\newcommand{\GL}{\operatorname{GL}}
\newcommand{\SO}{\operatorname{SO}}
\newcommand{\SU}{\operatorname{SU}}
\newcommand{\Sp}{\operatorname{Sp}}
\newcommand{\Spin}{\operatorname{Spin}}
\newcommand{\End}{\operatorname{End}}

\newcommand{\zLC}{\mathsf{LC}}
\newcommand{\zAdSC}{\mathsf{AdSC}}
\newcommand{\zdSC}{\mathsf{dSC}}
\newcommand{\zdS}{\mathsf{dS}}
\newcommand{\zAdS}{\mathsf{AdS}}
\newcommand{\zC}{\mathsf{C}}
\newcommand{\zG}{\mathsf{G}}
\newcommand{\zS}{\mathsf{S}}
\newcommand{\zTS}{\mathsf{TS}}
\newcommand{\zAdSG}{\mathsf{AdSG}}
\newcommand{\zdSG}{\mathsf{dSG}}
\newcommand{\ztAdSG}{\mathsf{AdSG}}
\newcommand{\ztdSG}{\mathsf{dSG}}
\newcommand{\EE}{\mathbb{E}}
\newcommand{\MM}{\mathbb{M}}

\newcommand{\spn}[1]{\operatorname{span}_\RR\left\{#1\right\}}
\newcommand{\eff}[1]{\colorbox{blue!30}{$#1$}}
\newcommand{\ari}[1]{\colorbox{green!30}{$#1$}}
\newcommand{\non}[1]{\colorbox{gris!20}{\textcolor{gris!50}{$#1$}}}
\definecolor{gris}{rgb}{0.5,0.5,0.5}
\newcommand{\zero}{{\color{gris}0}}
\allowdisplaybreaks[0]
\numberwithin{equation}{section}
\begin{document}

\title{Kinematical superspaces}

\author[Figueroa-O'Farrill]{José Figueroa-O'Farrill}
\author[Grassie]{Ross Grassie}
\address{Maxwell Institute and School of Mathematics, The University
  of Edinburgh, James Clerk Maxwell Building, Peter Guthrie Tait Road,
  Edinburgh EH9 3FD, Scotland, United Kingdom}
\email[JMF]{\href{mailto:j.m.figueroa@ed.ac.uk}{j.m.figueroa@ed.ac.uk}}
\email[RG]{\href{mailto:s1131494@sms.ed.ac.uk}{s1131494@sms.ed.ac.uk}}
\begin{abstract}
  We classify $N{=}1$ $d=4$ kinematical and aristotelian Lie
  superalgebras with spatial isotropy, but not necessarily parity nor
  time-reversal invariance. Employing a quaternionic formalism which
  makes rotational covariance manifest and simplifies many of the
  calculations, we find a list of $43$ isomorphism classes of Lie
  superalgebras, some with parameters, whose (nontrivial) central
  extensions are also determined.  We then classify their
  corresponding simply-connected homogeneous $(4|4)$-dimensional
  superspaces, resulting in a list of $27$ homogeneous superspaces,
  some with parameters, all of which are reductive.  We determine the
  invariants of low rank and explore how these superspaces are related
  via geometric limits.
\end{abstract}
\thanks{EMPG-19-19}
\maketitle
\tableofcontents

\section{Introduction}
\label{sec:introduction}

Four-dimensional rigid supersymmetry first appeared in 1971 in a paper
\cite{Golfand:1971iw} of Golfand and Likhtman, which is to our
knowledge the first appearance of what is now known as the $N{=}1$
$d{=}4$ Poincaré superalgebra. A few years later, Zumino
\cite{Zumino:1977av} studied rigid supersymmetry in $\zAdS_4$, based
on the simple Lie superalgebra $\osp(1|4)$. For many years these two
were the only known $N{=}1$ $d{=}4$ Lie superalgebras. They are both
$(10|4)$-dimensional, and, in fact, the Poincaré superalgebra can be
exhibited as a contraction of $\osp(1|4)$ à la Inönü--Wigner. If we
wish to extend ($N{=}1$ $d{=}4$) supersymmetry beyond Minkowski and
anti~de~Sitter spacetimes, we are faced with a choice. One can study
$N{=}1$ supersymmetry algebras associated to other four-dimensional
lorentzian manifolds, as in the Lie algebraic approach of
\cite{deMedeiros:2016srz}, which results in Lie superalgebras which
are filtered deformations of subalgebras of the Poincaré superalgebra.
These filtered deformations have dimension $(n|4)$ for $n \leq 10$, and
hence, in most cases, some of the spacetime symmetry is broken. A
second approach, which is the one taken here, is to keep the
dimension of the superalgebra fixed at $(10|4)$, but sacrificing the
existence of a lorentzian metric.

In short, the present paper extends (in dimension four) the recent
classification \cite{Figueroa-OFarrill:2018ilb} of spatially-isotropic
homogeneous spacetimes, whose geometric properties were further
studied in \cite{Figueroa-OFarrill:2019sex}, to a classification of
$(4|4)$-dimensional simply-connected spatially-isotropic homogeneous
superspaces.  In particular, we classify the $(10|4)$-dimensional Lie
superalgebras with spatial isotropy.  (See later for a precise
definition.)

It is a natural question to ask, as Bacry and Lévy-Leblond did half a
century ago \cite{MR0238545}, what the possible kinematics are.  This
question translates into the geometric problem of classifying the
spacetimes which admit a transitive action of a kinematical Lie group.  To
answer this question, one first needs to classify kinematical Lie groups
and then study their possible homogeneous spaces.  If we allow the
ambiguity of classifying homogeneous spaces up to coverings (or,
equivalently, classifying the simply-connected homogeneous spaces),
this problem has a largely algebraic solution: namely, the
classification of pairs $(\k,\h)$, where $\k$ is a kinematical Lie
algebra and $\h$ a suitable subalgebra.  With every such pair $(\k,\h)$
(subject to some mild conditions) there is associated a unique
simply-connected homogeneous space $M = \Kgr/\Hgr$, where $\Kgr$ is a
simply-connected (and connected) kinematical Lie group with Lie
algebra $\k$ and $\Hgr$ is the connected subgroup generated by $\h$.
On $M$, the generators of $\k$ act as infinitesimal rotations, boosts
and spatio-temporal translations, whereas the generators of $\h$ act
as infinitesimal rotations and boosts about a choice of ``origin''
determined by the subgroup $\Hgr$ itself.

Let us restrict ourselves to the case of four spacetime dimensions. In
their pioneering paper \cite{MR0238545}, Bacry and Lévy-Leblond
presented a classification of kinematical Lie algebras subject to the
assumptions of the existence of automorphisms interpretable as parity
and time-reversal. These ``by no means compelling'' assumptions were
removed in \cite{MR857383}, resulting in the classification of
kinematical Lie algebras (with spatial isotropy) up to
isomorphism. Already in these papers, the observation was made that
every such kinematical Lie algebra $\k$ (of dimension 10) admits a
six-dimensional subalgebra $\h$ so that the pair $(\k,\h)$, \emph{if
  geometrically realisable as a homogeneous space}, is a
four-dimensional spatially isotropic homogeneous spacetime of a
kinematical Lie group.  The precise relation between pairs $(\k,\h)$
and homogeneous spacetimes is a little subtle, and this problem was
revisited in \cite{Figueroa-OFarrill:2018ilb}, arriving at the
classification of simply-connected spatially-isotropic homogeneous
spacetimes which is summarised in Table~\ref{tab:spacetimes} below.
(The results in \cite{Figueroa-OFarrill:2018ilb} are not restricted to
four spacetime dimensions, but already in the four-dimensional case
they refine and slightly correct the list in \cite{MR857383}.)  We
choose a basis where $\k$ is spanned\footnote{The boosts generators
  $B_i$ are absent in the aristotelian spacetimes.} by
$\{J_i, B_i, P_i, H\}$ and $\h$ is spanned by $\{J_i, B_i\}$, so that
the pair $(\k,\h)$ is uniquely determined by specifying the Lie
brackets in this basis.  We use a standard shorthand notation for the
Lie brackets, where $[H,\B] = \B$ stands for $[H,B_i] = B_i$,
$[\J,\B] = \B$ stands for $[J_i,B_j] = \epsilon_{ijk} B_k$ and
$[\B,\P] = H$ stands for $[B_i,P_j]= \delta_{ij} H$, et cetera.  As
already discussed in the original papers \cite{MR0238545,MR857383},
the set of isomorphism classes of kinematical Lie algebras is
partially ordered by contractions, which manifest themselves
geometrically as limits between the homogeneous spacetimes.  Such
limits are discussed at length in \cite{Figueroa-OFarrill:2018ilb}.

\begin{table}[h!]
  \centering
  \caption{Simply-connected spatially-isotropic homogeneous spacetimes}
  \label{tab:spacetimes}
  \rowcolors{2}{blue!10}{white}
  \resizebox{\textwidth}{!}{
    \begin{tabular}{l|*{5}{>{$}l<{$}}|l}\toprule
      \multicolumn{1}{c|}{Label} & \multicolumn{5}{c|}{Nonzero Lie brackets in addition to $[\J,\J] = \J$, $[\J, \B] = \B$, $[\J,\P] = \P$} & \multicolumn{1}{c}{Comments}\\\midrule
      \hypertarget{S1}{$\MM^4$} & [H,\B] = -\P & & [\B,\B] = -\J & [\B,\P] = H & & Minkowski\\
      \hypertarget{S2}{$\zdS_4$} & [H,\B] = -\P & [H,\P] = -\B & [\B,\B]= -\J & [\B,\P] = H & [\P,\P]= \J & de~Sitter\\
      \hypertarget{S3}{$\zAdS_4$} & [H,\B] = -\P & [H,\P] = \B & [\B,\B]= -\J & [\B,\P] = H & [\P,\P] = -\J & anti~de~Sitter\\
      \midrule
      \hypertarget{S4}{$\EE^4$} &  [H,\B] = \P & & [\B,\B] = \J & [\B,\P] = H & & euclidean\\
      \hypertarget{S5}{$S^4$} &  [H,\B] = \P & [H,\P] = -\B & [\B,\B]= \J & [\B,\P] = H & [\P,\P]= \J & sphere\\
      \hypertarget{S6}{$H^4$} &  [H,\B] = \P & [H,\P] = \B & [\B,\B]= \J & [\B,\P] = H & [\P,\P] = -\J & hyperbolic space\\
      \midrule
      \hypertarget{S7}{$\zG$} & [H,\B] = -\P & & & & & galilean spacetime\\
      \hypertarget{S8}{$\zdSG$} & [H,\B] = -\P & [H,\P] = -\B & & & & galilean de~Sitter ($\zdSG= \ztdSG_{\gamma=-1}$)\\
      \hypertarget{S9}{$\ztdSG_\gamma$} & [H,\B] = -\P & [H,\P] = \gamma\B + (1+\gamma)\P & & & & torsional galilean de~Sitter ($\gamma\in (-1,1]$) \\
      \hypertarget{S10}{$\zAdSG$} & [H,\B] =  -\P & [H,\P] = \B & & & & galilean anti~de~Sitter  ($\zAdSG = \ztAdSG_{\chi=0}$)\\
      \hypertarget{S11}{$\zAdSG_\chi$} & [H,\B] = -\P & [H,\P] = (1+\chi^2) \B + 2\chi \P & & & & torsional galilean anti~de~Sitter ($\chi>0$) \\
      \midrule
      \hypertarget{S13}{$\zC$} & & & & [\B,\P] = H & & carrollian spacetime\\
      \hypertarget{S14}{$\zdSC$} & & [H,\P] = -\B & & [\B,\P] = H & [\P,\P] = \J & carrollian de~Sitter\\
      \hypertarget{S15}{$\zAdSC$} & & [H,\P] = \B & & [\B,\P] = H & [\P,\P] = -\J & carrollian anti~de~Sitter\\
      \hypertarget{S16}{$\zLC$} & [H,\B] = \B & [H,\P] = -\P & & [\B,\P] = H - \J & & carrollian light cone\\
      \midrule
      \hypertarget{A21}{$\zS$} & & & & & & aristotelian static \\
      \hypertarget{A22}{$\zTS$} & & [H,\P] = \P & & & & torsional aristotelian static\\
      \hypertarget{A23p}{$\RR\times S^3$} & & & & & [\P,\P] = \J & Einstein static universe\\
      \hypertarget{A23m}{$\RR\times H^3$} & & & & & [\P,\P] = - \J & hyperbolic Einstein static universe\\
      \bottomrule
    \end{tabular}
  }
\end{table}

The homogeneous spacetimes in Table~\ref{tab:spacetimes} fall into
different classes, depending on the invariant structures that they
possess.  From top to bottom, we have the lorentzian spacetimes, the
riemannian spaces, the galilean spacetimes, the carrollian spacetimes
and finally the aristotelian spacetimes.  Aristotelian spacetimes are
homogeneous spaces of aristotelian Lie groups, where the boosts are
absent.  Many aristotelian Lie algebras arise as quotients of
kinematical Lie algebras by the ideal generated by the boosts, when
the boosts do generate an ideal.  However not all aristotelian Lie
algebras arise in this way, which motivated their classification in
\cite{Figueroa-OFarrill:2018ilb}.

The lorentzian spaces in the table ($\hyperlink{S1}{\MM^4}$,
$\hyperlink{S2}{\zdS_4}$ and $\hyperlink{S3}{\zAdS_4}$) are maximally
symmetric and homogeneous spaces of the Poincaré group,
$\Spin(4,1) \cong \Sp(1,1)$ and $\Spin(3,2) \cong \Sp(4,\RR)$,
respectively.  The last two isomorphisms are the spin representations,
which shows that whereas the irreducible spinor representation of
$\Spin(3,2)$ is real and four-dimensional, that of $\Spin(4,1)$ is
quaternionic and two-dimensional.

This paper aims to answer the question of what are the possible
``super-kinematics'' (in four spacetime dimensions).  We will give a
full answer for the case of $N{=}1$ supersymmetry or, equivalently,
for the case of four real supercharges.  In other words, we classify
the superspaces which superise the homogeneous spacetimes in
Table~\ref{tab:spacetimes}.  More precisely, we classify
(simply-connected, spatially-isotropic) $(4|4)$-dimensional
homogeneous superspaces of kinematical Lie supergroups.  As in the
classical (i.e., non-supersymmetric) case, we will work at the
algebraic level and will classify pairs $(\s,\h)$, where $\s$ is a
kinematical Lie superalgebra and $\h$ an admissible subalgebra,
concepts which will be defined carefully in the paper.

In a way, the superspaces in this paper belong to the same family as
the well-known Minkowski and AdS superspaces, which are recalled in
Appendix~\ref{app:lorentzian}.  Two features shared by these two
superspaces is that their corresponding Lie superalgebras
$\s = \s_{\bar 0} \oplus \s_{\bar 1}$ are such that $\s_{\bar 0}$ is a
kinematical Lie algebra (Poincaré and $\so(3,2)$, respectively) and
the odd subspace $\s_{\bar 1}$ is a four-dimensional real
representation of $\s_{\bar 0}$ whose restriction to the rotational
subalgebra $\r \subset \s_{\bar 0}$ is the four-dimensional real
spinor representation of $\r \cong \sp(1)$; that is, it is the
one-dimensional quaternionic representation of $\sp(1)$ but thought of
as a real vector space.  We will say that $\s$ is an \emph{$N{=}1$
  supersymmetric extension} of the kinematical Lie algebra
$\s_{\bar 0}$ or a \emph{kinematical Lie superalgebra}, for short.
One of the main results in this paper is the classification of
kinematical Lie superalgebras up to isomorphism.

We are certainly not the first authors to ask what are the possible
``super-kinematics'' and indeed there are papers \cite{MR769149,
  MR1723340,CampoamorStursberg:2008hm, Huang:2014ega} which give
partial answers to that question. In \cite{MR769149} the authors
depart from the list of kinematical Lie algebras in \cite{MR0238545}
and consider their $N{=}1$ supersymmetric extensions while still
requiring the existence of automorphisms corresponding to parity and
time-reversal. They do this by solving the Jacobi identities for the
superalgebra, having fixed the action of the rotational generators ab
initio. Their list consists of those Lie superalgebras which can be
obtained by contraction from the anti de Sitter superalgebra
$\osp(1|4)$. In \cite{MR1723340}, among other results unrelated to the
present paper, the authors study some of the contractions of the anti
de Sitter superalgebra $\osp(1|4)$, paying particular attention to
(para-)Poincaré, galilean and Newton--Hooke superalgebras. In
\cite{CampoamorStursberg:2008hm}, the authors classify the kinematical
contractions of $\osp(1|4)$ (and also of the corresponding anti de
Sitter Lie algebras of order $3$) and arrive at supersymmetric
extensions of the Poincaré, galilean, Carroll and Newton--Hooke Lie
algebras. Finally, in \cite{Huang:2014ega} the authors classify the
contractions of $\osp(1|4)$ and in addition contract the associated
superspaces. Some of these contractions, particularly those which
result in a galilean superalgebra have also been studied by other
authors (see, e.g.,
\cite{Puzalowski:1978rv,Palumbo:1978gx,Clark:1983ne,deAzcarraga:1991fa});
although in some cases the resulting superalgebra is not an extension
of the galilean algebra but its universal central extension, the
Bargmann algebra.

In this paper, we give a fuller answer to the question, in that we do
not require the existence of parity nor time-reversal automorphisms
and hence we depart not from the kinematical Lie algebras in
\cite{MR0238545}, but from those in
\cite{MR857383,Figueroa-OFarrill:2017ycu}.  In particular, our list of
Lie superalgebras includes, but substantially extends, the Lie
superalgebras which can be constructed as contractions of $\osp(1|4)$.
Our approach is as follows.  We will first classify (up to
isomorphism) the $N{=}1$ supersymmetric extensions of the kinematical
Lie algebras (with three-dimensional space isotropy) listed, for
convenience, in Table~\ref{tab:kla}.  We work in full generality, but
in the end restrict attention to those superalgebras where the bracket
$[\s_{\bar 1},\s_{\bar 1}]$ is nonzero.  We solve this problem by
employing a uniform quaternionic formalism for all kinematical Lie
algebras and solving the Jacobi identities.  The isomorphism classes
of kinematical Lie superalgebras are listed in Table~\ref{tab:klsa},
which is the first main result in this paper.  We then classify the
(effective, geometrically realisable) pairs $(\s,\h)$ where $\s$ is a
kinematical Lie superalgebra and $\h \subset \s_{\bar 0}$ an
admissible subalgebra. As we will show, the pair $(\s, \h)$ defines a
homogeneous supermanifold which ``superises'' the homogeneous
spacetime described by $(\s_{\bar 0}, \h)$.  The list of homogeneous
superspaces is contained in Table~\ref{tab:superspaces}, which is the
ultimate goal of this paper and contains our answer to the question of
what are the possible $(4|4)$-dimensional ``super-kinematics''.
Figure~\ref{fig:super-limits} illustrates the different limits which
relate these superspaces.

\subsection*{Reader's guide} The paper is organised as follows.  In
Section~\ref{sec:formalism}, we define the objects of interest and
state the classification problems that we will solve in this paper.
In Section~\ref{sec:quat-form}, we define kinematical Lie superalgebras
and set up the quaternionic formalism we will employ throughout the
paper.  In Section~\ref{sec:lie-super-brack}, we set out the strategy
we shall follow in classifying kinematical superalgebras.  In
Section~\ref{sec:some-preliminary-results}, we collect some useful
preliminary results we will use often and in
Section~\ref{sec:automorphisms} we discuss the nature of the
automorphisms of kinematical superalgebras.

In Section~\ref{sec:kinem-lie-super}, we classify the kinematical and
aristotelian Lie superalgebras, arriving at Tables~\ref{tab:klsa} and
\ref{tab:alsa}.  In doing so, we had to determine the automorphisms of
the kinematical Lie algebras, which are summarised in
Table~\ref{tab:aut-kla}.  Once having classified the kinematical and
aristotelian Lie superalgebras, we determine their (nontrivial)
central extensions in Section~\ref{sec:central-extensions}.  For later
use, we need to determine the automorphisms of the Lie superalgebras
(which fix the rotational subalgebra) and this is done in
Section~\ref{sec:autom-kinem-lie}.

In Section~\ref{sec:homog-supersp}, we classify the pairs $(\s,\h)$ and
hence the simply-connected homogeneous superspaces, which are listed
in Table~\ref{tab:superspaces}. In that table we list, in particular,
the underlying homogeneous kinematical or aristotelian spacetime for
each of our superspaces.  In Section~\ref{sec:low-rank-invariants}, we
determine the invariant tensors of low rank in each of the superspaces
in Table~\ref{tab:superspaces}.  In Section~\ref{sec:limits-betw-supersp},
we explore how the superspaces in Table~\ref{tab:superspaces} are
related via geometric limits, arriving at the picture in
Figure~\ref{fig:super-limits}, which is to be contrasted with the
similar picture (see Figure~\ref{fig:sub-limits}) for the homogeneous
spacetimes.  Finally, in Section~\ref{sec:conclusions}, we offer some
conclusions and point to possible extensions of this work.

This paper contains the details of two classifications: kinematical
Lie superalgebras and their associated superspaces.  As such it is
somewhat lengthy and somewhat technical.  Readers who are pressed for
time might benefit from some hints about navigating the paper towards
the main results.  In order to arrive at these results we reformulated
the problem in terms of quaternions and this formalism is described in
Section~\ref{sec:quat-form}.  The Lie algebraic classifications are
the subject of Section~\ref{sec:kinem-lie-super}, but the main results
are Table~\ref{tab:klsa} for the kinematical superalgebras and
Table~\ref{tab:alsa} for the aristotelian superalgebras.  The
(nontrivial) central extensions are tabulated in
Table~\ref{tab:central-ext}.  Section~\ref{sec:homog-supersp} contains
the classification of the superspaces, starting with
Section~\ref{sec:homog-superm}, which explains the infinitesimal
description of the superspaces in terms of super Lie pairs $(\s,\h)$,
and ending with Table~\ref{tab:superspaces}, which lists the
superspaces together with the underlying spacetime and a description
of the corresponding Lie superalgebra.
Section~\ref{sec:limits-betw-supersp} discusses how these superspaces
relate to each other via geometric limits, leading to
Figure~\ref{fig:super-limits}.  The figures and the tables are
hyperlinked for ease of navigation.

\section{Basic definitions and the statement of the problem}
\label{sec:formalism}

In this section, we set up the classification problems of kinematical
Lie superalgebras and homogeneous kinematical superspaces and
introduce the quaternionic formalism we shall employ in the rest of
the paper.

\subsection{Kinematical Lie algebras}
\label{sec:kinem-lie-algebr}

Let $\k$ be a kinematical Lie algebra (with
three-dimensional spatial isotropy).  It is a real 10-dimensional Lie
algebra with a subalgebra $\r \cong \so(3)$ and such that under the
adjoint action of $\r$, $\k$ decomposes as $\k = \r \oplus 2 V \oplus
\RR$, where $V$ is the three-dimensional irreducible vector
representation of $\r$ and $\RR$ is the trivial one-dimensional scalar
representation.  A real basis for $\k$ is given by $J_i$, $B_i$, $P_i$
and $H$, where $i =1,2,3$, where $J_i$ span $\r$, $B_i$ and $P_i$ span
the two copies of $V$ and $H$ is a scalar generator.  The Lie brackets
common to all kinematical Lie algebras are (using summation
convention):
\begin{equation}\label{eq:klb}
  [J_i, J_j] = \epsilon_{ijk} J_k \qquad [J_i, B_j] = \epsilon_{ijk}
  B_k \qquad [J_i, P_j] = \epsilon_{ijk} P_k \qquad\text{and}\qquad
  [J_i, H] = 0.
\end{equation}
Such kinematical Lie algebras were classified up to isomorphism
by Bacry and Nuyts\cite{MR857383} (see also
\cite{Figueroa-OFarrill:2017ycu}) completing the earlier
classification of Bacry and Lévy-Leblond \cite{MR0238545} of
kinematical Lie algebras admitting time-reversal and parity
automorphisms.  Table~\ref{tab:kla} summarises the classification by
listing the Lie brackets, in addition to the ones in
equation~\eqref{eq:klb}.  We use the by now standard abbreviated
notation, where the vector indices are not explicitly written down, so
that, for instance,
\begin{equation}
  \begin{split}
    [H,\B] = -\P &\qquad\text{stands for}\qquad [H,B_i] = - P_i \\
    [\B, \P] = H &\qquad\text{stands for}\qquad [B_i, P_j] = \delta_{ij} H\\
    [\P,\P] = \J &\qquad\text{stands for}\qquad [P_i, P_j] = \epsilon_{ijk} J_k,
  \end{split}
\end{equation}
et cetera.  In this abbreviated notation, the brackets in
equation~\eqref{eq:klb} are written as
\begin{equation}
  [\J, \J] = \J \qquad [\J, \B] = \B \qquad [\J, \P] = \P \qquad\text{and}\qquad
  [\J, H] = 0.  
\end{equation}
The one-parameter families $3_\gamma$ and $4_\chi$ of kinematical Lie
algebras extend the lorentzian and euclidean Newton--Hooke Lie
algebras, which correspond to $3_{\gamma=-1}$ and $4_{\chi = 0}$,
respectively.

It should be remarked that the correspondence between kinematical Lie
algebras and their (simply-connected) homogeneous spacetimes is not
bijective: there are kinematical Lie algebras with no associated
homogeneous spacetimes and, conversely, there are kinematical Lie
algebras with which there are associated more than one homogeneous
spacetime.  In describing the spacetimes in
Table~\ref{tab:spacetimes}, we have changed basis in the kinematical
Lie algebra $\k$ to ensure that the stabiliser subalgebra $\h$ is
always spanned by $J_i$ and $B_i$.  This explains any perceived
discrepancy between Tables~\ref{tab:spacetimes} and \ref{tab:kla}.

\begin{table}[h!]
  \centering
  \caption{Kinematical Lie algebras}
  \label{tab:kla}
  \setlength{\extrarowheight}{2pt}
  \rowcolors{2}{blue!10}{white}
  \begin{tabular}{l|*{5}{>{$}l<{$}}|l}\toprule
    \multicolumn{1}{c|}{K\#} & \multicolumn{5}{c|}{Nonzero Lie brackets (besides $[\J,-]$ brackets)} & \multicolumn{1}{c}{Comment}\\\toprule
    \hypertarget{KLA1}{1} & & & & & & static \\
    \hypertarget{KLA2}{2} & [H ,\B] = -\P & & & & & galilean \\
    \hypertarget{KLA3}{3$_{\gamma\in[-1,1]}$} & [H ,\B] = \gamma \B & [H ,\P] = \P & & & & \\
    \hypertarget{KLA4}{4$_{\chi\geq0}$} & [H ,\B] = \chi \B + \P & [H ,\P] = \chi \P - \B & & & & \\
    \hypertarget{KLA5}{5} & [H ,\B] = \B + \P & [H , \P] = \P & & & &  \\
    \hypertarget{KLA6}{6} & & & [\B,\P] = H  & & & Carroll \\
    \hypertarget{KLA7}{7} & [H ,\B] = \P & & [\B,\P] = H  & [\B,\B] = \J & & euclidean \\
    \hypertarget{KLA8}{8}& [H ,\B] = - \P & & [\B,\P] = H  & [\B,\B] = - \J & & Poincaré \\
    \hypertarget{KLA9}{9}& [H ,\B] = \B & [H ,\P] = -\P &  [\B,\P] = H  - \J & & & $\so(4,1)$ \\
    \hypertarget{KLA10}{10}& [H ,\B] = \P & [H ,\P] = -\B & [\B,\P] = H  &  [\B,\B]= \J &  [\P,\P] = \J & $\so(5)$ \\
    \hypertarget{KLA11}{11}& [H ,\B] = -\P & [H ,\P] = \B & [\B,\P] = H  &  [\B,\B]= -\J &  [\P,\P] = -\J & $\so(3,2)$ \\
    \hypertarget{KLA12}{12}& & & & [\B,\B]= \B &  [\P,\P] = \B-\J & \\
    \hypertarget{KLA13}{13}& & & & [\B,\B]= \B & [\P,\P] = \J-\B & \\
    \hypertarget{KLA14}{14}& & & & [\B,\B] = \B & &  \\
    \hypertarget{KLA15}{15}& & & & [\B, \B] = \P & &  \\
    \hypertarget{KLA16}{16}& & [H ,\P] = \P & & [\B,\B] = \B & &  \\
    \hypertarget{KLA17}{17}& [H ,\B] = -\P & & & [\B,\B] = \P & &  \\
    \hypertarget{KLA18}{18}& [H ,\B] = \B & [H ,\P] = 2\P & & [\B,\B] = \P & & \\\bottomrule
  \end{tabular}
\end{table}

\subsection{Kinematical Lie superalgebras}
\label{sec:quat-form}

We start by defining the objects of interest.

\begin{definition}
  An ($N{=}1$) \textbf{kinematical Lie superalgebra} (with
  three-dimensional space isotropy) is a real Lie superalgebra
  $\s = \s_{\bar 0} \oplus \s_{\bar 1}$, where $\s_{\bar 0} = \k$ is a
  kinematical Lie algebra (with three-dimensional space isotropy) and
  $\s_{\bar 1} = S$, where $S$ is a representation of $\k$ which
  extends the four-dimensional real spinor representation of the
  rotational subalgebra $\r$.
\end{definition}

Under the isomorphism $\r \cong \sp(1)$, we may take $S$ to be a copy
of the quaternions and the action of $\r$ on $S$ is essentially given
by left quaternion multiplication.  Let us be more precise.  We shall
denote the quaternions by $\HH$ and the quaternion units by
$\ii,\jj,\kk$, where $\ii^2 = -1$, $\jj^2 = -1$ and
$\ii\jj= \kk = - \jj\ii$.

Let us define the following injective real linear maps (again using
the summation convention):
\begin{equation}\label{eq:quat-basis-s}
  \begin{split}
    \sJ : \Im\HH \to \s &\qquad\text{given by}\qquad \sJ(\omega) =
    \omega_i J_i \qquad\text{for}\qquad \omega = \omega_1 \ii +
    \omega_2 \jj + \omega_3 \kk \in \Im\HH\\
    \sB : \Im\HH \to \s &\qquad\text{given by}\qquad \sB(\beta) =
    \beta_i B_i \qquad\text{for}\qquad \beta = \beta_1 \ii +
    \beta_2 \jj + \beta_3 \kk \in \Im\HH\\
    \sP : \Im\HH \to \s &\qquad\text{given by}\qquad \sP(\pi) =
    \pi_i P_i \qquad\text{for}\qquad \pi = \pi_1 \ii +
    \pi_2 \jj + \pi_3 \kk \in \Im\HH\\
    \sQ : \HH \to \s &\qquad\text{given by}\qquad \sQ(s) =
    s_a Q_a \qquad\text{for}\qquad s = s_1 \ii +
    s_2 \jj + s_3 \kk + s_4 \in \HH,
  \end{split}
\end{equation}
where $(Q_1,Q_2,Q_3,Q_4)$ is a real basis for $\s_{\bar 1}$.  The
(nonzero) Lie brackets common to all kinematical Lie superalgebras are
then given in terms of quaternion multiplication as follows:
\begin{equation}\label{eq:klsa-brackets-quat}
  \begin{split}
    [\sJ(\omega),\sJ(\omega')] &= \tfrac12 \sJ([\omega,\omega'])\\
    [\sJ(\omega),\sB(\beta)] &= \tfrac12 \sB([\omega,\beta])\\
    [\sJ(\omega),\sP(\pi)] &= \tfrac12 \sP([\omega,\pi])\\
    [\sJ(\omega),\sQ(s)] &= \tfrac12 \sQ(\omega s)\\
  \end{split}
\end{equation}
where $\omega,\omega',\beta,\pi \in \Im\HH$ and $s \in \HH$ and where
$\omega s$ is the quaternion product and
$[\omega,\beta] := \omega \beta - \beta \omega$, et cetera, are
quaternion commutators.  One can easily check that the Jacobi
identities involving at least two vectors in $\r$ are satisfied by
virtue of the associativity of quaternion multiplication.  For each
kinematical Lie algebra $\k$, the additional Lie brackets can also be
written quaternionically.  For example,
\begin{equation}
  \begin{split}
    [H, \B] = - \P &\qquad\text{becomes}\qquad [H, \sB(\beta)] = -\sP(\beta)\\
    [\B, \P] = H &\qquad\text{becomes}\qquad [\sB(\beta), \sP(\pi)] =
    \Re(\bar\beta\pi) H = - \Re(\beta\pi) H\\
    [\P,\P] = \J &\qquad\text{becomes}\qquad [\sP(\pi), \sP(\pi')] =
    \tfrac12 \sJ([\pi,\pi']),
  \end{split}
\end{equation}
et cetera.

\subsection{Lie superalgebra brackets}
\label{sec:lie-super-brack}

Let $\s$ be a kinematical Lie superalgebra where $\s_{\bar 0} = \k$ is
a kinematical Lie algebra from Table~\ref{tab:kla}.  To determine $\s$
we need to specify the additional Lie brackets: $[H,\Q]$, $[\B,\Q]$,
$[\P,\Q]$ and $[\Q,\Q]$, subject to the Jacobi identity.  There are
four components to the Jacobi identity in a Lie superalgebra $\s =
\s_{\bar 0} \oplus \s_{\bar 1}$:
\begin{enumerate}
\item The $(\s_{\bar 0}, \s_{\bar 0},\s_{\bar 0})$ Jacobi identity simply
  says that $\s_{\bar 0}$ is a Lie algebra, which in our case is one
  of the kinematical Lie algebras $\k$ in Table~\ref{tab:kla}.
\item The $(\s_{\bar 0}, \s_{\bar 0},\s_{\bar 1})$ Jacobi identity says that
  $\s_{\bar 1}$ is a representation of $\s_{\bar 0}$ and, by
  restriction, also a representation of any Lie subalgebra of
  $\s_{\bar 0}$: for example, $\r$ in our case.
\item The $(\s_{\bar 0}, \s_{\bar 1},\s_{\bar 1})$ Jacobi identity says that
  the component of the Lie bracket $\bigodot^2 \s_{\bar 1} \to
  \s_{\bar 0}$ is $\s_{\bar 0}$-equivariant.  In particular, in our
  case, it is $\r$-equivariant.
\item The $(\s_{\bar 1}, \s_{\bar 1},\s_{\bar 1})$ component does not
  seem to have any representation-theoretic reformulation and needs to
  be checked explicitly.
\end{enumerate}

Our strategy will be the following.  We shall first determine the
space of $\r$-equivariant brackets $[H,\Q]$, $[\B,\Q]$, $[\P,\Q]$ and
$[\Q,\Q]$, which will turn out to be a $22$-dimensional real vector
space $\cV$. For each kinematical Lie algebra $\k = \s_{\bar 0}$ in
Table~\ref{tab:kla}, we then determine the algebraic variety
$\cJ \subset \cV$ cut out by the Jacobi identity.  We are eventually
interested in \emph{supersymmetry} algebras and hence we will restrict
attention to Lie superalgebras $\s$ for which $[\Q,\Q] \neq 0$, which
define a sub-variety $\cS \subset \cJ$.  The isomorphism classes of
kinematical Lie superalgebras (with $[\Q,\Q]\neq 0)$ are in one-to-one
correspondence with the orbits of $\cS$ under the subgroup
$G \subset \GL(\s_{\bar 0}) \times \GL(\s_{\bar 1})$ which acts by
automorphisms of $\k = \s_{\bar 0}$, since we have fixed $\k$ from the
start.  The group $G$ contains not just the automorphisms of the
kinematical Lie algebra $\k$ which act trivially on $\r$, but also
automorphisms which are induced by automorphisms of the quaternion
algebra.  We shall return to an explicit description of such
automorphisms below.

Let us start by determining the $\r$-equivariant brackets: $[H,\Q]$,
$[\B,\Q]$, $[\P,\Q]$ and $[\Q,\Q]$.  The bracket $[H,\Q]$ is an
$\r$-equivariant endomorphism of the spinor module $\Q$.  If we
identify $\r$ with the imaginary quaternions and $\Q$ with the
quaternions, the action of $\r$ on $\Q$ is via left quaternion
multiplication.  The endomorphisms of the representation $S$ which
commute with the action of $\r$ consist of left multiplication by
reals and right multiplication by quaternions, but for real numbers,
left and right multiplications agree, since the reals are central in
the quaternion algebra.  Hence the most general $\r$-equivariant
$[H,\Q]$ bracket takes the form
\begin{equation}\label{eq:hq-bracket}
  [H,\sQ(s)] = \sQ(s \hh) \qquad\text{for some}\qquad \hh = h_1 \ii + h_2
  \jj + h_3 \kk + h_4 \in \HH.
\end{equation}

The brackets $[\B,\Q]$ and $[\P,\Q]$ are $\r$-equivariant
homomorphisms $V \otimes S \to S$, where $V$ and $S$ are the vector
and spinor modules of $\so(3)$.  There is an $\r$-equivariant map
$V \otimes S \to S$ given by the ``Clifford action'', which in this
language is left multiplication by $\Im \HH$ on $\HH$.  Its kernel is
the 8-dimensional real representation $W$ of $\r$ with spin
$\frac32$.
Therefore, the space of
$\r$-equivariant homomorphisms $V \otimes S \to S$ is isomorphic to
the space of $\r$-equivariant endomorphisms of $S$, which, as we saw
before, is a copy of the quaternions.  In summary, the $[\B,\Q]$ and
$[\P,\Q]$ brackets take the form
\begin{equation}\label{eq:bq-pq-brackets}
  \begin{split}
    [\sB(\beta), \sQ(s)] &= \sQ(\beta s \bb) \qquad\text{for some}\qquad \bb = b_1 \ii + b_2
    \jj + b_3 \kk + b_4 \in \HH\\
    [\sP(\pi), \sQ(s)] &= \sQ(\pi s \pp) \qquad\text{for some}\qquad \pp = p_1 \ii + p_2
    \jj + p_3 \kk + p_4 \in \HH,
  \end{split}
\end{equation}
for all $\beta,\pi \in \Im\HH$ and $s \in \HH$.

Finally, we look at the $[\Q,\Q]$ bracket, which is an
$\r$-equivariant linear map $\bigodot^2 S \to \k = \RR \oplus 3 V$.
The symmetric square $\bigodot^2S$ is a 10-dimensional representation
of $\r$ which decomposes as $\RR \oplus 3 V$.  Indeed, on $S$ we have
an $\r$-invariant inner product given by
\begin{equation}
  \left<s_1,s_2\right> = \Re (\sbar_1 s_2)\qquad\text{where}\qquad
  s_1,s_2 \in \HH.
\end{equation}
It is clearly invariant under left multiplication by unit quaternions:
$\left<u s_1, u s_2\right> = \left<s_1,s_2\right>$ for all $u \in
\Sp(1)$.  We can use this inner product to identify $\bigodot^2 S$ with the
symmetric endomorphisms of $S$: linear maps $\lambda : S \to S$ such
that $\left<\lambda(s_1), s_2\right> = \left<s_1, \lambda
  (s_2)\right>$.  Letting $L_\qq$ and $R_\qq$ denote left and
right quaternion multiplication by $\qq \in \HH$, the space of symmetric
endomorphisms of $S$ is spanned by the identity endomorphism and
$L_\ii R_\ii$, $L_\ii R_\jj$, $L_\ii R_\kk$, $L_\jj R_\ii$, $L_\jj
R_\jj$, $L_\jj R_\kk$, $L_\kk R_\ii$, $L_\kk R_\jj$ and $L_\kk
R_\kk$.  The nine non-identity symmetric endomorphisms transform
under $\r$ according to three copies of $V$.  Since $\r$ acts on $S$
via left multiplication, it commutes with the $R_\qq$ and hence the
three copies of $V$ are
\begin{equation}
  \spn{L_\ii R_\ii, L_\jj R_\ii, L_\kk R_\ii} \oplus \spn{L_\ii R_\jj,
    L_\jj R_\jj, L_\kk R_\jj} \oplus \spn{L_\ii R_\kk, L_\jj R_\kk,
    L_\kk R_\kk}.
\end{equation}
The space of $\r$-equivariant linear maps $\bigodot^2 S \to 3 V
\oplus \RR$ is thus isomorphic to the space of $\r$-equivariant
endomorphisms of $\RR \oplus 3 V = \RR \oplus (\RR^3 \otimes V)$,
which is given by
\begin{equation}
  \End_\r\left(\RR \oplus (\RR^3 \otimes V) \right) \cong \End(\RR) \oplus
  \left(\End(\RR^3) \otimes \id_V\right).
\end{equation}
In summary, the $\r$-equivariant $[\Q,\Q]$ bracket is given by
polarisation from the following
\begin{equation}
  [\sQ(s), \sQ(s)] = c_0 |s|^2 H + \Re(\sbar \JJ s \bc_1)
  + \Re(\sbar \BB s  \bc_2) + \Re(\sbar \PP s
  \bc_3),
\end{equation}
where $c_0 \in \RR$,
$\bc_1, \bc_2, \bc_3 \in \Im\HH$ and
where we have introduced the shorthands
\begin{equation}
  \JJ = J_1 \ii + J_2 \jj + J_3 \kk, \quad   \BB = B_1 \ii + B_2 \jj +
  B_3 \kk, \quad\text{and}\quad \PP = P_1 \ii + P_2 \jj + P_3 \kk.
\end{equation}
Notice that if $\omega \in \Im\HH$, then $\sJ(\omega) = \Re(\bar\omega
\JJ)$, and similarly $\sB(\beta) = \Re(\bar\beta \BB)$ and $\sP(\pi) =
\Re(\bar\pi \PP)$, for $\beta,\pi \in \Im \HH$, so that we can rewrite
the $[\Q,\Q]$ bracket as
\begin{equation}\label{eq:QQdiagonal}
  [\sQ(s), \sQ(s)] =c_0 |s|^2 H - \sJ(s \bc_1 \sbar) - \sB(s \bc_2
  \sbar) - \sP(s \bc_3 \sbar),
\end{equation}
which polarises to give
\begin{equation}\label{eq:QQ}
  [\sQ(s), \sQ(s')] = c_0 \Re(\sbar s') H - \tfrac12 \sJ(s'
  \bc_1 \sbar + s \bc_1 \sbar') - \tfrac12
  \sB(s' \bc_2 \sbar + s \bc_2 \sbar') -
  \tfrac12 \sP(s' \bc_3 \sbar + s \bc_3 \sbar').
\end{equation}

In summary, we have that the $\r$-equivariant brackets by which we
extend the kinematical Lie algebra $\k$ live in a 22-dimensional real
vector space of parameters $\hh,\bb,\pp \in \HH$, $\bc_1,
\bc_2, \bc_3 \in \Im \HH$ and $c_0 \in \RR$.

\subsection{Some preliminary results}
\label{sec:some-preliminary-results}

As mentioned above, one of the components of the Jacobi identity for
the Lie superalgebra $\s$ says that $\s_{\bar 1}$ is an $\s_{\bar 0}$-module,
where $\s_{\bar 0}=\k$ is the underlying kinematical Lie algebra.  The
Jacobi identity
\begin{equation}
  [X,[Y,\sQ(s)]] - [Y,[X,\sQ(s)]] = [[X,Y],\sQ(s)] \qquad\text{for all
    $X,Y\in\k$}
\end{equation}
gives relations between the parameters $\hh,\bb,\pp \in \HH$ appearing in
the Lie brackets.

\begin{lemma}\label{lem:kmod}
  The following relations between $\hh,\bb,\pp \in \HH$ are implied by the
  corresponding $\k$-brackets:
  \begin{equation}
    \begin{split}
      [H,\B] = \lambda \B + \mu \P & \implies [\bb,\hh] = \lambda \bb + \mu \pp\\
      [H,\P] = \lambda \B + \mu \P & \implies [\pp,\hh] = \lambda \bb + \mu \pp\\
      [\B,\B] = \lambda \B + \mu \P + \nu \J & \implies \bb^2 = \tfrac12 \lambda \bb + \tfrac12 \mu \pp + \tfrac14 \nu\\
      [\P,\P] = \lambda \B + \mu \P + \nu \J & \implies \pp^2 = \tfrac12 \lambda \bb + \tfrac12 \mu \pp + \tfrac14 \nu\\
      [\B,\P] = \lambda H & \implies \bb \pp + \pp \bb = 0\quad\text{and}\quad [\bb,\pp] = \lambda \hh.
    \end{split}
  \end{equation}
\end{lemma}

\begin{proof}
  The $[\H,\B,\Q]$ Jacobi identity says for all $\beta \in \Im\HH$ and
  $s \in\HH$,
  \begin{equation}
    [[H,\sB(\beta)],\sQ(s)] = [H, [\sB(\beta),\sQ(s)]] - [\sB(\beta),
    [H, \sQ(s)]],
  \end{equation}
  which becomes
  \begin{equation}
    \lambda \sQ(\beta s \bb) + \mu \sQ(\beta s \pp) = \sQ(\beta s \bb \hh) -
    \sQ(\beta s \hh \bb).
  \end{equation}
  Since $\sQ$ is real linear and injective, it follows that
  \begin{equation}
    \lambda \beta s \bb + \mu \beta s \pp = \beta s [\bb,\hh],
  \end{equation}
  which, since it must hold for all $\beta \in \Im\HH$ and $s \in
  \HH$, becomes
  \begin{equation}
    [\bb,\hh] = \lambda \bb + \mu \pp,
  \end{equation}
  as desired.  Similarly, the $[\H,\P,\Q]$ Jacobi identity gives the
  second equation in the lemma.  The third equation follows from the
  $[\B,\B,\Q]$ Jacobi identity, which says that for all
  $\beta,\beta'\in\Im\HH$ and $s\in\HH$,
  \begin{equation}
    [[\sB(\beta),\sB(\beta')],\sQ(s)] =
    [\sB(\beta),[\sB(\beta'),\sQ(s)]] -  [\sB(\beta'),[\sB(\beta),\sQ(s)]],
  \end{equation}
  which becomes
  \begin{equation}
    \tfrac12 \lambda \sQ([\beta,\beta'] s \bb) + \tfrac12 \mu
    \sQ([\beta,\beta'] s \pp) + \tfrac14 \nu \sQ([\beta,\beta']s) =
    \sQ(\beta\beta's \bb^2) - \sQ(\beta'\beta s \bb^2).
  \end{equation}
  Again by linearity and injectivity of $\sQ$, this is equivalent to
  \begin{equation}
    \tfrac12 \lambda [\beta,\beta'] s \bb + \tfrac12 \mu
    [\beta,\beta'] s \pp + \tfrac14 \nu [\beta,\beta']s =
    [\beta,\beta']s \bb^2,
  \end{equation}
  which, being true for all $\beta,\beta'\in\Im\HH$ and $s \in \HH$, gives
  \begin{equation}
    \tfrac12 \lambda \bb + \tfrac12 \mu \pp + \tfrac14 \nu = \bb^2,
  \end{equation}
  as desired.  The fourth identity in the lemma follows similarly from
  the $[\P,\P,\Q]$ Jacobi identity.  Finally, we consider the
  $[\B,\P,\Q]$ Jacobi identity, which says that for all
  $\beta,\pi\in\Im\HH$ and $s \in\HH$,
  \begin{equation}
    [[\sB(\beta),\sP(\pi)],\sQ(s)] = [\sB(\beta),[\sP(\pi),\sQ(s)]] -
    [\sP(\pi),[\sB(\beta),\sQ(s)]],
  \end{equation}
  which expands to
  \begin{equation}
    -\lambda \Re(\beta\pi) \sQ(s\hh) = \sQ(\beta\pi s \pp \bb) - \sQ(\pi\beta s \bb \pp)
  \end{equation}
  or, equivalently,
  \begin{equation}\label{eq:BPQ-aux}
    -\lambda \Re(\beta\pi) s \hh = \beta\pi s \pp \bb - \pi \beta s \bb \pp,
  \end{equation}
  for all $\beta,\pi \in \Im\HH$ and $s \in \HH$.  For any two
  imaginary quaternions $\beta,\pi$, we have that
  \begin{equation}
    \beta\pi = \tfrac12 [\beta,\pi] + \Re(\beta\pi),
  \end{equation}
  which allows us to rewrite equation~\eqref{eq:BPQ-aux} as
  \begin{equation}
    \Re(\beta\pi) s (\lambda \hh - \bb \pp + \pp \bb) + \tfrac12 [\beta,\pi] s (
    \pp \bb + \bb \pp) = 0.
  \end{equation}
  Taking $\beta = \pi$ and $s=1$ we see that $\lambda \hh = [\bb,\pp]$ and taking
  $\beta$ and $\pi$ to be orthogonal and $s=1$, that $\pp \bb + \bb \pp = 0$, as
  desired.
\end{proof}

The components $[H,\Q,\Q]$, $[\B,\Q,\Q]$ and $[\P,\Q,\Q]$ of the
Jacobi identity are best studied on a case-by-case basis, but the
$[\Q,\Q,\Q]$ component gives a universal condition.

\begin{lemma}\label{lem:qqq}
  The $[\Q,\Q,\Q]$ component of the Jacobi identity implies
  \begin{equation}
    c_0 \hh = \tfrac12 \bc_1  + \bc_2 \bb + \bc_3 \pp.
  \end{equation}
\end{lemma}

\begin{proof}
  The $[\Q,\Q,\Q]$ component of the Jacobi identity is totally
  symmetric and hence, by polarisation, it is uniquely determined by
  its value on the diagonal.  In other words, it is equivalent to
  \begin{equation}
    [[\sQ(s),\sQ(s)], \sQ(s)] \stackrel{!}{=} 0 \qquad\text{for all $s \in \HH$.}
  \end{equation}
  Using equation~\eqref{eq:QQdiagonal}, this becomes
  \begin{equation}
    [c_0 |s|^2 H - \sJ(s \bc_1 \sbar) - \sB(s
    \bc_2 \sbar) - \sP(s \bc_3 \sbar), \sQ(s)]
    \stackrel{!}{=} 0,
  \end{equation}
  which expands to
  \begin{equation}
    c_0 |s|^2 \sQ(sh) - \tfrac12 \sQ(s \bc_1 \sbar s) - \sQ(s
    \bc_2 \sbar s \bb) - \sQ(s
    \bc_3 \sbar s \pp) \stackrel{!}{=} 0.
  \end{equation}
  Since $\sQ$ is injective, this becomes
  \begin{equation*}
    |s|^2 s (c_0 \hh - \tfrac12 \bc_1 -  \bc_2 \bb -
    \bc_3 \pp)  \stackrel{!}{=} 0.
  \end{equation*}
  This must hold for all $s \in \HH$, so in particular for $s=1$,
  proving the lemma.
\end{proof}

\subsection{Automorphisms}
\label{sec:automorphisms}

As mentioned above, once we determine the sub-variety $\cS$ cut out by
the Jacobi identity, we need to quotient by the action of the subgroup
$G \subset \GL(\s_{\bar 0}) \times \GL(\s_{\bar 1})$ which acts by
automorphisms of $\s_{\bar 0} = \k$ in order to arrive at the
isomorphism classes of Lie superalgebras.  In this section, we describe
the subgroup $G$ in more detail.  There are two kinds of elements of
$G$, those which act trivially on the rotational subalgebra $\r$ and
those which do not.  The latter consist of inner automorphisms of
$\k$ which are generated infinitesimally by the adjoint action of
$\J$, $\B$ and $\P$.  The ones generated by $\J$ are particularly easy
to describe in the quaternionic formulation, and we shall do so now in
more detail.

Let $u \in \Sp(1)$ be a unit norm quaternion.  Conjugation by $u$
defines a homomorphism $\Ad : \Sp(1) \to \Aut(\HH)$ whose kernel is
the central subgroup of $\Sp(1)$ consisting of $\pm 1$.  It is a
classical result that these are all the automorphisms of $\HH$.  Hence
$\Aut(\HH) \cong \SO(3)$, acting trivially on the real quaternions and
rotating the imaginary quaternions.  The action of $\Aut(\HH)$ on $\s$
leaves $H$ invariant and acts on the remaining generators by
pre-composing the linear maps $\sJ$, $\sB$, $\sP$ and $\sQ$ with
$\Ad_u$.  More precisely, let $\widetilde H = H$,
$\widetilde\sJ = \sJ \circ \Ad_u$, $\widetilde\sB = \sB \circ \Ad_u$,
$\widetilde\sP = \sP \circ \Ad_u$ and $\widetilde\sQ = \sQ \circ \Ad_u$.
Since the Lie brackets of $\k$ are given in terms of quaternion
multiplication, this transformation is an automorphism of $\k$, and we
have a group homomorphism $\Aut(\HH) \to \Aut(\k)$.  The action on the
remaining brackets (those involving $\Q$) is as follows.
The Lie brackets of $\s$ which involve $\Q$ are given by
\begin{equation}
  \begin{split}
    [H, \sQ(s)] &= \sQ(sh)\\
    [\sJ(\omega), \sQ(s)] &= \tfrac12 \sQ(\omega s)\\
    [\sB(\beta), \sQ(s)] &= \sQ(\beta s \bb)\\
    [\sP(\pi), \sQ(s)] &= \sQ(\pi s \pp)\\
    [\sQ(s), \sQ(s)] &= c_0 |s|^2 H - \sJ(s \bc_1 \sbar) -
    \sB(s \bc_2 \sbar) - \sP(s \bc_3 \sbar),
  \end{split}
\end{equation}
and hence under conjugation by $u \in \Sp(1)$,
\begin{equation}
  \begin{split}
    [\widetilde H, \widetilde \sQ(s)] &= \widetilde\sQ(s\widetilde \hh)\\
    [\widetilde \sJ(\omega), \widetilde\sQ(s)] &= \tfrac12 \widetilde\sQ(\omega s)\\
    [\widetilde \sB(\beta), \widetilde\sQ(s)] &= \widetilde\sQ(\beta s \widetilde \bb)\\
    [\widetilde\sP(\pi), \widetilde\sQ(s)] &= \widetilde\sQ(\pi s \widetilde \pp)\\
    [\widetilde\sQ(s), \widetilde\sQ(s)] &= c_0 |s|^2 \widetilde H - \widetilde\sJ(s \widetilde{\bc}_1 \sbar) -
    \widetilde\sB(s \widetilde{\bc}_2 \sbar) - \widetilde\sP(s \widetilde{\bc}_3 \sbar),
  \end{split}
\end{equation}
where $\widetilde \hh = \bar u \hh u$, $\widetilde \bb = \bar u \bb
u$, $\widetilde \pp = \bar u \pp u$, and $\widetilde{\bc}_i = \bar u
\bc_i u$ for $i=1,2,3$.  In other words, the scalar parameters $c_0$,
$\Re \hh$, $\Re \bb$ and $\Re \pp$ remain inert, but the imaginary
quaternion parameters $\Im \hh, \Im \bb, \Im \pp, \bc_{1,2,3}$ are
simultaneously rotated.  We will use these automorphisms very often in
the sequel.

There are other automorphisms of $\k$ which do transform $\r$: those
are the inner automorphisms generated by $\B$ and $\P$. Their
description depends on the precise form of $\k$ but they will not play
a rôle in our discussion.

In addition to these, $G$ also consists of automorphisms of $\k$
which leave $\r$ intact.  If a linear map $\Phi: \s \to \s$ restricts to an
automorphism of $\k$, then it is in particular $\r$-equivariant.  The
most general $\r$-equivariant linear map $\Phi : \s \to \s$ sends
$(\J,H,\B,\P,\Q) \mapsto (\J, \widetilde H, \widetilde \B, \widetilde
\P, \widetilde \Q)$, where
\begin{equation}
  \begin{split}
    \widetilde H &= \mu H\\
    \widetilde \sB(\beta) &= a\sB(\beta) + c\sP(\beta) + e \sJ(\beta)\\
    \widetilde \sB(\beta) &= b\sB(\beta) + d\sP(\beta) + f \sJ(\beta)\\
    \widetilde \sQ(s) &= \sQ(s\qq)
  \end{split}
\end{equation}
where $\mu \in \GL(1,\RR) = \RR^\times$, $\qq \in \GL(1,\HH) =
\HH^\times$ and $\begin{pmatrix} \zero & a & b \\ \zero & c & d \\ 1 &
  e & f \end{pmatrix} \in \GL(3,\RR)$.  In
\cite[§§3.1]{Figueroa-OFarrill:2018ilb} we worked out the
automorphisms (which fix $\r$) of $\k$ a kinematical Lie algebra
isomorphic to one of \hyperlink{KLA1}{$\mathsf{K1}$}-\hyperlink{KLA11}{$\mathsf{K11}$} 
in Table~\ref{tab:kla}.  The
automorphisms of the remaining kinematical Lie algebras in the table
are listed below (see Table~\ref{tab:aut-kla}).  In particular, we
find that, although the precise form of the automorphisms depends on $\k$, a
common feature is that the coefficients $e,f$ are always zero, so we
will set them to zero from now on without loss of generality.

Assuming that the pair $(A=\begin{pmatrix}a & b \\ c & d
\end{pmatrix}, \mu) \in \GL(2,\RR) \times \RR^\times$ is an
automorphism of $\k = \s_{\bar 0}$, the brackets involving $\Q$ change
as follows:
\begin{equation}
  \begin{split}
    [\widetilde H, \widetilde\sQ(s)] &= \widetilde\sQ(s\widetilde\hh)\\
    [\widetilde \sB(\beta), \widetilde\sQ(s)] &= \widetilde\sQ(\beta s \widetilde\bb)\\
    [\widetilde \sP(\pi), \widetilde\sQ(s)] &= \widetilde\sQ(\pi s \widetilde\pp)\\
    [\widetilde\sQ(s), \widetilde\sQ(s)] &= \widetilde c_0 |s|^2
    \widetilde H - \widetilde\sJ(s \widetilde\bc_1 \sbar) - \widetilde
    \sB(s \widetilde \bc_2 \sbar) - \widetilde \sP(s \widetilde
    \bc_3\sbar),
  \end{split}
\end{equation}
where $\widetilde \sJ(\omega) = \sJ(\omega)$ and
\begin{equation}\label{eq:autk-on-params}
  \begin{aligned}[m]
    \widetilde\hh &= \mu \qq \hh \qq^{-1}\\
    \widetilde\bb &= \qq (a \bb + c \pp) \qq^{-1}\\
    \widetilde\pp &=  \qq (b \bb + d \pp) \qq^{-1}\\
    \widetilde c_0 &= c_0 \frac{|\qq|^2}{\mu}
  \end{aligned}
  \qquad\qquad
  \begin{aligned}[m]
    \widetilde \bc_1 &= \qq \bc_1 \qqbar\\
    \widetilde \bc_2 &= \frac{1}{ad - bc} \qq (d \bc_2 - b \bc_3) \qqbar\\
    \widetilde \bc_3 &= \frac{1}{ad - bc} \qq (a \bc_3 - c \bc_2) \qqbar.\\
  \end{aligned}
\end{equation}

In summary, the group $G$ by which we must quotient the sub-variety
$\cS$ cut out by the Jacobi identity (and $[\Q,\Q]\neq 0$) acts as
follows on the generators:
\begin{equation}
  \begin{split}
    \sJ &\mapsto \sJ \circ \Ad_u\\
    \sB &\mapsto a \sB \circ \Ad_u + c \sP \circ \Ad_u\\
    \sP &\mapsto b \sB \circ \Ad_u + d \sP \circ \Ad_u\\
    H &\mapsto \mu H\\
    \sQ &\mapsto \sQ \circ \Ad_u \circ R_\qq
  \end{split}
\end{equation}
where $\mu\in \RR$ and $\qq \in \HH$ are nonzero, $u \in \Sp(1)$  and
$A:=\begin{pmatrix}a & b \\ c & d \end{pmatrix} \in \GL(2,\RR)$ with
$(A,\mu)$ an automorphism of $\k$.

Let $\Aut_\r(\k)$ denote the subgroup of $\GL(2,\RR) \times \RR^\times$
consisting of such $(A,\mu)$.  These subgroups are listed in
\cite[§3.1]{Figueroa-OFarrill:2018ilb} for the kinematical Lie
algebras \hyperlink{KLA1}{$\mathsf{K1}$}-\hyperlink{KLA11}{$\mathsf{K11}$} 
in Table~\ref{tab:kla}.  We will collect them in
Table~\ref{tab:aut-kla} for convenience and in addition also record
them for the remaining kinematical Lie algebras 
\hyperlink{KLA12}{$\mathsf{K12}$}-\hyperlink{KLA18}{$\mathsf{K18}$} in
Table~\ref{tab:kla}.

\begin{table}[h!]
  \centering
  \caption{Automorphisms of kinematical Lie algebras (acting trivially
  on $\r$)}
  \label{tab:aut-kla}
  \begin{tabular}{l|>{$}l<{$}}\toprule
    \multicolumn{1}{c|}{K\#} & \multicolumn{1}{c}{Typical $(A,\mu) \in \GL(2,\RR) \times \RR^\times$}\\
    \toprule
    \hyperlink{KLA1}{1} & \left(\begin{pmatrix} a & b \\ c & d \end{pmatrix}, \mu\right) \\[10pt]
    \hyperlink{KLA2}{2} & \left(\begin{pmatrix} a & \zero \\ c & d \end{pmatrix}, \frac{d}{a}\right) \\[10pt]
    \hyperlink{KLA3}{3$_{\gamma\in(-1,1)}$} & \left(\begin{pmatrix} a & \zero \\ \zero &
        d \end{pmatrix}, 1\right) \\[10pt]
    \hyperlink{KLA3}{3$_{-1}$} & \left(\begin{pmatrix} a & \zero \\ \zero & d \end{pmatrix}, 1\right), \left(\begin{pmatrix} \zero & b \\ c & \zero \end{pmatrix}, -1\right) \\[10pt]
    \hyperlink{KLA3}{3$_{1}$} & \left(\begin{pmatrix} a & b \\ c & d \end{pmatrix}, 1\right) \\[10pt]
    \hyperlink{KLA4}{4$_{\chi>0}$} & \left(\begin{pmatrix} a & b \\ -b & a \end{pmatrix}, 1\right) \\[10pt]
    \hyperlink{KLA4}{4$_0$} & \left(\begin{pmatrix} a & b \\ -b & a \end{pmatrix}, 1\right), \left(\begin{pmatrix} a & b \\ b & -a \end{pmatrix}, -1\right) \\[10pt]
    \hyperlink{KLA5}{5} & \left(\begin{pmatrix} a & \zero \\ c & a \end{pmatrix}, 1\right)  \\[10pt]
    \hyperlink{KLA6}{6} & \left(\begin{pmatrix} a & b \\ c & d \end{pmatrix}, ad-bc\right) \\[10pt]
    \hyperlink{KLA7}{7},\hyperlink{KLA8}{8} & \left(\begin{pmatrix} 1 & \zero \\ c & d \end{pmatrix}, d \right), \left(\begin{pmatrix} -1 & \zero \\ c & d \end{pmatrix}, -d\right)  \\[10pt]
    \hyperlink{KLA9}{9} & \left(\begin{pmatrix} a & \zero \\ \zero & a^{-1} \end{pmatrix}, 1\right),  \left(\begin{pmatrix} \zero & b \\ b^{-1} & \zero \end{pmatrix}, -1\right) \\[10pt]
    \hyperlink{KLA10}{10},\hyperlink{KLA11}{11} & \left(\begin{pmatrix} a & b \\ -b & a \end{pmatrix}, 1\right), \left(\begin{pmatrix} a & b \\ b & -a \end{pmatrix}, -1\right),\quad a^2 + b^2 = 1\\[10pt]
    \hyperlink{KLA12}{12},\hyperlink{KLA13}{13} & \left(\begin{pmatrix} 1 & \zero \\ \zero & 1 \end{pmatrix}, \mu\right) , \left(\begin{pmatrix} 1 & \zero \\ \zero & -1 \end{pmatrix}, \mu\right) \\[10pt]
    \hyperlink{KLA14}{14} & \left(\begin{pmatrix} 1 & \zero \\ \zero & d \end{pmatrix}, \mu\right) \\[10pt]
    \hyperlink{KLA15}{15} & \left(\begin{pmatrix} a & \zero \\ c & a^2 \end{pmatrix}, \mu\right) \\[10pt]
    \hyperlink{KLA16}{16} & \left(\begin{pmatrix} 1 & \zero \\ \zero & d \end{pmatrix}, 1\right) \\[10pt]
    \hyperlink{KLA17}{17} & \left(\begin{pmatrix} a & \zero \\ c & a^2 \end{pmatrix}, a\right) \\[10pt]
    \hyperlink{KLA18}{18} & \left(\begin{pmatrix} a & \zero \\ \zero & a^2 \end{pmatrix}, 1\right) \\[10pt]
    \bottomrule
    \end{tabular}
\end{table}

\section{The classifications of kinematical and aristotelian Lie superalgebras}
\label{sec:kinem-lie-super}

In this section, we classify the supersymmetric extensions of the
kinematical Lie algebras in Table~\ref{tab:kla}.  In addition, we will
also classify aristotelian Lie superalgebras, as some of the
homogeneous supermanifolds we will encounter later on will turn out to
be superisations of the aristotelian homogeneous spacetimes classified
in \cite[App.~A]{Figueroa-OFarrill:2018ilb}.

\subsection{Classification of kinematical Lie superalgebras}
\label{sec:class-kinem-lie}

We now proceed to analyse each kinematical Lie algebra
$\k$ in Table~\ref{tab:kla} in turn and impose the Jacobi identity for
the corresponding Lie superalgebras extending $\k$.  We recall that we
are only interested in those Lie superalgebras where $[\Q,\Q] \neq 0$,
so $c_0,  \bc_1, \bc_2, \bc_3$ cannot all simultaneously vanish.

\subsubsection{Kinematical Lie algebras without supersymmetric extensions}
\label{sec:kinem-lie-algebr-1}

There are three kinematical Lie algebras which cannot be extended to a
kinematical superalgebra: \hyperlink{KLA9}{$\so(4,1)$},
\hyperlink{KLA10}{$\so(5)$} and the euclidean algebra
(\hyperlink{KLA7}{$\mathsf{K7}$} in Table~\ref{tab:kla}).

\subsubsection*{The euclidean algebra}
\label{sec:euclidean-algebra}

From Lemma~\ref{lem:kmod} we find that $\pp = \hh = 0$ and $\bb^2 =
\frac14$, so in particular $\bb \in \RR$, and from Lemma~\ref{lem:qqq}
we find that $\bc_2 \bb + \frac12 \bc_1 = 0$.  The
$[H,\Q,\Q]$ component of the Jacobi identity shows that
$\bc_2 = 0$, so that also $\bc_1 = 0$.  The
$[\P,\Q,\Q]$ component of the Jacobi identity is trivially satisfied,
whereas the $[\B,\Q,\Q]$ component shows that $\bc_3 = 0$
and also that $c_0 = 0$.  In summary, there is no kinematical
superalgebra extending the euclidean algebra for which $[\Q,\Q] \neq
0$; although there is a kinematical superalgebra where $[\sB(\beta),
\sQ(s)] = \pm \frac12 \sQ(\beta s)$, where both choices of sign
are related by an automorphism of $\k$: e.g., time reversal
$(\J,\B,\P,H) \mapsto (\J, -\B, \P, -H)$ or parity $(\J,\B,\P,H)
\mapsto (\J, -\B, -\P, H)$.

\subsubsection*{$\so(4,1)$}
\label{sec:so4-1}

In this case, Lemma~\ref{lem:kmod} gives that $\pp = \bb = 0$, but then
the $[\B,\P,\Q]$ component of the Jacobi identity cannot be satisfied,
showing that the $\so(3)$ representation on the spinor module $S$
cannot be extended to a representation of $\so(4,1)$.  The result
would be different for $N=2$ extensions, since $\so(4,1) \cong
\sp(1,1)$ does have an irreducible spinorial representation of
quaternionic dimension $2$.

\subsubsection*{$\so(5)$}
\label{sec:so5}

From Lemma~\ref{lem:kmod} we find from $[H,\B]=\P$ that $\pp = [\bb,\hh]$
and, in particular, $\pp \in \Im \HH$.  But then $[\P,\P] = \J$ says
that $\pp^2 = \frac14$, so that in particular $\pp \in \RR$ and nonzero,
which is a contradiction.  Again this shows that the spinor
representation $S$ of $\so(3)$ does not extend to a representation of
$\so(5)$ and again the conclusion would be different for $N=2$
extensions, since $\so(5) \cong \sp(2)$ does admit a quaternionic
representation of quaternionic dimension $2$.

\subsubsection{Lorentzian kinematical superalgebras}
\label{sec:lorentz-kinem-super}

The Poincaré Lie algebra (\hyperlink{KLA8}{$\mathsf{K8}$}) and
\hyperlink{KLA11}{$\so(3,2)$} are lorentzian isometry Lie algebras: of
Minkowski and anti~de~Sitter spacetimes, respectively.  It is of
course well known that such spacetimes admit $N{=}1$ superalgebras of
maximal dimension.  We treat them in this section for completeness.

\subsubsection*{The Poincaré superalgebra}
\label{sec:poinc-super}

From Lemma~\ref{lem:kmod} we see that $\pp = \hh = 0$ and that $\bb^2 =
-\tfrac14$, so that in particular $\bb \in \Im \HH$.  From
Lemma~\ref{lem:qqq} we see that $\tfrac12 \bc_1 + \bc_2 \bb = 0$.  The
$[\P,\Q,\Q]$ component of the Jacobi identity is trivially satisfied,
whereas the $[H,\Q,\Q]$ component forces $\bc_1 = \bc_2 = 0$ and the
$[\B,\Q,\Q]$ component says $\bc_3 = 2 c_0 \bb$.  Demanding $[\Q,\Q]
\neq 0$ requires $c_0 \neq 0$.

Using the quaternion automorphism, we can rotate $\bb$ so that $\bb =
\frac12 \kk$ and via the automorphism of the Poincaré Lie algebra
which rescales $H$ and $P$ by the same amount, we can bring $c_0 =
1$.  In summary, we have a unique isomorphism class of kinematical Lie
superalgebras extending the Poincaré Lie algebra and consisting in the
additional Lie brackets
\begin{equation}
  [\sB(\beta), \sQ(s)] = \tfrac12 \sQ(\beta s \kk)
  \qquad\text{and}\qquad [\sQ(s), \sQ(s)] = |s|^2 H - \sP(s\kk \sbar).
\end{equation}
We will show below in Section~\ref{sec:unpack-quat-notat} that $\s$ is
isomorphic to the Poincaré superalgebra defined in the Introduction.

\subsubsection*{The AdS superalgebra}
\label{sec:ads-superalgebra}

Here Lemma~\ref{lem:kmod} and Lemma~\ref{lem:qqq} give the following
relations:
\begin{equation}
  \pp = [\hh,\bb],\quad \bb=[\pp,\hh],\quad \hh = [\bb,\pp],\quad \bb^2 = -\tfrac14,\quad
  \pp^2 = -\tfrac14\quad\text{and}\quad c_0 \hh = \tfrac12 \bc_1 + \bc_2 \bb
  + \bc_3 \pp,
\end{equation}
and in addition $\bb \pp + \pp \bb = 0$, which simply states that $\bb \perp
\pp$.  These relations imply that $\bb,\pp,\hh \in \Im\HH$ and that $(2\bb, 2\pp,
2\hh)$ is an oriented orthonormal basis for $\Im\HH$.  The remaining
Jacobi identities give
\begin{equation}
  \bc_2 = - 2 c_0 \pp, \quad \bc_3 = 2 c_0 \bb \implies \bc_1 = -2 c_0 \hh,
\end{equation}
and some other relations which are identically satisfied.  If $c_0=0$
then $[\Q,\Q]=0$, so we requires $c_0\neq 0$ and hence
$(\frac{\bc_1}{c_0}, \frac{\bc_2}{c_0}, \frac{\bc_3}{c_0})$ defines a
negatively oriented, orthonormal basis for $\Im\HH$.  The automorphism
group of $\HH$ acts transitively on the space of orthonormal oriented
bases, so we can choose $(2\bb, 2\pp, 2\hh) = (\ii,\jj,\kk)$ without loss of
generality.

The resulting Lie superalgebra becomes
\begin{equation}
  \begin{split}
    [H,\sQ(s)] &= \tfrac12 \sQ(s\kk)\\
    [\sB(\beta),\sQ(s)] &= \tfrac12 \sQ(\beta s\ii)\\
    [\sP(\pi),\sQ(s)] &= \tfrac12 \sQ(\pi s\jj)\\
    [\sQ(s),\sQ(s)] &= c_0 \left(|s|^2 H + \sJ(s\kk \sbar) + \sB(s\jj\bar
    s) - \sP(s\ii\sbar)\right).
  \end{split}
\end{equation}
We may rescale $\Q$ to bring $c_0$ to a sign, but we can then change
the sign via the automorphism of $\k$ which sends $(\J,\B,\P,H)
\mapsto (\J,\P,\B,-H)$ and the inner automorphism induced by the
automorphism of $\HH$ which sends $(\ii,\jj,\kk) \mapsto
(\jj,\ii,-\kk)$.  In summary, there is a unique kinematical Lie
superalgebra with $[\Q,\Q]\neq 0$ extending $\k = \so(3,2)$: namely,
\begin{equation}
  \begin{split}
    [H,\sQ(s)] &= \tfrac12 \sQ(s\kk)\\
    [\sB(\beta),\sQ(s)] &= \tfrac12 \sQ(\beta s\ii)\\
    [\sP(\pi),\sQ(s)] &= \tfrac12 \sQ(\pi s\jj)\\
    [\sQ(s),\sQ(s)] &= |s|^2 H + \sJ(s\kk \sbar) + \sB(s\jj\sbar) -
    \sP(s\ii\sbar).
  \end{split}
\end{equation}
To show that this Lie superalgebra is isomorphic to $\osp(1|4)$ we may
argue as follows.  We first prove that $\s_{\bar 0}$ leaves
invariant a symplectic form on $\s_{\bar 1}$.  The most general
rotationally invariant bilinear form on $\s_{\bar 1}$ is given by
\begin{equation}
  \omega(\sQ(s_1), \sQ(s_2)) := \Re (s_1 \qq \sbar_2) \qquad\text{for
    some $\qq \in \HH$.}
\end{equation}
Indeed, if $u \in \Sp(1)$ then
\begin{equation}
  \begin{split}
    (u \cdot \omega)(\sQ(s_1), \sQ(s_2)) &= \omega(u^{-1} \cdot
    \sQ(s_1), u^{-1} \cdot \sQ(s_2))\\
    &= \omega (\sQ(\ubar s_1),  \sQ(\ubar s_2))\\
    &= \Re(\ubar s_1 \qq \sbar_2 u)\\
    &= \Re(s_1 \qq \sbar_2)\\
    &= \omega(\sQ(s_1), \sQ(s_2)).
  \end{split}
\end{equation}
Demanding that $\omega$ be invariant under the other generators $H,
\B, \P$, we find that $\qq = \mu \kk$ for some $\mu \in \RR$.  
Acting infinitesimally now,
\begin{equation}
  \begin{split}
    (H \cdot \omega)(\sQ(s_1), \sQ(s_2)) &= -\omega([H, \sQ(s_1)], \sQ(s_2)) - \omega(\sQ(s_1), [H, \sQ(s_2)])\\
    &= - \tfrac12 \omega(\sQ(s_1\kk), \sQ(s_2)) - \tfrac12\omega(\sQ(s_1), \sQ(s_2\kk))\\
    &= -\tfrac12 \Re(s_1\kk \qq \sbar_2) + \tfrac12 \Re(s_1 \qq \kk \sbar_2)\\
    &= \tfrac12 \Re(s_1 [\qq,\kk] \sbar_2),
  \end{split}
\end{equation}
which must vanish for all $s_1,s_2 \in S$, so that $[\qq,\kk] = 0$ and
hence $\qq = \lambda \id + \mu\kk$ for some $\lambda,\mu \in \RR$. A
similar calculation with $\B$ and $\P$ shows that $\qq$ must
anticommute with $\ii$ and $\jj$ and thus $\qq = \mu \kk$. So the
action of $\s_{\bar 0} \cong \so(3,2)$ on $\s_{\bar 1}$ defines a Lie
algebra homomorphism $\so(3,2) \to \sp(4,\RR)$, which is clearly
nontrivial. Since $\so(3,2)$ is simple, it is injective and a
dimension count shows that this is an isomorphism. But as
representations of $\so(3,2)$, $\odot^2\s_{\bar 1} \cong \wedge^2V$,
where $V$ is the 5-dimensional vector representation of $\s_{\bar 0}$,
and since $\wedge^2 V \cong \so(V) \cong \s_{\bar 0}$ we have that
there is one-dimensional space of $\s_{\bar 0}$-equivariant maps
$\odot^2 \s_{\bar 1} \to \s_{\bar 0}$. Since $[\Q,\Q] \neq 0$ the
bracket $\odot^2\s_{\bar 1} \to \s_{\bar 0}$ is an isomorphism. This
then shows that $\s$ is, by definition, isomorphic to $\osp(1|4)$.

\subsubsection{The Carroll superalgebra}
\label{sec:carroll-superalgebra}

For $\k$ the Carroll Lie algebra (\hyperlink{KLA6}{$\mathsf{K6}$} in
Table~\ref{tab:kla}), Lemma~\ref{lem:kmod} implies that
$\pp = \bb = \hh = 0$ and then Lemma~\ref{lem:qqq} says that
$\bc_1 = 0$.  The $[\B,\Q,\Q]$ Jacobi says that $\bc_3 =0$ and the
$[\P,\Q,\Q]$ Jacobi says that $\bc_2 = 0$.  The only nonzero bracket
involving $\Q$ is
\begin{equation}
  [\sQ(s), \sQ(s)] = c_0 |s|^2 H,
\end{equation}
which is nonzero for $c_0 \neq 0$.  If so, we can set $c_0 = 1$ via an
automorphism of $\k$ which rescales $H$ and $\P$, say, by $c_0$.  In
summary, there is a unique Carroll superalgebra with brackets
\begin{equation}
  [\sQ(s), \sQ(s)] = |s|^2 H,
\end{equation}
in addition to those of the Carroll Lie algebra itself.

\subsubsection{The galilean superalgebras}
\label{sec:galil-super}

For $\k$ the galilean Lie algebra (\hyperlink{KLA2}{$\mathsf{K2}$} in
Table~\ref{tab:kla}), Lemma~\ref{lem:kmod} says that $\bb=\pp = 0$ and
Lemma~\ref{lem:qqq} says that $\bc_1 = 2 c_0 \hh$.  The $[\B,\Q,\Q]$
Jacobi identity says that $\bc_1 = 0$ and $c_0 = 0$.  The $[\P,\Q,\Q]$
Jacobi identity is now identically satisfied, whereas the $[H,\Q,\Q]$
Jacobi identity gives
\begin{equation}
  \hh \bc_2 + \bc_2 \bar \hh = 0 \qquad\text{and}\qquad \bc_2 + \hh \bc_3 +
  \bc_3 \bar \hh = 0.
\end{equation}
Since $\bc_2$ and $\bc_3$ cannot both vanish, we see that this is only
possible if $\hh \in \Im\HH$ and hence these equations become $[\hh,\bc_2]
=0$ and $\bc_2 = [\bc_3,\hh]$.  There are two cases to consider,
depending on whether or not $\hh$ vanishes.  If $\hh=0$, then $\bc_2 = 0$
and $\bc_3$ is arbitrary.  If $\hh\neq 0$, then on the one hand $\bc_2$
is collinear with $\hh$, but also $\bc_2 = [\bc_3,\hh]$, which means that
$\bc_2=0$ so that $\bc_3 \neq 0$ is collinear with $\hh$.  In either
case, $\bc_3 \neq 0$ and $\hh = \psi \bc_3$, where $\psi \in \RR$
can be zero.

This gives rise to the following additional brackets
\begin{equation}
  [H, \sQ(s)] = \psi \sQ(s\bc_3) \qquad\text{and}\qquad [\sQ(s),
  \sQ(s)] = - \sP(s\bc_3 \sbar).
\end{equation}
We may use the automorphisms of $\HH$ to bring $\bc_3 = \phi \kk$, for
some nonzero $\phi \in \RR$.  We can set $\phi = 1$ by an automorphism
of $\k$ which rescales $\P$ and also $\B$ and $\H$ suitably.  This
still leaves the freedom to set $\psi = 1$ if $\psi \neq 0$.
In summary, we have two galilean superalgebras:
\begin{equation}
  [H, \sQ(s)] =
  \begin{cases}
    0\\
    \sQ(s \kk)
  \end{cases}
  \qquad\text{and}\qquad [\sQ(s),\sQ(s)] = - \sP(s\kk \sbar).
\end{equation}
The first one (where $[H,\Q] = 0$) is a contraction of the Poincaré
superalgebra, whereas the second (where $[H,\Q] \neq 0$) is not.

\subsubsection{Lie superalgebras associated with the static kinematical Lie algebra}
\label{sec:lie-super-assoc-1}

This is \hyperlink{KLA1}{$\mathsf{K1}$} in Table~\ref{tab:kla}.  In this
case, Lemma~\ref{lem:kmod} says that $\bb=\pp=0$ and
Lemma~\ref{lem:qqq} says that $\bc_1 = 2 c_0 \hh$.  The $[H,\Q,\Q]$
Jacobi identity says that $\hh \in \Im \HH$ and that $[\hh,\bc_i] = 0$
for $i=1,2,3$. Finally either the $[\B,\Q,\Q]$ or $[\P,\Q,\Q]$ Jacobi
identities say that $\bc_1 = 0$, so that $\hh c_0 = 0$.  This means
that either $\hh=0$ or else $c_0 = 0$ (or both).

There are several branches:
\begin{enumerate}
\item If $c_0 = 0$ and $\hh \neq 0$, $\bc_2$ and $\bc_3$ are collinear
  with $\hh$, but cannot both be zero.  Using automorphisms of the
  static kinematical Lie algebra and the ability to rotate vectors, we
  can bring $\hh = \tfrac12 \kk$, $\bc_2 = 0$ and $\bc_3 = \kk$, so that
  we have a unique Lie superalgebra in this case, with additional
  brackets
  \begin{equation}
    [H,\sQ(s)] = \tfrac12 \sQ(s\kk) \qquad\text{and}\qquad
    [\sQ(s),\sQ(s)] = - \sP(s\kk\sbar).
  \end{equation}
\item If $c_0 = 0$ and $\hh=0$, $\bc_2$ and $\bc_3$ are unconstrained,
  but not both zero.  We distinguish two cases, depending on whether
  or not they are linearly independent:
  \begin{enumerate}
  \item If they are linearly dependent, so that they are collinear,
    then we can use automorphisms to set $\bc_2$, say, to zero and
    $\bc_3 = \kk$.  This results in the Lie superalgebra
    \begin{equation}
      [\sQ(s),\sQ(s)] = - \sP(s\kk\sbar).
    \end{equation}
  \item If they are linearly independent, we can bring them to $\bc_2
    = \jj$ and $\bc_3 = \kk$, resulting in the Lie superalgebra
    \begin{equation}
      [\sQ(s),\sQ(s)] = - \sB(s\jj\sbar) - \sP(s\kk\sbar).
    \end{equation}
  \end{enumerate}
\item Finally, if $c_0\neq 0$, then $\hh=0$ and again $\bc_2$ and
  $\bc_3$ are unconstrained, but can now be zero.  Moreover we can
  rescale $H$ so that $c_0 = 1$.  We have three cases to consider,
  depending on whether they span a zero-, one- or two-dimensional real
  subspace of $\Im \HH$:
  \begin{enumerate}
  \item If $\bc_2 = \bc_3 = 0$ we have the Lie superalgebra
    \begin{equation}
      [\sQ(s),\sQ(s)] = |s|^2 H.
    \end{equation}
  \item If $\bc_2$ and $\bc_3$ span a line, then we may use the
    automorphisms to set $\bc_2 = 0$ and $\bc_3 = \kk$, resulting in
    the Lie superalgebra
    \begin{equation}
      [\sQ(s),\sQ(s)] = |s|^2 H - \sP(s\kk\sbar).
    \end{equation}
  \item Finally, if $\bc_2$ and $\bc_3$ are linearly independent, we
    may use the automorphisms to set $\bc_2 = \jj$ and $\bc_3 = \kk$,
    resulting in the Lie superalgebra
    \begin{equation}
      [\sQ(s),\sQ(s)] = |s|^2 H - \sB(s\jj\sbar) - \sP(s\kk\sbar).
    \end{equation}
  \end{enumerate}
\end{enumerate}

\subsubsection{Lie superalgebras associated with kinematical Lie algebra $\mathsf{K3}_\gamma$}
\label{sec:lie-super-assoc-3}

Here Lemma~\ref{lem:kmod} says that $\bb = \pp = 0$ and
Lemma~\ref{lem:qqq} says that $\bc_1 = 2 c_0 \hh$.  The $[\B,\Q,\Q]$
Jacobi identity says that $\bc_1 = 0$ and $c_0 = 0$, whereas the
$[\P,\Q,\Q]$ Jacobi identity offers no further conditions.  Finally,
the $[H,\Q,\Q]$ Jacobi identity gives two conditions
\begin{equation}
  \gamma \bc_2 = \hh \bc_2 + \bc_2 \bar \hh \qquad\text{and}\qquad
  \bc_3 = \hh \bc_3 + \bc_3 \bar \hh,
\end{equation}
which are equivalent to
\begin{equation}
  (\gamma - 2 \Re \hh) \bc_2 = [\Im \hh, \bc_2] \qquad\text{and}\qquad
  (1 - 2 \Re \hh) \bc_3 = [\Im \hh, \bc_3].
\end{equation}
We see that we must distinguish two cases: $\gamma = 1$ and $\gamma
\in [-1,1)$.

If $\gamma \neq 1$, then we have two cases, depending on whether $\Re
\hh = \frac12$ or $\Re \hh = \frac12 \gamma$.  In the former case, $\bc_2
= 0$ and $\Im \hh$ is collinear with $\bc_3 \neq 0$, whereas in the
latter, $\bc_3 = 0$ and $\Im \hh$ is collinear with $\bc_2 \neq 0$.

If $\gamma = 1$, then $\Re \hh = \frac12$ and $\bc_2$, $\Im \hh$ and
$\bc_3$ are all collinear, with at least one of $\bc_2$ and $\bc_3$
nonzero.  When $\gamma =1$, the automorphisms of $\kk$ include the
general linear group $\GL(2,\RR)$ acting on the two copies of the
vector representation.  Using this we can always assume that $\bc_2 =
0$ and $\bc_3 \neq 0$.

In either case, all nonzero vectors are collinear and we can rotate
them to lie along the $\kk$ axis.  In the case $\gamma = 1$, we have
a one-parameter family of Lie superalgebras:
\begin{equation}
  [H,\sQ(s)] = \tfrac12 \sQ(s(1+\lambda \kk)) \qquad\text{and}\qquad
  [\sQ(s),\sQ(s)] = - \sP(s\kk \sbar),
\end{equation}
where we have used the freedom to rescale $\P$ in order to set
$\bc_3 = \kk$.  This is also a Lie superalgebra for $\gamma \neq 1$.

If $\gamma \neq 1$, we have an additional one-parameter family of Lie
superalgebras:
\begin{equation}
  [H,\sQ(s)] = \tfrac12 \sQ(s(\gamma +\lambda \kk)) \qquad\text{and}\qquad
 [\sQ(s),\sQ(s)] = - \sB(s\kk \sbar).
\end{equation}

The parameter $\lambda$ is essential; that is, Lie superalgebras with
different values of $\lambda$ are not isomorphic.  One way to test
this is the following.  Let $[-,-]_\lambda$ denote the above Lie
bracket.  This satisfies the Jacobi identity for all $\lambda \in
\RR$.  Write it as $[-,-]_\lambda = (1-\lambda) [-,-]_0 + \lambda
[-,-]_1$.  The difference $[-,-]_1 - [-,-]_0$ is a cocycle of the Lie
superalgebra with $\lambda = 0$.  The parameter would be inessential
if and only if it is a coboundary.  One can check that this is not the
case.  This same argument shows that the parameters appearing in other
Lie superalgebras are essential as well.

\subsubsection{Lie superalgebras associated with kinematical Lie algebra $\mathsf{K4}_\chi$}
\label{sec:lie-super-assoc-4}

Here Lemma~\ref{lem:kmod} says $\bb = \pp = 0$ and Lemma~\ref{lem:qqq}
says that $\bc_1 = 2 c_0 \hh$.  Then either the $[\B,\Q,\Q]$ or
$[\P,\Q,\Q]$ Jacobi identities force $\bc_1=0$ and $c_0= 0$.  The
$[H,\Q,\Q]$ Jacobi identity results in the following two equations:
\begin{equation}
  \chi \bc_2 - \bc_3 = \hh \bc_2 + \bc_2 \bar \hh \qquad\text{and}\qquad
  \chi \bc_3 + \bc_2 = \hh \bc_3 + \bc_3 \bar \hh,
\end{equation}
or equivalently,
\begin{equation}
  (\chi - 2 \Re \hh) \bc_2 - \bc_3 = [\Im \hh, \bc_2] \qquad\text{and}\qquad
  (\chi - 2 \Re \hh) \bc_3 + \bc_2 = [\Im \hh, \bc_3].
\end{equation}
Taking the inner product of the first equation with $\bc_2$ and of the
second equation with $\bc_3$ and adding, we find
\begin{equation}
  (\chi - 2 \Re \hh) (|\bc_2|^2 + |\bc_3|^2) = 0, 
\end{equation}
and since $\bc_2$ and $\bc_3$ cannot both be zero, we see that $\Re \hh
= \frac\chi2$, and hence that
\begin{equation}
  [\Im \hh, \bc_2] = - \bc_3 \qquad\text{and}\qquad [\Im \hh, \bc_3] = \bc_2,
\end{equation}
so that $\bc_3 \perp \bc_2$.  This shows that $(\Im \hh, \bc_3, \bc_2)$
is an oriented orthogonal (but not necessarily orthonormal) basis.  We
can rotate them so that $\Im \hh = \phi \jj$, $\bc_3 = \psi \kk$ and
$\bc_2 = 2\phi\psi \ii$, but then we see that $\phi^2 =
\tfrac14$.  Using the automorphism of $\k$ which rescales $\B$ and
$\P$ simultaneously by the same amount we can assume that $\bc_3 =
\kk$ and hence if $\Im \hh = \pm \tfrac12 \jj$ then $\bc_2 = \pm \ii$.
But the two signs are related by the automorphism of $\HH$ which sends
$(\ii,\jj,\kk) \mapsto (-\ii,-\jj,\kk)$.  In summary, we have a
unique Lie superalgebra associated with this kinematical Lie algebra:
\begin{equation}
  [H,\sQ(s)] = \tfrac12 \sQ(s(\chi + \jj)) \qquad\text{and}\qquad
  [\sQ(s), \sQ(s)] = - \sB(s\ii\sbar) - \sP(s\kk\sbar).
\end{equation}

\subsubsection{Lie superalgebras associated with kinematical Lie algebra $\mathsf{K5}$}
\label{sec:lie-super-assoc}

Here Lemma~\ref{lem:kmod} says that $\bb=\pp=0$ and Lemma~\ref{lem:qqq}
says that $\bc_1 = 2 c_0 \hh$.  The $[\B,\Q,\Q]$ Jacobi identity forces
$c_0 = \bc_1 = 0$, which then makes the $[\P,\Q,\Q]$ Jacobi identity
be satisfied identically.  The $[H,\Q,\Q]$ Jacobi identity gives two
further equations
\begin{equation}
  \bc_2 = \hh \bc_2 + \bc_2 \bar \hh \qquad\text{and}\qquad \bc_2 + \bc_3
  = \hh \bc_3 + \bc_3 \bar \hh.
\end{equation}
The first equation is equivalent to
\begin{equation}
  (1 - 2 \Re(\hh)) \bc_2 = [\Im \hh, \bc_2].
\end{equation}
If $\bc_2 \neq 0$, then $\Re \hh = \frac12$ and $\Im \hh$ is collinear
with $\bc_2$. But then the second equation says that $\bc_2 = [\Im \hh,
\bc_3]$, which is incompatible with $\bc_2$ and $\Im \hh$ being
collinear.  Therefore $\bc_2 = 0$ and the second equation then says
that $\Re \hh = \frac12$  and $\Im \hh$ collinear with $\bc_3 \neq 0$.
We have the following additional brackets
\begin{equation}
  [H, \sQ(s)] = \tfrac12 \sQ(s (1 + \lambda \bc_3))
  \qquad\text{and}\qquad [\sQ(s),\sQ(s)] = - \sP(s \bc_3 \sbar),
\end{equation}
where $\lambda \in \RR$.  We may rotate $\bc_3$ to $\psi \kk$, for some
nonzero $\psi \in \RR$.  We can then rescale $\P$ and $\B$
simultaneously by the same amount to set $\psi = 1$.  In summary, we
are left with the following one-parameter family of Lie superalgebras:
\begin{equation}
  [H, \sQ(s)] = \tfrac12 \sQ(s (1 + \lambda \kk))
  \qquad\text{and}\qquad [\sQ(s),\sQ(s)] = - \sP(s \kk \sbar).
\end{equation}

As in the case of the Lie superalgebras associated with Lie algebra
\hyperlink{KLA3}{$\mathsf{K3}_\gamma$}, the parameter $\lambda$ is essential and Lie
superalgebras with different values of $\lambda$ are not isomorphic.

\subsubsection{Lie superalgebras associated with kinematical Lie algebra $\mathsf{K12}$}
\label{sec:lie-super-assoc-12}

Lemma~\ref{lem:kmod} says that $\bb^2 = \frac12 \bb$, so that
$\bb\in \RR$, $[\hh,\pp]=0$ and $\pp^2= \frac12 (\bb-\frac12)$, so
that $\pp \in \Im\HH$.  (In particular, $\bb\pp= 0$.)
Lemma~\ref{lem:qqq} does not simplify at this stage.  The $[H,\Q,\Q]$
Jacobi identity says that $c_0 \Re \hh = 0$ and that
$\hh \bc_i + \bc_i \bar \hh = 0$ for $i=1,2,3$.  The $[\B,\Q,\Q]$
Jacobi identity says that $\bb\bc_1 = 0$, $\bb\bc_3 = 0$ and
$\bc_1 = (2\bb -1) \bc_2$.  Finally, the $[\P,\Q,\Q]$ Jacobi identity
says that $c_0 \pp = 0$, among other conditions that will turn out not
to play a rôle.

We have two branches depending on the value of $\bb$:
\begin{enumerate}
\item If $\bb=0$, $\pp^2= -\frac14$, so that $c_0 = 0$.  This means $\bc_1
  + \bc_2 =0$ and $\bc_3 = 2 \bc_1 \pp$ and none of $\bc_{1,2,3}$ can
  vanish.  This means that $\Re \hh = 0$ and that $\hh$ and $\bc_i$ are
  collinear for all $i=1,2,3$.  Also $\hh$ and $\pp$ are collinear and
  this is inconsistent, unless $\hh = 0$: indeed, if $\pp$ and $\bc_i$ are
  collinear with $\hh \neq 0$, then $\bc_3 = 2\bc_1 \pp$ cannot be
  satisfied, since the LHS is imaginary but the RHS is real and both
  are nonzero.  Therefore we conclude that $\hh=0$.  The condition
  $\bc_3 = 2 \bc_1 \pp$ says that there exists $\psi > 0$ such that
  $(\psi^{-1} \bc_1, 2\pp, \psi^{-1} \bc_3)$ is an oriented
  orthonormal basis, which can be rotated to $(\ii,\jj,\kk)$.  In
  other words, we can write $\bc_1 = \psi \ii$, $\pp = \frac12\jj$
  and $\bc_3 = \psi \kk$, so that $\bc_2 = - \psi \ii$.  We may
  rescale $\sQ$ to bring $\psi =1$ and we may rotate $(\ii,\jj,\kk)
  \mapsto (-\ii,\jj,-\kk)$ to arrive at the following Lie
  superalgebra:
  \begin{equation}
    [\sP(\pi), \sQ(s)] = \tfrac12 \sQ(s\jj) \qquad\text{and}\qquad
    [\sQ(s),\sQ(s)] = \sJ(s\ii\sbar) - \sB(s\ii\sbar) + \sP(s\kk\bar
    s).
  \end{equation}
\item If $\bb=\frac12$, then $\pp=0$ and also $\bc_1= \bc_3 = 0$ and
  $\bc_2 = 2 c_0 \hh$ with $c_0 \neq 0$.  We have two sub-branches,
  depending on whether or not $\hh=0$.
  \begin{enumerate}
  \item If $\hh=0$ we have the following Lie superalgebra, after
    rescaling $H$ to set $c_0 = 1$:
    \begin{equation}
      [\sB(\beta),\sQ(s)] = \tfrac12 \sQ(\beta s) \qquad\text{and}\qquad
      [\sQ(s), \sQ(s)] = |s|^2 H.
    \end{equation}
  \item On the other hand, if $\hh \neq 0$, we may rotate it so that
    $2h = \psi \kk$ for some $\psi$ such that $\psi c_0 > 0$.  Then we
    may rescale $H$ and $\sQ$ in such that a way that we bring
    $\psi c_0 =1$, thus arriving at the following Lie superalgebra:
    \begin{equation}
      [\sB(\beta),\sQ(s)] = \tfrac12 \sQ(\beta s), \qquad [H,\sQ(s)] =
      \tfrac12 \sQ(s\kk) \qquad\text{and}\qquad
      [\sQ(s), \sQ(s)] = |s|^2 H - \sB(s\kk\sbar).
    \end{equation}
  \end{enumerate}
\end{enumerate}

\subsubsection{Lie superalgebras associated with kinematical Lie algebra $\mathsf{K13}$}
\label{sec:lie-super-assoc-13}

Here Lemma~\ref{lem:kmod} says that $\bb^2= \frac12 \bb$, so that $\bb \in
\RR$ and $\pp^2=-\tfrac12(\bb-\frac12) \in \RR$.  Lemma~\ref{lem:qqq} does
not simplify further at this stage.  The $[H,\Q,\Q]$ Jacobi identity
says that $c_0 \Re \hh = 0$ and $\hh\bc_i + \bc_i \bar \hh = 0$ for
$i=1,2,3$.  The $[\B,\Q,\Q]$ Jacobi identity says that $\bb\bc_1 =
\bb\bc_3 = 0$, whereas $(\bb-\frac12)\bc_2 = \frac12 \bc_1$.  Finally, the
$[\P,\Q,\Q]$ Jacobi identity says that $\bc_1 = 2\pp \bc_3$, $\bc_3 = -2
\pp \bc_2$ and $\bc_3 = 2 \pp \bc_1$.

As usual we have two branches depending on the value of $\bb$:
\begin{enumerate}
\item If $\bb=0$, then $\pp^2= \frac14$.  Due to the automorphism of $\k$ which
  changes the sign of $\P$, we may assume $\pp = \frac12$ without loss
  of generality.  It follows that $\bc_1 = c_0 \hh$ and that $\bc_2 = -
  \bc_1 = - c_0 \hh$ and that $\bc_3 = \bc_1 = c_0 \hh$.  If $c_0 = 0$
  then $\bc_i = 0$ for all $i$, so we must have $c_0 \neq 0$.  In that
  case, $\hh \in \Im \HH$ and $\hh$ is collinear with all $\bc_i$ for
  $i=1,2,3$.  We distinguish two cases, depending on whether or not
  $\hh=0$:
  \begin{enumerate}
  \item If $\hh\neq 0$, we may rotate it so that $\hh = \psi \kk$ where
    $\psi c_0 > 0$.  We may rescale $H \mapsto \psi^{-1} H$
    (which is an automorphism of $\k$) and rescale $\Q$ to bring
    $\psi c_0 = 1$.  In summary, we arrive at the following Lie
    superalgebra:
    \begin{equation}
      [H,\sQ(s)] = \sQ(s\kk),\qquad [\sP(\pi), \sQ(s)] = \tfrac12
      \sQ(\pi s) \qquad\text{and}\qquad [\sQ(s),\sQ(s)] = |s|^2 H -
      \sJ(s\kk\sbar) + \sB(s\kk\sbar) - \sP(s\kk\sbar).
    \end{equation}
  \item If $\hh = 0$, then we have the Lie superalgebra
    \begin{equation}
      [\sP(\pi), \sQ(s)] = \tfrac12
      \sQ(\pi s) \qquad\text{and}\qquad [\sQ(s),\sQ(s)] = |s|^2 H.
    \end{equation}
  \end{enumerate}
  
\item If $\bb=\frac12$, then $\pp=0$ and $\bc_1 = \bc_3 = 0$ with $\bc_2 =
  2 c_0 \hh$ with $c_0 \neq 0$ and $\hh \in \Im\HH$.  Again we distinguish
  between vanishing and nonvanishing $\hh$:
  \begin{enumerate}
  \item If $\hh \neq 0$, we may rotate it so that $2h = \psi \kk$
    with $\psi c_0 > 0$.  We apply the $\k$-automorphism $H \mapsto
    \psi^{-1} H$ and rescale $\Q$ to bring $\psi c_0 = 1$, thus
    resulting in the Lie superalgebra
    \begin{equation}
      [H,\sQ(s)] = \tfrac12 \sQ(s\kk),\qquad [\sB(\beta), \sQ(s)] = \tfrac12
      \sQ(\beta s) \qquad\text{and}\qquad [\sQ(s),\sQ(s)] = |s|^2 H -
      \sB(s\kk\sbar).      
    \end{equation}
  \item If $\hh = 0$, we arrive at the Lie superalgebra
    \begin{equation}
      [\sB(\beta), \sQ(s)] = \tfrac12
      \sQ(\beta s) \qquad\text{and}\qquad [\sQ(s),\sQ(s)] = |s|^2 H.
    \end{equation}
  \end{enumerate}
\end{enumerate}

\subsubsection{Lie superalgebras associated with kinematical Lie algebra $\mathsf{K14}$}
\label{sec:lie-super-assoc-14}

Here Lemma~\ref{lem:kmod} says that $\pp=0$ and $2 \bb^2 = \bb$, so that $\bb
\in \RR$.  Lemma~\ref{lem:qqq} says that $\tfrac12 \bc_1 + \bc_2 \bb = c_0 \hh$.  The
$[\P,\Q,\Q]$ Jacobi identity says that $\bc_1=0$, so that $c_0 \hh =
\bc_2 \bb$.  The $[\B,\Q,\Q]$ Jacobi identity says that $(2\bb-1) \bc_2 = 0$ and
$\bb\bc_3=0$, whereas the $[\B,\Q,\Q]$ Jacobi identity says that $\hh
\bc_i + \bc_i \bar \hh = 0$ for $i=2,3$.

We have two branches, depending on the value of $\bb$:
\begin{enumerate}
\item If $\bb=0$ then $\bc_2 = 0$ and we have two sub-branches depending
  on whether or not $c_0 = 0$:
  \begin{enumerate}
  \item If $c_0 = 0$ then $\bc_3 \neq 0$, so that $\Re \hh = 0$ and $\hh$
    is collinear with $\bc_3$.  We may rotate $\bc_3$ to lie along
    $\kk$, say, and then use automorphisms of $\k$ to set $\bc_3 =
    \kk$.  If $\hh \neq 0$, we may also set it equal to $\kk$.  In
    summary, we have two isomorphism classes of Lie superalgebras
    here:
    \begin{equation}
      [H,\sQ(s)] =
      \begin{cases}
        0 \\
        \sQ(s\kk)
      \end{cases}
      \qquad\text{and}\qquad
      [\sQ(s),\sQ(s)] = - \sP(s\kk\sbar).
    \end{equation}
  \item If $c_0 \neq 0$, then $\hh=0$ and $\bc_3$ is free: if nonzero we
    may rotate it to $\kk$ and rescaling $\P$, which is an
    automorphism of $\k$, we can bring it to $\kk$.  Rescaling $H$ we
    can bring $c_0 = 1$.  This gives two isomorphism classes of Lie
    superalgebras:
    \begin{equation}
      [\sQ(s),\sQ(s)] = |s|^2 H \qquad\text{and}\qquad
      [\sQ(s),\sQ(s)] = |s|^2 H - \sP(s\kk\sbar).
    \end{equation}
  \end{enumerate}
\item If $\bb=\tfrac12$, then $\bc_3 = 0$ and $\bc_2 = 2 c_0 \hh$, and we
  have two cases, depending on whether or not $\hh=0$.
  \begin{enumerate}
  \item If $\hh=0$ then $\bc_2 = 0$, and then $c_0 \neq 0$.  Rescaling
    $H$ we can set $c_0=1$ and we arrive at the Lie superalgebra
    \begin{equation}
      [\sB(\beta),\sQ(s)] = \tfrac12 \sQ(\beta s)
      \qquad\text{and}\qquad
      [\sQ(s),\sQ(s)] = |s|^2 H.
    \end{equation}
  \item If $\hh\neq 0$ we can rotate and rescale $\Q$ such that $\bc_2 =
    2 c_0 \hh = \kk$ and then we can rescale $H$ so that $c_0 = 1$.
    The resulting Lie superalgebra is now
    \begin{equation}
      [H,\sQ(s)] = \tfrac12 \sQ(s\kk),\qquad [\sB(\beta),\sQ(s)] =
      \tfrac12 \sQ(\beta s) \qquad\text{and}\qquad [\sQ(s),\sQ(s)] =
      |s|^2 H - \sB(s\kk\sbar).
    \end{equation}
  \end{enumerate}
\end{enumerate}

\subsubsection{Lie superalgebras associated with kinematical Lie algebra $\mathsf{K15}$}
\label{sec:lie-super-assoc-15}

Here Lemma~\ref{lem:kmod} says that $\bb = \pp = 0$, whereas
Lemma~\ref{lem:qqq} says that $\bc_1 = 2 c_0 \hh$.  The $[\B,\Q,\Q]$
Jacobi identity says that $\bc_1 = \bc_2 = 0$, and hence the
$[\P,\Q,\Q]$ component is identically satisfied.  Finally, the
$[H,\Q,\Q]$ Jacobi identity says that $\hh \bc_3 + \bc_3 \bar \hh = 0$,
which expands to
\begin{equation}
  2 \Re(\hh) \bc_3 + [\Im \hh, \bc_3] = 0.
\end{equation}

We have two branches of solutions:
\begin{enumerate}
\item If $c_0 = 0$, then $\bc_3 \neq 0$ and hence $\Re \hh = 0$ and $\hh$
  is collinear with $\bc_3$.  We may rotate $\bc_3$ to lie along $\kk$
  and then rescale $\Q$ so that $\bc_3 = \kk$.  If $\hh \neq 0$, we may
  use automorphisms of $\k$ to set $\hh = \kk$ as well.  In summary, we
  have two isomorphism classes of Lie superalgebras:
  \begin{equation}
    [H, \sQ(s)] =
    \begin{cases}
      \sQ(s\kk)\\
      0
    \end{cases} \qquad\text{and}\qquad
    [\sQ(s), \sQ(s)] = - \sP(s\kk\sbar).
  \end{equation}
\item If $c_0 \neq 0$, then $\hh = 0$ and $\bc_3$ is unconstrained.  If
  nonzero, we may rotate it to lie along $\kk$, rescale $\Q$ so that
  $\bc_3 = \kk$ and then use automorphisms of $\k$ to set $c_0 = 1$.
  In summary, we have two isomorphism classes of Lie superalgebras:
  \begin{equation}
    [\sQ(s), \sQ(s)] = |s|^2 H \qquad\text{or}\qquad
    [\sQ(s), \sQ(s)] = |s|^2 H - \sP(s\kk\sbar).
  \end{equation}
\end{enumerate}

\subsubsection{Lie superalgebras associated with kinematical Lie algebra $\mathsf{K16}$}
\label{sec:lie-super-assoc-16}

Here Lemma~\ref{lem:kmod} says that $\pp=0$ and $\bb(\bb-\frac12) =0$, so
that $\bb \in \RR$.  Lemma~\ref{lem:qqq} then says that $c_0 \hh = \frac12
\bc_1 + \bc_2 \bb$.  Now the $[\P,\Q,\Q]$ Jacobi identity says that
$c_0=0$ and $\bc_1= 0$, so that $\bc_2 \bb = 0$.  The $[H,\Q,\Q]$ Jacobi
identity says that $\hh\bc_2 + \bc_2 \bar \hh = 0$ and $\hh\bc_3 + \bc_3
\bar \hh = \bc_3$.  Finally the $[\B,\Q,\Q]$ Jacobi identity says that
$\bb\bc_3 = 0$ and $(\bb-\frac12)\bc_2= 0$.

Notice that if $\bb=\frac12$ then $\bc_3=0$ and $\bc_2 = 0$,
contradicting $[\Q,\Q]\neq 0$, so we must have $\bb=0$.  Now $\bc_2=0$
and hence $\bc_3 \neq 0$.  It then follows that $\Re \hh = \frac12$ and
$\Im \hh$ is collinear with $\bc_3$.  We can rescale $\P$ (which is an
automorphism of $\k$) and rotate so that $\bc_3 = \kk$, so that $\hh =
\frac12 (1 + \lambda \kk)$ for $\lambda \in \RR$.  The resulting
one-parameter family of Lie superalgebras is then
\begin{equation}
  [H,\sQ(s)] = \tfrac12 \sQ(s(1+\lambda \kk)) \qquad\text{and}\qquad
  [\sQ(s), \sQ(s)] = - \sP(s\kk\sbar).
\end{equation}

As in the case of the Lie superalgebras associated with Lie algebras
\hyperlink{KLA3}{$\mathsf{K3}_\gamma$} and \hyperlink{KLA5}{$\mathsf{K5}$}, the
parameter $\lambda$ is essential and Lie superalgebras with different
values of $\lambda$ are not isomorphic.

\subsubsection{Lie superalgebras associated with kinematical Lie algebra $\mathsf{K17}$}
\label{sec:lie-super-assoc-17}

Here Lemma~\ref{lem:kmod} simply sets $\bb = \pp = 0$ and
Lemma~\ref{lem:qqq} says $\bc_1 = 2 c_0 \hh$.  The $[\P,\Q,\Q]$ Jacobi
identity sets $\bc_1 = 0$ and hence $c_0 \hh = 0$.  The $[\B,\Q,\Q]$
Jacobi identity sets $c_0 = 0$ and $\bc_2 = 0$, whereas the
$[H,\Q,\Q]$ Jacobi identity says that $\hh$ is collinear with $\bc_3
\neq 0$.  We can rotate $\bc_3$ to lie along $\kk$ and rescale $\sQ$
to effectively set it to $\kk$.  Then $\hh = \frac\psi2 \kk$ for some
$\psi$ and rescaling $H$ allows us to set $\psi =1$.  In
summary, we have a unique Lie superalgebra associated with this
kinematical Lie algebra: namely,
\begin{equation}
  [H,\sQ(s)] = \tfrac12 \sQ(s\kk) \qquad\text{and}\qquad
  [\sQ(s), \sQ(s)] = - \sP(s\kk\sbar).
\end{equation}

\subsubsection{Lie superalgebras associated with kinematical Lie algebra $\mathsf{K18}$}
\label{sec:lie-super-assoc-18}

Here Lemma~\ref{lem:kmod} simply sets $\bb = \pp = 0$ and
Lemma~\ref{lem:qqq} says $\bc_1 = 2 c_0 \hh$.  The $[\P,\Q,\Q]$ Jacobi
identity sets $\bc_1 = 0$ and $c_0 = 0$, whereas the $[\B,\Q,\Q]$
Jacobi identity sets $\bc_2 = 0$.  Finally, the $[H,\Q,\Q]$ Jacobi
identity says that $\Re \hh = 1$ and $\Im \hh = \lambda \bc_3$ for some
$\lambda \in \RR$.  We can rotate $\bc_3$ to lie along $\kk$ and rescale $\sQ$
to effectively set it to $\kk$.  Then $\hh = 1 + \lambda \kk$. In
summary, we have a one-parameter family of Lie superalgebras
associated with this kinematical Lie algebra: namely,
\begin{equation}
  [H,\sQ(s)] = \sQ(s(1+\lambda \kk)) \qquad\text{and}\qquad
  [\sQ(s), \sQ(s)] = - \sP(s\kk\sbar).
\end{equation}

As in the case of the Lie superalgebras associated with Lie algebras
\hyperlink{KLA3}{$\mathsf{K3}_\gamma$}, \hyperlink{KLA5}{$\mathsf{K5}$} and
\hyperlink{KLA16}{$\mathsf{K16}$}, the parameter $\lambda$ is essential and
Lie superalgebras with different values of $\lambda$ are not
isomorphic.

\subsubsection{Summary}
\label{sec:summary}

Table~\ref{tab:klsa} summarises the results.  In that table we list
the isomorphism classes of kinematical Lie superalgebras (with
$[\Q,\Q]\neq 0$).  Recall that the Lie brackets involving $\Q$ are the
$[\Q,\Q]$ bracket and also
\begin{equation}
  [H, \sQ(s)] = \sQ(s\hh), \qquad [\sB(\beta),\sQ(s)] = \sQ(\beta s \bb),
  \qquad [\sP(\pi), \sQ(s) ] = \sQ(\pi s \pp),
\end{equation}
for some $\hh,\bb,\pp \in \HH$.  In Table~\ref{tab:klsa} we list any nonzero
values of $\hh,\bb,\pp$ and the $[\Q,\Q]$ bracket.  The first column is
simply the label for the Lie superalgebra, the second column is the
corresponding kinematical Lie algebra, the next columns are
$\hh,\bb,\pp$ and $[\Q,\Q]$.  The next four columns are the possible
$\so(3)$-equivariant $\ZZ$-gradings (with $\J$ of degree $0$)
compatible with the $\ZZ_2$-grading; that is, such that the parity is
the reduction modulo $2$ of the degree.  This requires, in particular,
that $q$ be an odd integer, which we can take to be $-1$ by
convention, if so desired.

\begin{table}[h!]
  \centering
  \caption{Kinematical Lie superalgebras (with $[\Q,\Q]\neq 0$)}
  \label{tab:klsa}
  \setlength{\extrarowheight}{2pt}
  \rowcolors{2}{blue!10}{white}
  \resizebox{\textwidth}{!}{
    \begin{tabular}{l|l*{4}{|>{$}c<{$}}*{4}{|>{$}c<{$}}}\toprule
      \multicolumn{1}{c|}{S\#} & \multicolumn{1}{c|}{$\k$} & \multicolumn{1}{c|}{$\hh$}& \multicolumn{1}{c|}{$\bb$} & \multicolumn{1}{c|}{$\pp$} & \multicolumn{1}{c|}{$[\sQ(s),\sQ(s)]$} & w_ H & w_{\B} & w_{\P} & w_{\Q}\\
      \toprule
      \hypertarget{KLSA1}{1} & \hyperlink{KLA1}{$\mathsf{K1}$}& \tfrac12 \kk & & & -\sP(s\kk\sbar) & 0 & 2m & 2q & q \\
      \hypertarget{KLSA2}{2} & \hyperlink{KLA1}{$\mathsf{K1}$} & & & & |s|^2 H - \sB(s\jj\sbar) - \sP(s\kk\sbar)& 2q & 2q & 2q & q \\
      \hypertarget{KLSA3}{3} & \hyperlink{KLA1}{$\mathsf{K1}$} & & & & |s|^2 H - \sP(s\kk\sbar) & 2q & 2m & 2 q & q \\
      \hypertarget{KLSA4}{4} & \hyperlink{KLA1}{$\mathsf{K1}$} & & & & |s|^2 H & 2q & 2m & 2p & q \\
      \hypertarget{KLSA5}{5} & \hyperlink{KLA1}{$\mathsf{K1}$} & & & & - \sB(s\jj\sbar) - \sP(s\kk\sbar) & 2n & 2q & 2 q & q \\
      \hypertarget{KLSA6}{6} & \hyperlink{KLA1}{$\mathsf{K1}$} & & & & -\sP(s\kk\sbar) & 2n & 2m & 2q & q \\
      \hypertarget{KLSA7}{7} & \hyperlink{KLA2}{$\mathsf{K2}$} & \kk & & & -\sP(s\kk\sbar) & 0 & 2q & 2 q & q \\
      \hypertarget{KLSA8}{8} & \hyperlink{KLA2}{$\mathsf{K2}$} & & & & -\sP(s\kk\sbar) & 2n & 2(q-n) & 2q & q \\
      \hypertarget{KLSA9}{9$_{\gamma\in[-1,1],\lambda\in\RR}$} & \hyperlink{KLA3}{$\mathsf{K3}_\gamma$} & \tfrac12 (1 + \lambda \kk) & & & -\sP(s\kk\sbar) & 0 & 2m & 2q & q \\
      \hypertarget{KLSA10}{10$_{\gamma\in[-1,1),\lambda\in\RR}$} & \hyperlink{KLA3}{$\mathsf{K3}_\gamma$} & \tfrac12 (\gamma + \lambda \kk) & & & -\sB(s\kk\sbar) & 0 & 2q & 2p & q \\
      \hypertarget{KLSA11}{11$_{\chi\geq0}$} & \hyperlink{KLA4}{$\mathsf{K4}_\chi$} & \tfrac12 (\chi + \jj) & & & -\sB(s\ii\sbar) - \sP(s\kk\sbar) & 0 & 2q & 2q & q \\
      \hypertarget{KLSA12}{12$_{\lambda\in\RR}$} & \hyperlink{KLA5}{$\mathsf{K5}$} & \tfrac12 (1 + \lambda \kk) & & & -\sP(s\kk\sbar) & 0 & 2q & 2q & q \\
      \hypertarget{KLSA13}{13} & \hyperlink{KLA6}{$\mathsf{K6}$} & & & & |s|^2 H & 2q & 2m & 2(q-m) & q \\
      \hypertarget{KLSA14}{14} & \hyperlink{KLA8}{$\mathsf{K8}$} & & \tfrac12 \kk & & |s|^2 H - \sP(s\kk\sbar) & 2q & 0 & 2q & q \\
      \hypertarget{KLSA15}{15} & \hyperlink{KLA11}{$\mathsf{K11}$} & \tfrac12 \kk & \tfrac12 \ii & \tfrac12 \jj & |s|^2 H + \sJ(s\kk\sbar) + \sB(s\jj\sbar) - \sP(s\ii\sbar) & - & - & - & - \\
      \hypertarget{KLSA16}{16} & \hyperlink{KLA12}{$\mathsf{K12}$} & & & \tfrac12 \jj & \sJ(s\ii\sbar) - \sB(s\ii\sbar) + \sP(s\kk\sbar) & - & - & - & - \\
      \hypertarget{KLSA17}{17} & \hyperlink{KLA12}{$\mathsf{K12}$} & & \tfrac12 & & |s|^2 H & 2q & 0 & 0 & q \\
      \hypertarget{KLSA18}{18} & \hyperlink{KLA12}{$\mathsf{K12}$} & \tfrac12 \kk & \tfrac12 & & |s|^2 H - \sB(s \kk \sbar) & - & - & - & - \\
      \hypertarget{KLSA19}{19} & \hyperlink{KLA13}{$\mathsf{K13}$} & \kk & & \tfrac12 & |s|^2 H - \sJ(s\kk\sbar) + \sB(s\kk\sbar) - \sP(s\kk\sbar) & - & - & - & - \\
      \hypertarget{KLSA20}{20} & \hyperlink{KLA13}{$\mathsf{K13}$} & & & \tfrac12 & |s|^2 H & 2q & 0 & 0 & q \\
      \hypertarget{KLSA21}{21} & \hyperlink{KLA13}{$\mathsf{K13}$} & & \tfrac12 & & |s|^2 H &  2q & 0 & 0 & q \\
      \hypertarget{KLSA22}{22} & \hyperlink{KLA13}{$\mathsf{K13}$} & \tfrac12 \kk & \tfrac12 & & |s|^2 H - \sB(s\kk\sbar) & - & - & - & - \\
      \hypertarget{KLSA23}{23} & \hyperlink{KLA14}{$\mathsf{K14}$} & \kk & & & - \sP(s\kk\sbar) & 0 & 0 & 2q & q \\
      \hypertarget{KLSA24}{24} & \hyperlink{KLA14}{$\mathsf{K14}$} & & & & - \sP(s\kk\sbar) & 2n & 0 & 2q & q \\
      \hypertarget{KLSA25}{25} & \hyperlink{KLA14}{$\mathsf{K14}$} & & & & |s|^2 H & 2q & 0 & 2p & q \\
      \hypertarget{KLSA26}{26} & \hyperlink{KLA14}{$\mathsf{K14}$} & & & & |s|^2 H - \sP(s\kk\sbar) & 2q & 0 & 2q & q \\
      \hypertarget{KLSA27}{27} & \hyperlink{KLA14}{$\mathsf{K14}$} & & \tfrac12 & & |s|^2 H & 2q & 0 & 2p & q \\
      \hypertarget{KLSA28}{28} & \hyperlink{KLA14}{$\mathsf{K14}$} & \tfrac12 \kk & \tfrac12 & & |s|^2 H - \sB(s\kk\sbar) & - & - & - & - \\
      \hypertarget{KLSA29}{29} & \hyperlink{KLA15}{$\mathsf{K15}$} & \kk & & & -\sP(s\kk\sbar) & - & - & - & - \\
      \hypertarget{KLSA30}{30} & \hyperlink{KLA15}{$\mathsf{K15}$} & & & & -\sP(s\kk\sbar) & - & - & - & - \\
      \hypertarget{KLSA31}{31} & \hyperlink{KLA15}{$\mathsf{K15}$} & & & & |s|^2 H & 2q & 2m & 4m & q \\
      \hypertarget{KLSA32}{32} & \hyperlink{KLA15}{$\mathsf{K15}$} & & & & |s|^2 H -\sP(s\kk\sbar) & - & - & - & - \\
      \hypertarget{KLSA33}{33$_{\lambda\in\RR}$} & \hyperlink{KLA16}{$\mathsf{K16}$} & \tfrac12 (1 + \lambda \kk) & & & -\sP(s\kk\sbar) & 0 & 0 & 2q & q \\
      \hypertarget{KLSA34}{34} & \hyperlink{KLA17}{$\mathsf{K17}$} & \tfrac12 \kk & & & -\sP(s\kk\sbar) & - & - & - & - \\
      \hypertarget{KLSA35}{35$_{\lambda\in\RR}$} & \hyperlink{KLA18}{$\mathsf{K18}$} & 1 + \lambda \kk & & & -\sP(s\kk\sbar) & - & - & - & - \\
      \bottomrule
    \end{tabular}
  }
  \caption*{The first column is our identifier for $\s$, whereas the
    second column is the kinematical Lie algebra $\k = \s_{\bar 0}$ in
    Table~\ref{tab:kla}.  The next four columns specify the brackets
    of $\s$ not of the form $[\J,-]$.  Supercharges $\sQ(s)$ are
    parametrised by $s \in \HH$, whereas $\sJ(\omega)$, $\sB(\beta)$
    and $\sP(\pi)$ are parametrised by $\omega,\beta,\pi \in
    \Im\HH$. The brackets are given by $[H,\sQ(s)] = \sQ(s\hh)$,
    $[\sB(\beta),\sQ(s)]=\sQ(\beta s \bb)$ and
    $[\sP(\pi),\sQ(s)] = \sQ(\pi s \pp)$, for some
    $\hh,\bb,\pp\in\HH$.  (This formalism is explained in
    Section~\ref{sec:quat-form}.)  The final four columns specify
    compatible gradings of $\s$, with $m,n,p,q \in \ZZ$ and $q$ odd.}
\end{table}

\subsection{Classification of aristotelian Lie superalgebras}
\label{sec:class-arist-lie}

Table~\ref{tab:ALAs} lists the aristotelian Lie algebras (with
three-dimensional space isotropy), classified in
\cite[App.~A]{Figueroa-OFarrill:2018ilb}.  In this section, we classify
the $N=1$ supersymmetric extensions of the aristotelian Lie algebras
(with $[\Q,\Q] \neq 0$).

\begin{table}[h!]
  \centering
  \caption{Aristotelian Lie algebras and their spacetimes}
  \label{tab:ALAs}
  \rowcolors{2}{blue!10}{white}
  \begin{tabular}{l|*{2}{>{$}l<{$}}|l}\toprule
    \multicolumn{1}{c|}{A\#} & \multicolumn{2}{c|}{Nonzero Lie brackets}& \multicolumn{1}{c}{Spacetime}\\\midrule
    \hypertarget{ALA1}{1} & & & static \\
    \hypertarget{ALA2}{2} & [H,\P] = \P & & torsional static\\
    \hypertarget{ALA3p}{3$_+$} & & [\P,\P] =  \J & $\RR \times S^3$\\
    \hypertarget{ALA3m}{3$_-$} & & [\P,\P] = - \J & $\RR \times H^3$ \\
    \bottomrule
  \end{tabular}
\end{table}

\subsubsection{Lie superalgebras associated with aristotelian Lie
  algebra $\mathsf{A1}$}
\label{sec:lie-super-ALA1}

We start with the static aristotelian Lie algebra
\hyperlink{ALA1}{$\mathsf{A1}$}, whose only nonzero brackets are
$[\J,\J] = \J$ and $[\J,\P] = \P$.  Any supersymmetric extension $\g$
has possible brackets
\begin{equation}
  [H,\sQ(s)] = \sQ(s\hh), \qquad [\sP(\pi),\sQ(s)] = \sQ(\pi s \pp)
  \qquad\text{and}\qquad [\sQ(s), \sQ(s)] = c_0 |s|^2 H - \sJ(s \bc_1
  \sbar) - \sP(s \bc_3 \sbar),
\end{equation}
for some $\hh,\pp \in \HH$, $c_0 \in \RR$ and $\bc_1, \bc_3 \in \Im
\HH$, using the same notation as in
Section~\ref{sec:kinem-lie-super}.  We can reuse
Lemmas~\ref{lem:kmod} and \ref{lem:qqq}, by setting $\bb = 0$ and
$\bc_2 = 0$ and ignoring $\B$.  Doing so we find that $\pp = 0$ and
that $\bc_1 = 2 c_0 \hh$.  The $[H,\Q,\Q]$ component of the Jacobi
identity gives $c_0 \Re\hh = 0$ (which already follows from
Lemma~\ref{lem:qqq}), $\bc_1 \hhbar + \hh \bc_1 = 0$ and $\bc_3 \hhbar
+ \hh \bc_3 = 0$.  The $[\P,\Q,\Q]$ component of the Jacobi identity
says that $[s\bc_1 \sbar, \pi] = 0$ for all $\pi \in \Im\HH$ and $s
\in \HH$, which says $\bc_1 =0$ and hence $c_0 \hh = 0$.  This gives
rise to two branches:
\begin{enumerate}
\item If $c_0 = 0$, then $\bc_3 \neq 0$ and the condition $\bc_3 \hhbar
  + \hh \bc_3 = 0$ is equivalent to $[\Im\hh,\bc_3] = - 2 \bc_3 \Re\hh$,
  which says $\Re\hh = 0$ and hence that $\hh$ and $\bc_3$ are
  collinear.  We can change basis so that $\bc_3 = \kk$ and $\hh =
  \kk$ if nonzero.  This leaves two possible Lie superalgebras
  depending on whether or not $\hh = 0$:
  \begin{equation}
    [H, \sQ(s)] =
    \begin{cases}
      \sQ(s\kk)\\
      0
    \end{cases}
\qquad\text{and}\qquad [\sQ(s), \sQ(s)] =
    - \sP(s\kk\sbar).
  \end{equation}
\item If $c_0 \neq 0$, then $\hh = 0$ and $\bc_3$ is free.  We can set
  $c_0 = 1$ and, if nonzero, we can also set $\bc_3 = \kk$.  This
  gives two possible Lie superalgebras:
  \begin{equation}
    [\sQ(s), \sQ(s)] =
    \begin{cases}
      |s|^2 H \\
      |s|^2 H - \sP(s\kk\sbar).
    \end{cases}
  \end{equation}
\end{enumerate}

\subsubsection{Lie superalgebras associated with aristotelian Lie
  algebra $\mathsf{A2}$}
\label{sec:lie-super-ALA2}

Let us now consider the aristotelian Lie algebra
\hyperlink{ALA2}{$\mathsf{A2}$}, with additional bracket
$[H,\P] = \P$.  Lemma~\ref{lem:kmod} again says $\pp = 0$ and
Lemma~\ref{lem:qqq} again says that $\bc_1 = 2 c_0 \hh$.  The
$[H,\Q,\Q]$ component of the Jacobi identity implies that
$c_0 \Re\hh = 0$ (which, again, is redundant),
$\bc_1 \hhbar + \hh \bc_1 = 0$ and $\bc_3 \hhbar + \hh \bc_3 = \bc_3$,
whereas the $[\P,\Q,\Q]$ component results in
$[s\bc_1 \sbar, \pi] = 2 c_0 |s|^2\pi$ for all $\pi \in \Im\HH$ and
$s \in \HH$.  This can only be the case if $c_0 = 0$ and hence
$\bc_1=0$, which then forces $\bc_3 \neq 0$.  The equation
$\bc_3 \hhbar + \hh \bc_3 = \bc_3$ results in
$[\Im\hh, \bc_3] = (1-2\Re\hh) \bc_3$, which implies
$\Re\hh = \tfrac12$ and $\Im\hh$ collinear with $\bc_3$.  We can
change basis so that $\bc_3 = \kk$ and we end up with a one-parameter
family of Lie superalgebras with brackets
\begin{equation}
  [H,\sQ(s)] = \sQ(\tfrac12 s (1 + \lambda \kk)), \qquad [\sQ(s),
  \sQ(s)] = - \sP(s\kk\sbar)
\end{equation}
for $\lambda \in \RR$, in addition to $[H,\sP(\pi)] = \sP(\pi)$.

\subsubsection{Lie superalgebras associated with aristotelian Lie
  algebras $\mathsf{A3}_{\pm}$}
\label{sec:lie-super-ALA3}

Finally, we consider the aristotelian Lie algebras \hyperlink{ALA3p}{$\mathsf{A3}_\pm$}
with bracket $[\P,\P] = \pm \J$.  Lemma~\ref{lem:kmod} says
that $[\hh,\pp] =0$ and $\pp^2 = \pm \tfrac14$, whereas
Lemma~\ref{lem:qqq} says that $c_0 \hh = \tfrac12\bc_1 + \bc_3 \pp$.
The $[H,\Q,\Q]$ Jacobi says $c_0\Re\hh = 0$, $\bc_1 \hhbar + \hh \bc_1
= 0$ and $\bc_3 \hhbar + \hh \bc_3 = 0$, whereas the $[\P,\Q,\Q]$
Jacobi gives the following relations:
\begin{equation}
  c_0 \Re(\sbar\pi s \pp) = 0, \qquad \pi s \pp \bc_3 \sbar - s \bc_3
  \ppbar \sbar \pi = \tfrac12 [\pi, s\bc_1
  \sbar]\qquad\text{and}\qquad
  \pi s \pp \bc_1 \sbar - s \bc_1
  \ppbar \sbar \pi = \pm \tfrac12 [\pi, s\bc_3 \sbar].
\end{equation}
We must distinguish two cases depending on the choice of signs.
\begin{enumerate}
\item Let's take the $+$ sign.  Then $\pp^2 = \tfrac14 \in \RR$.
  Without loss of generality we can take $\pp = \tfrac12$ by changing
  the sign of $\P$ if necessary.  Then the $[\P,\Q,\Q]$ Jacobi
  equations say that $\bc_1 = \bc_3$ and hence $c_0 \hh = \bc_1$.  If
  $c_0 = 0$, then $\bc_1 = \bc_3 = 0$, hence we take $c_0 \neq 0$ and
  thus $\Re \hh = 0$.  We can change basis so that $c_0 = 1$ and hence
  $\hh = \bc_1 = \bc_3$.  If nonzero, we can take them all equal to
  $\kk$.  In summary, we have two possible aristotelian Lie
  superalgebras extending \hyperlink{ALA3p}{$\mathsf{A3}_+$}, with
  brackets $[\sP(\pi), \sP(\pi')] = \tfrac12 \sJ([\pi,\pi'])$ and in
  addition either
  \begin{equation}
    [\sP(\pi), \sQ(s)] = \sQ(\tfrac12 \pi s), \qquad\text{and}\qquad
    [\sQ(s), \sQ(s) ] = |s|^2 H
  \end{equation}
  or
  \begin{equation}
    [H,\sQ(s)] = \sQ(s\kk), \qquad [\sP(\pi), \sQ(s)] = \sQ(\tfrac12
    \pi s) \qquad\text{and}\qquad
    [\sQ(s), \sQ(s) ] = |s|^2 H - \sJ(s\kk\sbar) - \sP(s\kk\sbar).
  \end{equation}
  
\item Let us now take the $-$ sign.  Here $\pp^2 = -\frac14$, so that
  $\pp \in \Im\HH$ (and $\pp \neq 0$) and hence $\Im\hh$ collinear
  with $\pp$.  The $[H,\Q,\Q]$ Jacobi equations force $\hh = 0$ and
  the $[\P,\Q,\Q]$ Jacobi equations force $c_0 = 0$ and
  $\bc_3 \pp = -\tfrac12 \bc_1$. 
  This means that $(\bc_1, 2\pp, \bc_3)$ is an oriented orthonormal
  frame for $\Im\HH$ and hence we can rotate them so that
  $(\bc_1, 2\pp, \bc_3) = (-\jj,\ii,\kk)$, for later uniformity.  This
  results in the aristotelian Lie superalgebra extending
  \hyperlink{ALA3m}{$\mathsf{A3}_-$} by the following brackets in
  addition to $[\sP(\pi), \sP(\pi')] = \tfrac12 \sJ([\pi,\pi'])$:
  \begin{equation}
    [\sP(\pi), \sQ(s)] = \sQ(\tfrac12 \pi s \ii) \qquad\text{and}\qquad
    [\sQ(s), \sQ(s)] = \sJ(s\jj\sbar) - \sP(s\kk\sbar).
  \end{equation}
\end{enumerate}

These results are summarised in Table~\ref{tab:alsa} below, together
with the possible compatible $\ZZ$-gradings.  This table also
classifies the homogeneous aristotelian superspaces.

\begin{table}[h!]
  \centering
  \caption{Aristotelian Lie superalgebras (with $[\Q,\Q]\neq 0$)}
  \label{tab:alsa}
  \setlength{\extrarowheight}{2pt}
  \rowcolors{2}{blue!10}{white}
  \begin{tabular}{l|l*{3}{|>{$}c<{$}}*{3}{|>{$}c<{$}}}\toprule
    \multicolumn{1}{c|}{S\#} & \multicolumn{1}{c|}{$\a$} & \multicolumn{1}{c|}{$\hh$}& \multicolumn{1}{c|}{$\pp$} & \multicolumn{1}{c|}{$[\sQ(s),\sQ(s)]$} & w_ H & w_{\P} & w_{\Q}\\
    \toprule
    \hypertarget{ALSA36}{36} &  \hyperlink{ALA1}{$\mathsf{A1}$} & \kk & & - \sP(s\kk\sbar) & 0 & 2q & q \\
    \hypertarget{ALSA37}{37} &  \hyperlink{ALA1}{$\mathsf{A1}$} & & & - \sP(s\kk\sbar) & 2n & 2q & q \\
    \hypertarget{ALSA38}{38} &  \hyperlink{ALA1}{$\mathsf{A1}$} & & & |s|^2 H & 2q & 2p & q \\
    \hypertarget{ALSA39}{39} &  \hyperlink{ALA1}{$\mathsf{A1}$} & & & |s|^2 H - \sP(s\kk\sbar) & 2q & 2q & q \\
    \hypertarget{ALSA40}{40$_{\lambda\in\RR}$} &  \hyperlink{ALA2}{$\mathsf{A2}$} & \tfrac12(1 + \lambda \kk) & & -\sP(s\kk\sbar) & 0 & 2q & q\\
    \hypertarget{ALSA41}{41} & \hyperlink{ALA3p}{$\mathsf{A3}_+$} & & \tfrac12 & |s|^2 H & 2q & 0 & q \\
    \hypertarget{ALSA42}{42} & \hyperlink{ALA3p}{$\mathsf{A3}_+$} & \kk & \tfrac12 & |s|^2 H - \sJ(s\kk\sbar) - \sP(s\kk\sbar) & - & - & - \\
    \hypertarget{ALSA43}{43} & \hyperlink{ALA3m}{$\mathsf{A3}_-$} & & \tfrac12\ii &  \sJ(s\jj\sbar) - \sP(s\kk\sbar) & - & - & - \\
    \bottomrule
  \end{tabular}
  \caption*{The first column is our identifier for $\s$, whereas the
    second column is the aristotelian Lie algebra $\a = \s_{\bar 0}$
    in Table~\ref{tab:ALAs}.  The next three columns specify the
    brackets of $\s$ not of the form $[\J,-]$.  Supercharges $\sQ(s)$
    are parametrised by $s \in \HH$, whereas $\sJ(\omega)$ and
    $\sP(\pi)$ are parametrised by $\omega,\pi \in \Im\HH$. The
    brackets are given by $[H,\sQ(s)] = \sQ(s\hh)$ and
    $[\sP(\pi),\sQ(s)] = \sQ(\pi s \pp)$, for some $\hh,\pp\in\HH$.
    (The formalism is explained in Section~\ref{sec:quat-form}.)
    The final three columns are compatible gradings of $\s$, with
    $n,p,q \in \ZZ$ and $q$ odd.}
\end{table}

\subsection{Unpacking the quaternionic notation}
\label{sec:unpack-quat-notat}

The quaternionic formalism we have employed in the classification of
kinematical and aristotelian Lie superalgebras, which has the virtue of
uniformity and ease in computation, does result in expressions which
are perhaps unfamiliar and which therefore might hinder comparison
with other formulations.  In this section, we will go through an
example illustrating how to unpack the notation.

The nonzero brackets of the Poincaré superalgebra
\hyperlink{KLSA14}{$\mathsf{S14}$} are given by
equation~\eqref{eq:klsa-brackets-quat} and
\begin{equation}
  [\sB(\beta), \sQ(s)] = \sQ(\tfrac12 \beta s \kk) \qquad\text{and}\qquad
  [\sQ(s),\sQ(s)] = |s|^2 H - \sP(s\kk\sbar),
\end{equation}
where
\begin{equation}
  \sB(\beta) = \sum_{i=1}^3 \beta_i B_i \qquad\text{and}\qquad \sQ(s) = \sum_{a=1}^4
  s_a Q_a,
\end{equation}
and where
\begin{equation}
  \beta = \beta_1 \ii + \beta_2 \jj + \beta_3 \kk
  \qquad\text{and}\qquad s = s_1 \ii + s_2 \jj + s_3 \kk + s_4.
\end{equation}
This allows us to simply unpack the brackets into the following
\begin{equation}
  [B_i, Q_a] = \tfrac12 \sum_{b=1}^4 Q_b \beta_i{}^b{}_a
  \qquad\text{and}\qquad
  [Q_a, Q_b] = \sum_{\mu=0}^3 P_\mu \gamma^\mu_{ab},
\end{equation}
where we have introduced $P_0 = H$ and where the matrices
$\boldsymbol{\beta}_i := [\beta_i{}^b{}_a]$ are given by
\begin{equation}
  \boldsymbol{\beta}_1 =
  \begin{pmatrix}
    \zero & -\id \\
    -\id & \zero
  \end{pmatrix},
  \qquad
  \boldsymbol{\beta}_2 =
  \begin{pmatrix}
    \zero & i\sigma_2\\
    -i\sigma_2 & \zero
  \end{pmatrix}
  \qquad\text{and}\qquad
  \boldsymbol{\beta}_3 =
  \begin{pmatrix}
    \id & \zero\\
    \zero & -\id
  \end{pmatrix}
\end{equation}
and where the symmetric matrices $\boldsymbol{\gamma}^\mu :=
[\gamma^\mu_{ab}]$ are given by
\begin{equation}
  \boldsymbol{\gamma}^0 =
  \begin{pmatrix}
    \id & \zero \\ \zero & \id
  \end{pmatrix},
  \qquad
  \boldsymbol{\gamma}^1 =
  \begin{pmatrix}
     \zero & \id \\ \id & \zero
  \end{pmatrix},
  \qquad
  \boldsymbol{\gamma}^2 =
  \begin{pmatrix}
    \zero & -i\sigma_2\\ i\sigma_2 & \zero
  \end{pmatrix}
  \qquad\text{and}\qquad
  \boldsymbol{\gamma}^3 =
  \begin{pmatrix}
    -\id & \zero \\ \zero & \id
  \end{pmatrix}.
\end{equation}

As shown in Section~\ref{sec:low-rank-invariants}, there is a
two-parameter family of symplectic forms on the
spinor representation $S$ which are invariant under the action of
$B_i$ and $J_i$.  They are given by
\begin{equation}
  \omega(s_1,s_2) := \Re(s_1 (\alpha \ii + \beta \jj) \sbar_2),
\end{equation}
for $\alpha,\beta \in \RR$ not both zero.  We may normalise $\omega$
such that $\alpha^2 + \beta^2 = 1$, resulting in a circle of
symplectic structures.  Relative to the standard real basis
$(\ii,\jj,\kk,1)$ for $\HH$, the matrix $\Omega$ of $\omega$ is given
by $\Omega = i \sigma_2 \otimes (-\alpha \sigma_1 + \beta \sigma_3)$,
whose inverse is $\Omega^{-1} = - \Omega$, due to the chosen
normalisation.  Let us define endomorphisms $\gamma^\mu$ of $S$ such
that $(\gamma^\mu)^a{}_b = (\Omega^{-1})^{ac} \gamma^\mu_{cb}$.
Explicitly, they are given by
\begin{equation}
  \begin{aligned}[m]
    \gamma^0 &= i\sigma_2 \otimes (\alpha \sigma_1 - \beta\sigma_3)\\
    \gamma^1 &= \sigma_3 \otimes (\alpha \sigma_1 - \beta\sigma_3)
  \end{aligned}
  \qquad\qquad
  \begin{aligned}[m]
    \gamma^2 &= -\mathbb{1} \otimes (\alpha \sigma_3 + \beta\sigma_1)\\
    \gamma^3 &= \sigma_1 \otimes (\alpha \sigma_1 - \beta\sigma_3).
  \end{aligned}
\end{equation}
It then follows that these endomorphisms represent the Clifford
algebra $\Cl(1,3)$:
\begin{equation}
  \gamma^\mu \gamma^\nu + \gamma^\nu \gamma^\mu = -2 \eta^{\mu\nu}
  \mathbb{1}.
\end{equation}
We thus arrive at the description of the Poincaré superalgebra
described in the appendix.

\subsection{Central extensions}
\label{sec:central-extensions}

In this section, we determine the possible central extensions of the
kinematical and aristotelian Lie superalgebras.

We start with the kinematical Lie superalgebras. Let
$\s = \s_{\bar 0} \oplus \s_{\bar 1}$ be one of the Lie superalgebras
in Table~\ref{tab:klsa}. By a \textbf{central extension} of $\s$, we
mean a short exact sequence of Lie superalgebras \begin{equation}
  \begin{tikzcd}
    0 \arrow[r] & \z \arrow[r] & \widehat\s \arrow[r] & \s \arrow[r] & 0,
  \end{tikzcd}
\end{equation}
where $\z$ is central in $\widehat\s$.  We may choose a vector space
splitting and view (as a vector space) $\widehat\s = \s \oplus \z$ and
the Lie bracket is given, for $(X,z), (Y,z') \in \s \oplus \z$, by
\begin{equation}
  [(X,z), (Y,z')]_{\widehat\s} = \left( [X,Y]_{\s}, \omega(X,Y) \right),
\end{equation}
where $\omega : \wedge^2\s \to \z$ is a cocycle.  (Here $\wedge$ is
taken in the super sense, so that it is symmetric on odd elements.)
Central extensions of $\s$ are classified up to isomorphism by the
Chevalley--Eilenberg cohomology group $H^2(\s)$, which by
Hochschild--Serre, can be computed from the subcomplex relative to the
rotational subalgebra $\r \subset \s_{\bar 0}$.  Indeed, we have the
isomorphism \cite{MR0054581}
\begin{equation}
  H^2(\s) \cong H^2(s,\r).
\end{equation}
Let $W = \spn{H,\B,\P,\Q}$.  Then the cochains in $C^2(s,\r)$ are
$\r$-equivariant maps $\wedge^2W \to \RR$ or, equivalently,
$\r$-invariant vectors in $\wedge^2 W^*$.  This is a two-dimensional
real vector space which, in quaternionic language, is given for $x,y \in \RR$ by
\begin{equation}
  \omega(\sB(\beta), \sP(\pi)) = x \Re(\beta \pi) = - \omega(\sP(\pi),
  \sB(\beta)) \qquad\text{and}\qquad \omega(\sQ(s_1), \sQ(s_2)) = y
  \Re(s_1\sbar_2).
\end{equation}
The cocycle conditions (i.e., the Jacobi identities of the central
extension $\widehat\s$) has several components.  Letting $\sV$ stand
for either $\sB$ or $\sP$, the cocycle conditions are given by
\begin{equation}
  \begin{split}
    \omega([H,\sV(\alpha)], \sV(\beta)) + \omega(\sV(\alpha), [H,\sV(\beta)]) &= 0,\\
    \omega([\sV(\alpha),\sV(\beta)], \sV(\gamma)) + \text{cyclic} &= 0,\\
    \omega([H,\sQ(s)],\sQ(s)) &= 0,\\
    2 \omega([\sV(\alpha),\sQ(s)], \sQ(s)) + \omega([\sQ(s),\sQ(s)],\sV(\alpha)) &= 0.
  \end{split}
\end{equation}
The first two of the above equations only involve the even
generators and hence depend only on the underlying kinematical Lie
algebra, whereas the last two equations do depend on the precise
superalgebra we are dealing with.  In the case of aristotelian Lie
superalgebras, there is no $\B$ and hence $\sV = \sP$ in the above
equations and, of course, the cocycle can only modify the $[\Q,\Q]$
bracket and hence the cocycle conditions are simply
\begin{equation}
  \omega([H,\sQ(s)],\sQ(s)) = 0 \qquad\text{and}\qquad
    \omega([\sP(\alpha),\sQ(s)], \sQ(s)) = 0.
\end{equation}

The calculations are routine, and we will not give any details, but
simply collect the results in Table~\ref{tab:central-ext}, where $Z$
is the basis for the one-dimensional central ideal $\z = \spn{Z}$, and
where we list only the brackets which are liable to change under
central extension.

\begin{table}[h!]
  \centering
  \caption{Central extensions of kinematical and aristotelian Lie superalgebras}
  \label{tab:central-ext}
  \setlength{\extrarowheight}{2pt}
  \rowcolors{2}{blue!10}{white}
  \begin{tabular}{l*{2}{|>{$}c<{$}}}\toprule
    \multicolumn{1}{c|}{S\#} & \multicolumn{1}{c|}{$[\sB(\beta),\sP(\pi)]$} & \multicolumn{1}{c}{$[\sQ(s),\sQ(s)]$} \\
    \toprule
    \hyperlink{KLSA1}{1} & & |s|^2 Z - \sP(s\kk\sbar) \\
    \hyperlink{KLSA4}{4} & -\Re(\beta\pi) Z & |s|^2 H \\
    \hyperlink{KLSA5}{5} & & |s|^2 Z - \sB(s\jj\sbar) - \sP(s\kk\sbar) \\
    \hyperlink{KLSA6}{6} & & |s|^2 Z -\sP(s\kk\sbar) \\
    \hyperlink{KLSA7}{7} & & |s|^2 Z -\sP(s\kk\sbar) \\
    \hyperlink{KLSA8}{8} & & |s|^2 Z -\sP(s\kk\sbar) \\
    \hyperlink{KLSA10}{10$_{\gamma=0,\lambda\in\RR}$} & & |s|^2 Z -\sB(s\kk\sbar) \\
    \hyperlink{KLSA11}{11$_{\chi=0}$} & & |s|^2 Z - \sB(s\ii\sbar) - \sP(s\kk\sbar) \\
    \hyperlink{KLSA13}{13} & -\Re(\beta\pi) (H + Z) & |s|^2 H \\
    \hyperlink{KLSA23}{23} & & |s|^2 Z - \sP(s\kk\sbar) \\
    \hyperlink{KLSA24}{24} & & |s|^2 Z - \sP(s\kk\sbar) \\
    \hyperlink{KLSA29}{29} & & |s|^2 Z -\sP(s\kk\sbar) \\
    \hyperlink{KLSA30}{30} & & |s|^2 Z -\sP(s\kk\sbar) \\
    \hyperlink{KLSA34}{34} & & |s|^2 Z -\sP(s\kk\sbar) \\
    \midrule
    \hyperlink{ALSA36}{36} & - & |s|^2 Z -\sP(s\kk\sbar) \\
    \hyperlink{ALSA37}{37} & - & |s|^2 Z -\sP(s\kk\sbar) \\
    \bottomrule
  \end{tabular}
  \caption*{The first column is our identifier for $\s$, whereas the
    other two columns are the possible central terms in the central
    extension $\widehat\s$.  Here $\beta,\pi \in\Im\HH$ and $s \in
    \HH$ are (some of) the parameters defining the Lie brackets in the
    quaternionic formalism explained in Section~\ref{sec:quat-form}.}
\end{table}

\subsection{Automorphisms of kinematical Lie superalgebras}
\label{sec:autom-kinem-lie}

In the next section, we will classify the homogeneous superspaces
associated to the kinematical Lie superalgebras.  As we will explain
below, the first stage is to classify ``super Lie pairs'' up to
isomorphism.  To that end, it behoves us to determine the group of
automorphisms of the Lie superalgebras in Table~\ref{tab:klsa}, to
which we now turn.

Without loss of generality, we can restrict to automorphisms which are
the identity when restricted to $\r$: we call them $\r$-fixing
automorphisms.  Following from our discussion in
Section~\ref{sec:automorphisms}, these are parametrised by triples
\begin{equation}
  \left(A := \begin{pmatrix}a & b\\ c & d\end{pmatrix}, \mu, 
  \qq\right) \in \GL(2,\RR) \times \RR^\times \times \HH^\times
\end{equation}
subject to the condition that the associated linear transformations
leave the Lie brackets in $\s$ unchanged.

It is easy to read off from equation~\eqref{eq:autk-on-params} what
$(A, \mu,\qq)$ must satisfy for the $\r$-equivariant linear
transformation $\Phi : \s \to \s$ defined by them to be an
automorphism of $\s$, namely:
\begin{equation}\label{eq:aut-s}
  \begin{aligned}[m]
    \hh\qq &= \mu \qq \hh \\
    \bb\qq &= \qq (a \bb + c \pp) \\
    \pp\qq &=  \qq (b \bb + d \pp) \\
    \mu c_0 &= |\qq|^2 c_0
  \end{aligned}
  \qquad\qquad
  \begin{aligned}[m]
    \qq\bc_1\qqbar &= \bc_1 \\
    \qq\bc_2\qqbar &= a \bc_2 + b \bc_3\\
    \qq\bc_3 \qqbar &= c \bc_2 + d \bc_3.\\
  \end{aligned}
\end{equation}
It is then a straightforward -- albeit lengthy -- process to go
through each Lie superalgebra in Table~\ref{tab:klsa} and solve
equations \eqref{eq:aut-s} for $(A, \mu,\qq)$. In particular, $(A,\mu)
\in \Aut_\r(\k)$ and they are given in Table~\ref{tab:aut-kla}.  The
results of this section are summarised in Tables~\ref{tab:aut-klsa}
and \ref{tab:aut-klsa-extra}, which list the $\r$-fixing automorphisms
for the Lie superalgebras \hyperlink{KLSA1}{$\mathsf{S1}$}-
\hyperlink{KLSA15}{$\mathsf{S15}$}
and \hyperlink{KLSA16}{$\mathsf{S16}$}-\hyperlink{KLSA35}{$\mathsf{S35}$}, 
respectively, in Table~\ref{tab:klsa}.

The first six Lie superalgebras in Table~\ref{tab:klsa} are
supersymmetric extensions of the static kinematical Lie algebra for
which $(A,\mu)$ can be any element in $\GL(2,\RR) \times \RR^\times$.

\subsubsection{Automorphisms of Lie superalgebra $\mathsf{S1}$}
\label{sec:autom-kinem-lie-1}

Here $\hh = \frac12\kk$, $\bb=\pp=0$, $c_0 = 0$, $\bc_1 = \bc_2 = 0$
and $\bc_3 = \kk$.  The invariance conditions~\eqref{eq:aut-s} give
\begin{equation}
  \mu \qq \kk = \kk\qq, \qquad b \kk = 0 \qquad\text{and}\qquad d \kk
  = \qq \kk \qqbar
\end{equation}
The second equation requires $b=0$.  The third equation says that the
real linear map $\alpha_\qq: \HH \to \HH$ defined by $\alpha_\qq (\xx)
= \qq \xx \qqbar$ preserves the $\kk$-axis in $\Im \HH$.

\begin{lemma}\label{lem:dk=qkqbar}
  Let $\qq\kk\qqbar = d \qq$ for some $d\in \RR$.  Then either $d =
  |\qq|^2$ and $\qq \in \spn{1,\kk}$ or $d = - |\qq|^2$ and $\qq \in
  \spn{\ii,\jj}$.
\end{lemma}

\begin{proof}
  Taking the quaternion norm of both sides of the equation $\qq \kk
  \qqbar = d \qq$ and using that $\qq \neq 0$, we see that $d =
  \pm |\qq|^2$ and hence right multiplying by $\qq$, the equation
  becomes $\pm \kk\qq = \qq \kk$.  If $\kk \qq = \qq\kk$,
  then $\qq \in \spn{1,\kk}$ and $d= |\qq|^2$, whereas if $-\kk\qq =
  \qq\kk$, then $\qq \in \spn{\ii,\jj}$ and $d = -|\qq|^2$.
\end{proof}

Taking the quaternion norm of the first equation, shows that $\mu =
\pm 1$ and hence that $d = \mu |\qq|^2$.  In summary, we have that the
typical automorphism $(A,\mu,\qq)$ takes one of two possible forms:
\begin{equation}
  \begin{split}
    &A = \begin{pmatrix} a & \zero \\ c & |\qq|^2
    \end{pmatrix},\qquad \mu= 1 \qquad\text{and}\qquad \qq = q_4 + q_3 \kk\\
   \text{or}\qquad &A =  \begin{pmatrix}
      a & \zero \\ c & - |\qq|^2
    \end{pmatrix}, \qquad \mu = - 1 \qquad\text{and}\qquad \qq = q_1 \ii
    + q_2 \jj.
  \end{split}
\end{equation}

\subsubsection{Automorphisms of Lie superalgebra $\mathsf{S2}$}
\label{sec:autom-kinem-lie-2}

Here $\hh = \bb = \pp = 0$, $c_0=1$, $\bc_1 = 0$, $\bc_2 = \jj$ and
$\bc_3 = \kk$.  The invariance conditions~\eqref{eq:aut-s} give
\begin{equation}
  \mu = |\qq|^2, \qquad a \jj + b \kk = \qq \jj \qqbar
  \qquad\text{and}\qquad c \jj + d \kk = \qq \kk \qqbar.
\end{equation}
The last two equations say that the real linear map $\alpha_\qq : \HH
\to \HH$ defined earlier preserves the $(\jj,\kk)$-plane in $\Im \HH$.

\begin{lemma}\label{lem:conj-plane}
  The map $\alpha_\qq : \HH \to \HH$ preserves the $(\jj,\kk)$-plane
  in $\Im \HH$ if and only if $\qq \in \spn{1,\ii} \cup
  \spn{\jj,\kk}$.
\end{lemma}

\begin{proof}
  Since $\qq \neq 0$, we can write it as $\qq = |\qq| \uu$, for some
  unique $\uu \in \Sp(1)$ and $\alpha_\qq = |\qq|^2 \alpha_\uu$.
  The map $\alpha_\qq$ preserves separately the real and imaginary
  subspaces of $\HH$ and $\alpha_\qq$ preserves the $(\jj,\kk)$-plane
  if and only if $\alpha_\uu$ does.  But for $\uu \in \Sp(1)$,
  $\alpha_\uu$ acts on $\Im \HH$ by rotations and hence if $\alpha_\uu$
  preserves $(\jj,\kk)$-plane, it also preserves the perpendicular
  line: the $\ii$-axis in this case and since it must preserve length,
  $\alpha_\uu (\ii) = \pm \ii$.  It follows that $\alpha_\qq(\ii) = \pm |\qq|^2
  \ii$, so that $\alpha_\qq$ too preserves the $\ii$-axis.  By an
  argument similar to that of Lemma~\ref{lem:dk=qkqbar} it follows
  that $\qq$ belongs either to the complex line in $\HH$
  generated by $\ii$ or to its perpendicular complement.
\end{proof}

From the Lemma we have two cases to consider: $\qq = q_4 + q_1 \ii$ or
$\qq = q_2 \jj + q_3 \kk$.  In each case we can use the last two
equations to solve for $a,b,c,d$ in terms of the components of $\qq$.
Summarising, we have that the typical automorphism $(A,\mu,\qq)$ takes
one of two possible forms:
\begin{equation}
  \begin{split}
    &A = \begin{pmatrix} q_4^2-q_1^2 & 2 q_1 q_4 \\ -2 q_1 q_4 & q_4^2 - q_1^2
    \end{pmatrix},\qquad \mu= q_1^2+q_4^2 \qquad\text{and}\qquad \qq = q_4 + q_1 \ii\\
   \text{or}\qquad &A =  \begin{pmatrix}
     q_2^2-q_3^2 & 2 q_2 q_3 \\ 2 q_2 q_3 & q_3^2 - q_2^2
    \end{pmatrix}, \qquad \mu = q_2^2+q_3^2 \qquad\text{and}\qquad \qq = q_2 \jj
    + q_3 \kk.
  \end{split}
\end{equation}

\subsubsection{Automorphisms of Lie superalgebra $\mathsf{S3}$}
\label{sec:autom-kinem-lie-3}

Here $\hh=\bb=\pp=0$, $c_0=1$, $\bc_1 = \bc_2 = 0$ and $\bc_3 = \kk$.
The invariance conditions~\eqref{eq:aut-s} give
\begin{equation}
  \mu = |\qq|^2, \qquad b\kk = 0 \qquad\text{and}\qquad d\kk = \qq \kk \qqbar.
\end{equation}
This is very similar to the case of the Lie superalgebra \hyperlink{KLSA1}{$\mathsf{S1}$} and, in
particular, Lemma~\ref{lem:dk=qkqbar} applies.  The typical
automorphism $(A,\mu,\qq)$ takes one of two possible forms:
\begin{equation}
  \begin{split}
    &A = \begin{pmatrix} a & \zero \\ c & |\qq|^2
    \end{pmatrix},\qquad \mu= |\qq|^2 \qquad\text{and}\qquad \qq = q_4 + q_3 \kk\\
   \text{or}\qquad &A =  \begin{pmatrix}
      a & \zero \\ c & - |\qq|^2
    \end{pmatrix}, \qquad \mu = |\qq|^2 \qquad\text{and}\qquad \qq = q_1 \ii
    + q_2 \jj.
  \end{split}
\end{equation}

\subsubsection{Automorphisms of Lie superalgebra $\mathsf{S4}$}
\label{sec:autom-kinem-lie-4}

Here $\hh=\bb=\pp=0$, $c_0=1$ and $\bc_1 = \bc_2 = \bc_3 = 0$.  The
only condition is $\mu = |\qq|^2$.  Hence the typical automorphism
$(A,\mu,\qq)$ takes the form
\begin{equation}
  A =
  \begin{pmatrix}
    a & b \\ c & d
  \end{pmatrix},
  \qquad \mu = |\qq|^2 \qquad\text{and}\qquad \qq \in \HH^\times.
\end{equation}

\subsubsection{Automorphisms of Lie superalgebra $\mathsf{S5}$}
\label{sec:autom-kinem-lie-5}

Here $\hh= \bb = \pp = 0$, $c_0 = 0$, $\bc_1 = 0$, $\bc_2 = \jj$ and
$\bc_3 = \kk$.  The invariance conditions~\eqref{eq:aut-s} are as for
Lie superalgebra \hyperlink{KLSA2}{$\mathsf{S2}$}, except that $\mu$ is
unconstrained.  In other words, the typical automorphism $(A,\mu,\qq)$
takes one of two possible forms:
\begin{equation}
  \begin{split}
    &A = \begin{pmatrix} q_4^2-q_1^2 & 2 q_1 q_4 \\ -2 q_1 q_4 & q_4^2 - q_1^2
    \end{pmatrix},\qquad \mu \qquad\text{and}\qquad \qq = q_4 + q_1 \ii\\
   \text{or}\qquad &A =  \begin{pmatrix}
     q_2^2-q_3^2 & 2 q_2 q_3 \\ 2 q_2 q_3 & q_3^2 - q_2^2
    \end{pmatrix}, \qquad \mu \qquad\text{and}\qquad \qq = q_2 \jj
    + q_3 \kk.
  \end{split}
\end{equation}

\subsubsection{Automorphisms of Lie superalgebra $\mathsf{S6}$}
\label{sec:autom-kinem-lie-6}

Here $\hh=\bb=\pp=0$, $c_0= 0$, $\bc_1 = \bc_2 = 0$ and $\bc_3= \kk$.
This is similar to Lie superalgebra \hyperlink{KLSA3}{$\mathsf{S3}$}, except
that $\mu$ remains unconstrained.  In summary, the typical
automorphisms $(A,\mu,\qq)$ takes one of two possible forms:
\begin{equation}
  \begin{split}
    &A = \begin{pmatrix} a & \zero \\ c & |\qq|^2
    \end{pmatrix},\qquad \mu \qquad\text{and}\qquad \qq = q_4 + q_3 \kk\\
   \text{or}\qquad &A =  \begin{pmatrix}
      a & \zero \\ c & - |\qq|^2
    \end{pmatrix}, \qquad \mu \qquad\text{and}\qquad \qq = q_1 \ii + q_2 \jj.
  \end{split}
\end{equation}

The next two Lie superalgebras (\hyperlink{KLSA7}{$\mathsf{S7}$} and
\hyperlink{KLSA8}{$\mathsf{S8}$}) are supersymmetric extensions of the
galilean Lie algebra, where $(A,\mu)$ take the form
\begin{equation}
  A =
  \begin{pmatrix}
    a & b \\ c & d 
  \end{pmatrix} \qquad\text{and}\qquad \mu = \frac{d}{a}.
\end{equation}

\subsubsection{Automorphisms of Lie superalgebra $\mathsf{S7}$}
\label{sec:autom-kinem-lie-7}

Here $\hh=\kk$, $\bb = \pp = 0$, $c_0 = 0$, $\bc_1 = \bc_2 = 0$ and
$\bc_3= \kk$.  The invariance conditions~\eqref{eq:aut-s} are
\begin{equation}
  d \qq\kk = a \kk \qq \qquad\text{and}\qquad d\kk = \qq \kk \qqbar.
\end{equation}
Multiplying the second equation on the right by $\qq$, using the
first equation and the fact that $\qq \neq 0$, results in $a =
d^2/|\qq|^2$, so that $a > 0$.  Taking the quaternion norm of the
first equation shows that $a = |d|$, so that $a = |\qq|^2$.  The first
equation now follows from the second, and that is solved by
Lemma~\ref{lem:dk=qkqbar}.

In summary the typical automorphism $(A,\mu,\qq)$ takes one of two
possible forms:
\begin{equation}
  \begin{split}
    &A = \begin{pmatrix} |\qq|^2 & \zero \\ c & |\qq|^2
    \end{pmatrix},\qquad \mu = 1 \qquad\text{and}\qquad \qq = q_4 + q_3 \kk\\
   \text{or}\qquad &A =  \begin{pmatrix}
      |\qq|^2 & \zero \\ c & - |\qq|^2
    \end{pmatrix}, \qquad \mu=-1 \qquad\text{and}\qquad \qq = q_1 \ii + q_2 \jj.
  \end{split}
\end{equation}

\subsubsection{Automorphisms of Lie superalgebra $\mathsf{S8}$}
\label{sec:autom-kinem-lie-8}

Here $\hh=\bb=\pp=0$, $c_0=0$, $\bc_1=\bc_2 =0$ and $\bc_3 = \kk$.
The invariance conditions~\eqref{eq:aut-s} reduce to just $d\kk =
\qq\kk\qqbar$, which we solve by Lemma~\ref{lem:dk=qkqbar}.  In
summary, the typical automorphism $(A,\mu,\qq)$ is as in the
previous Lie superalgebra, except that $a$ is unconstrained (but
nonzero).  It can thus take one of two possible forms:
\begin{equation}
  \begin{split}
    &A = \begin{pmatrix} a & \zero \\ c & |\qq|^2
    \end{pmatrix},\qquad \mu =\frac{|\qq|^2}{a} \qquad\text{and}\qquad \qq = q_4 + q_3 \kk\\
   \text{or}\qquad &A = \begin{pmatrix}
      a & \zero \\ c & - |\qq|^2
    \end{pmatrix}, \qquad \mu=-\frac{|\qq|^2}{a} \qquad\text{and}\qquad \qq = q_1 \ii + q_2 \jj.
  \end{split}
\end{equation}

The next two classes of Lie superalgebras are associated with the
one-parameter family of kinematical Lie algebras
\hyperlink{KLA3}{$\mathsf{K3}_\gamma$}, whose typical automorphisms
$(A,\mu)$ depend on the value of $\gamma \in [-1,1]$.  In the interior
of the interval, it takes the form
\begin{equation}
  A =
  \begin{pmatrix}
    a & \zero \\ \zero & d
  \end{pmatrix} \qquad\text{and}\qquad \mu = 1
\end{equation}
but at the boundaries this is enhanced: at $\gamma = -1$ one can also
have automorphisms of the form
\begin{equation}
  A =
  \begin{pmatrix}
    \zero & b \\ c & \zero
  \end{pmatrix} \qquad\text{and}\qquad \mu = -1,
\end{equation}
whereas at $\gamma = 1$, the typical automorphism takes the form
\begin{equation}
  A =
  \begin{pmatrix}
    a & b \\ c & d
  \end{pmatrix} \qquad\text{and}\qquad \mu = 1.
\end{equation}

\subsubsection{Automorphisms of Lie superalgebra $\mathsf{S9}$$_{\gamma,\lambda}$}
\label{sec:autom-kinem-lie-9}

Here $\hh= \tfrac12 (1 + \lambda \kk)$, $\bb=\pp=0$, $c_0=0$, $\bc_1 =
\bc_2 = 0$ and $\bc_3 = \kk$.  The invariance
conditions~\eqref{eq:aut-s} reduce to $b=0$ and, in addition, 
\begin{equation}
  \mu \qq (1 + \lambda \kk) = (1+\lambda \kk) \qq
  \qquad\text{and}\qquad d\kk = \qq\kk\qqbar.
\end{equation}
Taking the norm of the first equation, we find that $\mu = \pm 1$.  If
$\mu =1$, then $\lambda[\kk,\qq] =0$ so that either $\lambda\neq 0$,
in which case $\qq \in \spn{1,\kk}$ or $\lambda = 0$ and $\qq$ is not
constrained by this equation.  The second equation is dealt with by
Lemma~\ref{lem:dk=qkqbar}, which implies in particular that $d =
\pm |\qq|^2$ and since $\qq \neq 0$, $d \neq 0$.  This precludes the
case $\mu = -1$ by inspecting the possible automorphisms $(A,\mu)$  of
$\k$.  In summary, for generic $\gamma$ and $\lambda$, the typical
automorphism $(A,\mu,\qq)$ takes the form
\begin{equation}
  A =
  \begin{pmatrix}
    a & \zero \\ \zero & |\qq|^2
  \end{pmatrix}, \qquad \mu = 1 \qquad\text{and}\qquad \qq = q_4 + q_3
  \kk,
\end{equation}
which is enhanced for $\gamma = 1$ (but $\lambda$ still generic) to
\begin{equation}
  A =
  \begin{pmatrix}
    a & \zero \\ c & |\qq|^2
  \end{pmatrix}, \qquad \mu = 1 \qquad\text{and}\qquad \qq = q_4 + q_3
  \kk.
\end{equation}
If $\lambda = 0$, then the automorphisms are enhanced by the addition
of $(A,\mu,\qq)$ of the form
\begin{equation}
  A =
  \begin{pmatrix}
    a & \zero \\ \zero & -|\qq|^2
  \end{pmatrix}, \qquad \mu = 1 \qquad\text{and}\qquad \qq = q_1\ii + q_2 \jj,
\end{equation}
for generic $\gamma$ or, for $\gamma = 1$ only,  also
\begin{equation}
  A =
  \begin{pmatrix}
    a & \zero \\ c & -|\qq|^2
  \end{pmatrix}, \qquad \mu = 1 \qquad\text{and}\qquad \qq = q_1\ii +
  q_2 \jj.
\end{equation}

\subsubsection{Automorphisms of Lie superalgebra $\mathsf{S10}$$_{\gamma,\lambda}$}
\label{sec:autom-kinem-lie-10}

Here $\hh= \tfrac12(\gamma + \lambda\kk)$, $\bb=\pp=0$, $c_0=0$,
$\bc_1=\bc_3=0$ and $\bc_2 = \kk$.  The invariance
conditions~\eqref{eq:aut-s} imply that $c=0$ and also
\begin{equation}
  \mu \qq (\gamma + \lambda\kk) = (\gamma + \lambda\kk) \qq
  \qquad\text{and}\qquad a\kk = q\kk\qqbar.
\end{equation}
It is very similar to the previous Lie superalgebra, except that here
$\gamma \neq 1$.  Lemma~\ref{lem:dk=qkqbar} says now that either $a
= |\qq|^2$ and $\qq = q_4 + q_3 \kk$ or $a = - |\qq|^2$ and $\qq = q_1
\ii + q_2 \jj$. In particular, since $\qq \neq 0$, $a\neq 0$.  From
the expressions for the automorphisms $(A,\mu)$ of $\k$, we see that
$\mu = 1$.  This means that the first equation says $\qq$ commutes
with $\gamma + \lambda \kk$.   If $\lambda = 0$, this condition is
vacuous, but if $\lambda \neq 0$, then it forces $\qq = q_4 + q_3\kk$
and hence $a = |\qq|^2$.

In summary, for $\lambda \neq 0$ we have that
$(A,\mu,\qq)$ takes the form
\begin{equation}
    A =
    \begin{pmatrix}
      |q|^2 & \zero \\ \zero & d \end{pmatrix}, \qquad \mu = 1
    \qquad\text{and}\qquad \qq = q_4 +  q_3\kk,
\end{equation}
whereas if $\lambda = 0$ it can also take the form
\begin{equation}
    A =
    \begin{pmatrix}
      -|q|^2 & \zero \\ \zero & d \end{pmatrix}, \qquad \mu = 1
    \qquad\text{and}\qquad \qq = q_1\ii +  q_2\jj.
\end{equation}

The next Lie superalgebra is based on the kinematical Lie algebra
\hyperlink{KLA4}{$\mathsf{K4}_\chi$}, whose automorphisms $(A,\mu)$ take the form
\begin{equation}
  A =
  \begin{pmatrix}
    a & b \\ -b & a
  \end{pmatrix} \qquad\text{and}\qquad \mu = 1
\end{equation}
for generic $\chi$, whereas if $\chi = 0$, then they can also be of
the form
\begin{equation}
  A =
  \begin{pmatrix}
    a & b \\ b & -a
  \end{pmatrix} \qquad\text{and}\qquad \mu = -1.
\end{equation}

\subsubsection{Automorphisms of Lie superalgebra $\mathsf{S11}$$_\chi$}
\label{sec:autom-kinem-lie-11}

Here $\hh = \tfrac12 (\chi + \jj)$, $\bb=\pp=0$, $c_0=0$, $\bc_1 = 0$,
$\bc_2 = \ii$ and $\bc_3 = \kk$.  The invariance
conditions~\eqref{eq:aut-s} reduce to
\begin{equation}
  \mu \qq (\chi + \jj) = (\chi + \jj)\qq, \qquad \qq\ii\qqbar = a\ii +
  b\kk \qquad\text{and}\qquad \qq\kk\qqbar = c\ii + d\kk.
\end{equation}
The last two equations are solved via Lemma~\ref{lem:conj-plane}:
either $\qq = q_4 + q_2\jj$ or else $\qq = q_1\ii + q_3\kk$.  This
latter case can only happen when $\chi = 0$.  Substituting these
possible expressions for $\qq$ in the last two equations, we determine
the entries of the matrix $A$.

In summary, $(A,\mu,\qq)$ takes the form
\begin{equation}
  A =
  \begin{pmatrix}
    q_4^2 - q_2^2 & - 2 q_2 q_4 \\ 2 q_2 q_4 & q_4^2 - q_2^2
  \end{pmatrix}, \qquad \mu = 1 \qquad\text{and}\qquad \qq = q_4 + q_2\jj,
\end{equation}
and (only) if $\chi = 0$ it can also take the form
\begin{equation}
  A =
  \begin{pmatrix}
    q_1^2 - q_3^2 & 2 q_1 q_3 \\ 2 q_1 q_3 & q_3^2 - q_1^2
  \end{pmatrix}, \qquad \mu = -1 \qquad\text{and}\qquad \qq = q_1\ii + q_3\kk.
\end{equation}

The next Lie superalgebra is the supersymmetric extension of the
kinematical Lie algebra \hyperlink{KLA5}{$\mathsf{K5}$}, whose automorphisms
$(A,\mu)$ are of the form
\begin{equation}
  A =
  \begin{pmatrix}
    a & \zero \\ c & a
  \end{pmatrix} \qquad\text{and}\qquad \mu = 1.
\end{equation}

\subsubsection{Automorphisms of Lie superalgebra $\mathsf{S12}$$_\lambda$}
\label{sec:autom-kinem-lie-12}

Here $\hh = \tfrac12 (1 + \lambda \kk)$, $\bb=\pp=0$, $c_0 =0$, $\bc_1
= \bc_2 = 0$ and $\bc_3 = \kk$.  The invariance
conditions~\eqref{eq:aut-s} reduce to
\begin{equation}
  \qq\hh = \hh\qq \qquad\text{and}\qquad a\kk  = \qq \kk \qqbar.
\end{equation}
The second equation is solved via Lemma~\ref{lem:dk=qkqbar}, which
says that either $a = |\qq|^2$ and $\qq = q_4 + q_3 \kk$ or $a =
- |\qq|^2$ and $\qq = q_1\ii + q_2\jj$.  The first equation is
identically satisfied  if $\lambda =0$, but otherwise it forces
$\qq = q_4 + q_3 \kk$ and hence $a= |\qq|^2$.  In summary, for general
$\lambda$, an automorphism $(A,\mu,\qq)$ takes the form
\begin{equation}
  A =
  \begin{pmatrix}
    |\qq|^2 & \zero \\ c & |\qq|^2
  \end{pmatrix}, \qquad \mu = 1 \qquad\text{and}\qquad \qq = q_4 + q_3\kk,
\end{equation}
whereas if $\lambda = 0$, it can also take the form
\begin{equation}
  A =
  \begin{pmatrix}
    -|\qq|^2 & \zero \\ c & -|\qq|^2
  \end{pmatrix}, \qquad \mu = 1 \qquad\text{and}\qquad \qq = q_1\ii + q_2\jj.
\end{equation}

The next Lie superalgebra is the supersymmetric extension of the
Carroll algebra, whose automorphisms $(A,\mu)$ take the form
\begin{equation}
  A =
  \begin{pmatrix}
    a & b \\ c & d
  \end{pmatrix} \qquad\text{and}\qquad \mu = ad - bc.
\end{equation}

\subsubsection{Automorphisms of Lie superalgebra $\mathsf{S13}$}
\label{sec:autom-kinem-lie-13}

Here $\hh = \bb = \pp = 0$, $c_0 =1$ and $\bc_1 = \bc_2 = \bc_3 = 0$.
The invariance conditions \eqref{eq:aut-s} reduce to a single
condition: $ad - bc = |\qq|^2$.  The automorphisms $(A,\mu,\qq)$ are of
the form
\begin{equation}
  A  =   \begin{pmatrix}
    a & b \\ c & d
  \end{pmatrix}, \qquad \mu = ad - bc = |\qq|^2 \qquad\text{and}\qquad
  \qq \in\HH^\times.
\end{equation}

The next Lie superalgebra is the Poincaré superalgebra whose
($\r$-fixing) automorphisms $(A,\mu)$ can take one of two possible
forms:
\begin{equation}
  \begin{split}
    &A = \begin{pmatrix} 1 & \zero \\ c & d
    \end{pmatrix} \qquad\text{and}\qquad \mu = d\\
   \text{or}\qquad &A = \begin{pmatrix}
      -1 & \zero \\ c & d
    \end{pmatrix}\qquad\text{and}\qquad \mu = -d.
  \end{split}
\end{equation}

\subsubsection{Automorphisms of Lie superalgebra $\mathsf{S14}$}
\label{sec:autom-kinem-lie-14}

Here $\hh=\pp=0$, $\bb= \tfrac12 \kk$, $c_0=1$, $\bc_1 = \bc_2 = 0$
and $\bc_3 = \kk$.  The invariance conditions~\eqref{eq:aut-s}
translate into
\begin{equation}
  \pm \qq \kk = \kk\qq, \qquad d = \pm |\qq|^2  \qquad\text{and}\qquad
  d\kk = \qq\kk\qqbar,
\end{equation}
where the signs are correlated and the last equation follows from the first two.

Choosing the plus sign, $\qq\kk = \kk\qq$, so that $\qq =
q_4 + q_3 \kk$ and $d = |\qq|^2$, whereas choosing the minus sign,
$\qq\kk = - \kk\qq$, so that $\qq = q_1 \ii + q_2 \jj$ and $d =
-|\qq|^2$.

In summary, automorphisms $(A,\mu,\qq)$ of the Poincaré superalgebra
take the form
\begin{equation}
  \begin{split}
    &A = \begin{pmatrix} 1 & \zero \\ c & |\qq|^2
    \end{pmatrix}, \qquad \mu = |\qq|^2 \qquad\text{and}\qquad
    \qq = q_4 + q_3 \kk\\
    \text{or}\qquad
    &A = \begin{pmatrix} -1 & \zero \\ c & -|\qq|^2
    \end{pmatrix}, \qquad \mu = |\qq|^2 \qquad\text{and}\qquad
    \qq = q_1\ii + q_2 \jj.
  \end{split}
\end{equation}

The next Lie superalgebra is the AdS superalgebra, whose ($\r$-fixing)
automorphisms $(A,\mu)$ are of the form
\begin{equation}
  A =
  \begin{pmatrix}
    a & b \\ \mp b & \pm a
  \end{pmatrix} \qquad\text{and}\qquad \mu = \pm 1,
\end{equation}
where $a^2 + b^2 = 1$.

\subsubsection{Automorphisms of Lie superalgebra $\mathsf{S15}$}
\label{sec:autom-kinem-lie-15}

Here $\hh=\tfrac12\kk$, $\bb=\tfrac12\ii$, $\pp=\tfrac12\jj$, $c_0 =
1$, $\bc_1 = \kk$, $\bc_2 = \jj$ and $\bc_3= \ii$.  The invariance
conditions~\eqref{eq:aut-s} include $\mu = |\qq|^2$, which forces $\mu
= 1$.  Taking this into account, another of the invariance
conditions~\eqref{eq:aut-s} is $\qq\kk = \kk \qq$, which together with
$|\qq|=1$, forces $\qq = e^{\theta\kk}$.  The remaining invariance
conditions are
\begin{equation}
  a \qq \ii - b \qq \jj = \ii\qq, \qquad
  b \qq \ii + a \qq \jj = \jj\qq, \qquad
  a\jj\qq + b \ii\qq = \qq\jj \qquad\text{and}\qquad
  a\ii\qq - b\jj\qq = \qq\ii.
\end{equation}
Given the expression for $\qq$, these are solved by $a = \cos2\theta$
and $b = \sin2\theta$.  In summary, the ($\r$-fixing) automorphisms
$(A,\mu,\qq)$ of the AdS superalgebra are of the form
\begin{equation}
  A =
  \begin{pmatrix}
    \cos2\theta & \sin2\theta \\ -\sin2\theta & \cos2\theta
  \end{pmatrix}, \qquad \mu = 1 \qquad\text{and}\qquad \qq = e^{\theta\kk}.
\end{equation}

The next three Lie superalgebras in Table~\ref{tab:klsa} are
supersymmetric extensions of the kinematical Lie algebra
\hyperlink{KLA12}{$\mathsf{K12}$} in Table~\ref{tab:kla}, whose $\r$-fixing
automorphisms $(A,\mu)$ take the following form:
\begin{equation}
  A =
  \begin{pmatrix}
    1 & \zero \\ \zero & \pm 1
  \end{pmatrix} \qquad\text{and}\qquad \mu \in \RR^\times.
\end{equation}

\subsubsection{Automorphisms of Lie superalgebra $\mathsf{S16}$}
\label{sec:autom-kinem-lie-16}

Here $\hh=\bb=0$, $\pp= \tfrac12\jj$, $c_0= 0$, $\bc_1 = -\ii$, $\bc_2
= \ii$ and $\bc_3 = -\kk$.  The invariance conditions~\eqref{eq:aut-s}
reduce to
\begin{equation}
  \pm \qq \jj = \jj\qq, \qquad \pm \qq\kk = \kk \qq
  \qquad\text{and}\qquad \qq\ii\qqbar = \ii.
\end{equation}
It follows from the last equation that $|\qq|=1$ and hence that
$\qq\ii = \ii\qq$.  Depending on the (correlated) signs of the first
two equations, we find that, for the plus sign, $\qq$ commutes with
$\ii$, $\jj$ and $\kk$ and hence $\qq \in \RR$, but since $|\qq|=1$,
we must have $\qq = \pm 1$.  For the minus sign, we find that $\qq$
commutes with $\ii$ but anticommutes with $\jj$ and $\kk$, so that
$\qq = \pm \ii$, after taking into account that $|\qq|=1$.  In
summary, the automorphisms $(A,\mu,\qq)$ of this Lie
superalgebra take one of two possible forms:
\begin{equation}
  \begin{split}
    &A = \begin{pmatrix} 1 & \zero \\ \zero & 1
    \end{pmatrix},\qquad \mu \in \RR^\times \qquad\text{and}\qquad \qq
    = \pm 1\\
   \text{or}\qquad &A = \begin{pmatrix}
      1 & \zero \\ \zero & -1
    \end{pmatrix},\qquad \mu \in \RR^\times \qquad\text{and}\qquad \qq
    = \pm \ii.
  \end{split}
\end{equation}

\subsubsection{Automorphisms of Lie superalgebra $\mathsf{S17}$}
\label{sec:autom-kinem-lie-17}

Here $\hh = \pp = 0$, $\bb = \tfrac12$, $c_0 =1$ and $\bc_1 = \bc_2 =
\bc_3 = 0$.  There is only one invariance condition: namely, $\mu
= |\qq|^2$, and hence the automorphisms $(A,\mu,\qq)$ take the form
\begin{equation}
  A =
  \begin{pmatrix}
    1 & \zero \\ \zero & \pm 1
  \end{pmatrix}, \qquad \mu = |\qq|^2 \qquad\text{and}\qquad \qq \in \HH^\times.
\end{equation}

\subsubsection{Automorphisms of Lie superalgebra $\mathsf{S18}$}
\label{sec:autom-kinem-lie-18}

Here $\hh=\tfrac12\kk$, $\bb = \tfrac12$, $\pp = 0$, $c_0=1$, $\bc_1 =
\bc_3 = 0$ and $\bc_2=\kk$.  The invariance
conditions~\eqref{eq:aut-s} reduce to
\begin{equation}
  \mu \qq\kk = \kk\qq, \qquad \mu = |\qq|^2 \qquad\text{and}\qquad \kk
  = \qq\kk\qqbar.
\end{equation}
From the first equation we see that $\mu = \pm 1$, but from the
second it must be positive, so $\mu = 1$, which says implies that
$|\qq| = 1$ and hence that $\qq$ commutes with $\kk$.  In summary, the
typical automorphism $(A,\mu,\qq)$ takes the form
\begin{equation}
  A =
  \begin{pmatrix}
    1 & \zero \\ \zero & \pm 1
  \end{pmatrix}, \qquad \mu = 1 \qquad\text{and}\qquad \qq = e^{\theta \kk}.
\end{equation}

The next four Lie superalgebras in Table~\ref{tab:klsa} are
supersymmetric extensions of the kinematical Lie algebra \hyperlink{KLA13}{$\mathsf{K13}$} in
Table~\ref{tab:kla}, whose typical $\r$-fixing automorphisms $(A,\mu)$
take the form
\begin{equation}
  A =
  \begin{pmatrix}
    1 & \zero \\ \zero & \pm 1
  \end{pmatrix} \qquad\text{and}\qquad \mu \in \RR^\times.
\end{equation}

\subsubsection{Automorphisms of Lie superalgebra $\mathsf{S19}$}
\label{sec:autom-kinem-lie-19}

Here $\hh=\kk$, $\bb =0$, $\pp = \tfrac12$, $c_0=1$, $\bc_1 = \bc_3 =
\kk$ and $\bc_2= -\kk$.  The invariance conditions~\eqref{eq:aut-s}
are given by
\begin{equation}
  \mu \qq \kk = \kk\qq, \qquad \mu = |\qq|^2 \qquad\text{and}\qquad
  d\qq = \qq.
\end{equation}
The last equation says that $d=1$, whereas the first says that $\mu =
\pm 1$, but from the second equation it is positive and thus $\mu =
1$.  This also means $|\qq|=1$ and that $\qq\kk=\kk\qq$.  In summary,
the typical automorphism $(A,\mu,\qq)$ of $\s$ takes the form
\begin{equation}
  A =
  \begin{pmatrix}
    1 & \zero \\ \zero & 1
  \end{pmatrix}, \qquad \mu = 1 \qquad\text{and}\qquad \qq = e^{\theta\kk}.
\end{equation}

\subsubsection{Automorphisms of Lie superalgebra $\mathsf{S20}$}
\label{sec:autom-kinem-lie-20}

Here $\hh=\bb =0$, $\pp = \tfrac12$, $c_0=1$, $\bc_1 = \bc_2 = \bc_3=
0$.  The invariance conditions~\eqref{eq:aut-s} are given by
\begin{equation}
  d\qq = \qq \qquad\text{and}\qquad \mu = |\qq|^2.
\end{equation}
The first equation simply sets $d= 1$ and, in summary, the typical
automorphism of $\s$ is takes the form
\begin{equation}
  A =
  \begin{pmatrix}
    1 & \zero \\ \zero & 1
  \end{pmatrix}, \qquad \mu = |\qq|^2 \qquad\text{and}\qquad \qq \in \HH^\times.
\end{equation}

\subsubsection{Automorphisms of Lie superalgebra $\mathsf{S21}$}
\label{sec:autom-kinem-lie-22}

Here $\hh=\pp=0$, $\bb=\tfrac12$, $c_0=1$ and $\bc_1 = \bc_2 = \bc_3 =
0$.  The only invariance condition is $\mu = |\qq|^2$, so that the
typical automorphism $(A,\mu,\qq)$ takes the form
\begin{equation}
  A =
  \begin{pmatrix}
    1 & \zero \\ \zero & \pm 1
  \end{pmatrix}, \qquad \mu = |\qq|^2 \qquad\text{and}\qquad \qq \in \HH^\times.
\end{equation}

\subsubsection{Automorphisms of Lie superalgebra $\mathsf{S22}$}
\label{sec:autom-kinem-lie-21}

Here $\hh = \tfrac12\kk$, $\bb=\tfrac12$, $\pp=0$, $c_0=1$,
$\bc_1=\bc_3=0$ and $\bc_2=\kk$.  The invariance
conditions~\eqref{eq:aut-s} reduce to
\begin{equation}
  \mu \qq \kk = \kk\qq \qquad\text{and}\qquad \mu = |\qq|^2.
\end{equation}
The first equation says that $\mu = \pm 1$, but the second equation
says it is positive, so that $\mu = 1$ and $|\qq|=1$.  Furthermore,
$\qq$ commutes with $\kk$, so that $\qq = e^{\theta\kk}$.  In summary,
the typical automorphism $(A,\mu,\qq)$ takes the form
\begin{equation}
  A =
  \begin{pmatrix}
    1 & \zero \\ \zero & \pm 1
  \end{pmatrix}, \qquad \mu = 1 \qquad\text{and}\qquad \qq = e^{\theta\kk}.
\end{equation}

The next six Lie superalgebras in Table~\ref{tab:klsa} are
supersymmetric extensions of the kinematical Lie algebra
\hyperlink{KLA14}{$\mathsf{K14}$} in Table~\ref{tab:kla}, whose $\r$-fixing
automorphisms $(A,\mu)$ take the form
\begin{equation}
  A =
  \begin{pmatrix}
    1 & \zero \\ \zero & d
  \end{pmatrix} \qquad\text{and}\qquad \mu \in \RR^\times.
\end{equation}

\subsubsection{Automorphisms of Lie superalgebra $\mathsf{S23}$}
\label{sec:autom-kinem-lie-23}

Here $\hh= \kk$, $\bb=\pp=0$, $c_0=0$, $\bc_1 = \bc_2 = 0$ and $\bc_3
= \kk$.  The invariance conditions~\eqref{eq:aut-s} reduce to
\begin{equation}
  \mu \qq \kk = \kk \qq \qquad\text{and}\qquad d \kk = \qq \kk \qqbar.
\end{equation}
The first equation says that $\mu = \pm 1$, so that $\pm \qq \kk = \kk
\qq$.  The second equation follows from Lemma~\ref{lem:dk=qkqbar}:
either $d = |\qq|^2$ and hence $\qq = q_4 + q_3 \kk$ or $d = -|\qq|^2$
and hence $\qq = q_1 \ii + q_2 \jj$.  In summary, the typical
automorphism $(A,\mu,\qq)$ takes one of two possible forms:
\begin{equation}
  \begin{split}
    &A = \begin{pmatrix} 1 & \zero \\ \zero & |\qq|^2
    \end{pmatrix},\qquad \mu =1 \qquad\text{and}\qquad \qq
    = q_4 + q_3\kk\\
   \text{or}\qquad &A = \begin{pmatrix}
      1 & \zero \\ \zero & -|\qq|^2
    \end{pmatrix},\qquad \mu = -1 \qquad\text{and}\qquad \qq = q_1 \ii
    + q_2 \jj.
  \end{split}
\end{equation}

\subsubsection{Automorphisms of Lie superalgebra $\mathsf{S24}$}
\label{sec:autom-kinem-lie-24}

Here $\hh=\bb=\pp=0$, $c_0=0$, $\bc_1 = \bc_2 = 0$ and $\bc_3 = \kk$.
Hence the only invariance condition is $d\kk = \qq\kk\qqbar$.
Lemma~\ref{lem:dk=qkqbar} says that either $d = |\qq|^2$ and hence
$\qq = q_4 + q_3\kk$ or else $d = - |\qq|^2$ and hence $\qq = q_1 \ii
+ q_2 \jj$.  In summary, the typical automorphism $(A,\mu,\qq)$ takes
one of two possible forms:
\begin{equation}
  \begin{split}
    &A = \begin{pmatrix} 1 & \zero \\ \zero & |\qq|^2
    \end{pmatrix},\qquad \mu \in\RR^\times \qquad\text{and}\qquad \qq
    = q_4 + q_3\kk\\
   \text{or}\qquad &A = \begin{pmatrix}
      1 & \zero \\ \zero & -|\qq|^2
    \end{pmatrix},\qquad \mu \in\RR^\times \qquad\text{and}\qquad \qq = q_1 \ii
    + q_2 \jj.
  \end{split}
\end{equation}

\subsubsection{Automorphisms of Lie superalgebra $\mathsf{S25}$}
\label{sec:autom-kinem-lie-25}

Here $\hh=\bb=\pp=0$, $c_0=1$ and $\bc_1 = \bc_2 = \bc_3 = 0$, so that
the only invariance condition is $\mu = |\qq|^2$.  In summary, the
typical automorphism $(A,\mu,\qq)$ takes the form
\begin{equation}
  A =
  \begin{pmatrix}
    1 & \zero \\ \zero & d
  \end{pmatrix}, \qquad \mu = |\qq|^2 \qquad\text{and}\qquad \qq \in \HH^\times.
\end{equation}

\subsubsection{Automorphisms of Lie superalgebra $\mathsf{S26}$}
\label{sec:autom-kinem-lie-26}

Here $\hh=\bb=\pp=0$, $c_0=1$, $\bc_1 = \bc_2 = 0$ and $\bc_3 = \kk$,
so that there are two conditions in \eqref{eq:aut-s}:
\begin{equation}
  \mu = |\qq|^2 \qquad\text{and}\qquad d \kk = \qq \kk \qqbar.
\end{equation}
The second equation can be solved via Lemma~\ref{lem:dk=qkqbar}:
either $d = |\qq|^2$ and $\qq = q_4 + q_3\kk$ or $d= - |\qq|^2$ and
$\qq = q_1 \ii + q_2 \jj$.  In summary, the automorphisms
$(A,\mu,\qq)$ take one of two possible forms:
\begin{equation}
  \begin{split}
    &A = \begin{pmatrix} 1 & \zero \\ \zero & |\qq|^2
    \end{pmatrix},\qquad \mu =|\qq|^2 \qquad\text{and}\qquad \qq
    = q_4 + q_3\kk\\
   \text{or}\qquad &A = \begin{pmatrix}
      1 & \zero \\ \zero & -|\qq|^2
    \end{pmatrix},\qquad \mu =|\qq|^2 \qquad\text{and}\qquad \qq = q_1 \ii
    + q_2 \jj.
  \end{split}
\end{equation}

\subsubsection{Automorphisms of Lie superalgebra $\mathsf{S27}$}
\label{sec:autom-kinem-lie-27}

Here $\hh=\pp=0$, $\bb=\frac12$, $c_0= 1$ and $\bc_1 = \bc_2 = \bc_3 =
0$, so that the only invariance condition is $\mu = |\qq|^2$.
Therefore the typical automorphism $(A,\mu,\qq)$ takes the form
\begin{equation}
  A =
  \begin{pmatrix}
    1 & \zero \\ \zero & d
  \end{pmatrix}, \qquad \mu = |\qq|^2 \qquad\text{and}\qquad \qq \in \HH^\times.
\end{equation}

\subsubsection{Automorphisms of Lie superalgebra $\mathsf{S28}$}
\label{sec:autom-kinem-lie-28}

Here $\hh=\bb=\frac12$, $\pp=0$, $c_0=1$, $\bc_1 = \bc_3 = 0$ and
$\bc_2 = \kk$.  The invariance conditions~\eqref{eq:aut-s} reduce to
the following:
\begin{equation}
  \mu = |\qq|^2, \qquad \mu \qq = \qq \qquad\text{and}\qquad \kk = \qq
  \kk \qqbar.
\end{equation}
From the second equation we see that $\mu = 1$, so that from the first
$|\qq| = 1$ and hence $\kk \qq = \qq \kk$, so that $\qq =
e^{\theta\kk}$.  In summary, the typical automorphism $(A,\mu,\pp)$
takes the form
\begin{equation}
  A =
  \begin{pmatrix}
    1 & \zero \\ \zero & d
  \end{pmatrix}, \qquad \mu = 1 \qquad\text{and}\qquad \qq = e^{\theta\kk}.
\end{equation}

The next four Lie superalgebras in Table~\ref{tab:klsa} are
supersymmetric extensions of the kinematical Lie algebra
\hyperlink{KLA15}{$\mathsf{K15}$} in Table~\ref{tab:kla}, whose $\r$-fixing
automorphisms $(A,\mu)$ take the form
\begin{equation}
  A =
  \begin{pmatrix}
    a & \zero \\ c & a^2
  \end{pmatrix} \qquad\text{and}\qquad \mu \in \RR^\times.
\end{equation}

\subsubsection{Automorphisms of Lie superalgebra $\mathsf{S29}$}
\label{sec:autom-kinem-lie-29}

Here $\bb=\pp=0$, $\hh= \kk$, $c_0 = 0$, $\bc_1 = \bc_2 = 0$ and
$\bc_3 = \kk$.  The invariance conditions~\eqref{eq:aut-s} result in
\begin{equation}
  \mu\qq\kk = \kk\qq \qquad\text{and}\qquad a^2 \kk = \qq\kk\qqbar.
\end{equation}
Taking the norm of the first equation, we see that $\mu = \pm 1$, and
of the second equation, $a^2=|\qq|^2$.  This then says that $\qq$
commutes with $\kk$, so that $\mu = 1$ and $\qq = q_4 + q_3\kk$.  In
summary, the typical automorphism $(A,\mu,\qq)$ takes the form
\begin{equation}
  A =
  \begin{pmatrix}
    \pm |\qq| & \zero \\ c & |\qq|^2
  \end{pmatrix}, \qquad \mu = 1 \qquad\text{and}\qquad \qq = q_4 +
  q_3\kk.
\end{equation}

\subsubsection{Automorphisms of Lie superalgebra $\mathsf{S30}$}
\label{sec:autom-kinem-lie-30}

Here $\hh=\bb=\pp =0$, $c_0 = 0$, $\bc_1 = \bc_2 = 0$ and $\bc_3 =
\kk$.  The only invariance condition is $a^2\kk = \qq\kk\qqbar$.
Taking the norm, $a^2 = |\qq|^2$ and hence $\kk\qq = \qq\kk$ and thus
$\qq = q_4 + q_3 \kk$.  Hence the typical automorphism $(A,\mu,\qq)$
takes the form
\begin{equation}
  A =
  \begin{pmatrix}
    \pm |\qq| & \zero \\ c & |\qq|^2
  \end{pmatrix}, \qquad \mu \in \RR^\times \qquad\text{and}\qquad \qq
  = q_4 + q_3 \kk.
\end{equation}

\subsubsection{Automorphisms of Lie superalgebra $\mathsf{S31}$}
\label{sec:autom-kinem-lie-31}

Here $\hh=\bb=\pp=0$, $c_0=1$ and $\bc_1 = \bc_2 = \bc_3 = 0$, so that
the only invariance condition is $\mu = |\qq|^2$.  In summary, the
typical automorphism $(A,\mu,\qq)$ takes the form
\begin{equation}
  A =
  \begin{pmatrix}
    a & \zero \\ c & a^2
  \end{pmatrix}, \qquad \mu = |\qq|^2 \qquad\text{and}\qquad \qq \in \HH^\times.
\end{equation}

\subsubsection{Automorphisms of Lie superalgebra $\mathsf{S32}$}
\label{sec:autom-kinem-lie-32}

Here $\hh=\bb=\pp=0$, $c_0=1$, $\bc_1 = \bc_2 = 0$ and $\bc_3 = \kk$,
so that there are two invariance conditions:
\begin{equation}
  \mu = |\qq|^2 \qquad\text{and}\qquad a^2 \kk = \qq \kk \qqbar.
\end{equation}
The second shows that $a^2 = |\qq|^2$ and hence $\qq$ commutes with
$\kk$, so that $\qq = q_4 + q_3 \kk$.  In summary, the
typical automorphism $(A,\mu,\qq)$ takes the form
\begin{equation}
  A =
  \begin{pmatrix}
    \pm |\qq| & \zero \\ c & |\qq|^2
  \end{pmatrix}, \qquad \mu = |\qq|^2 \qquad\text{and}\qquad \qq = q_4
  + q_3 \kk.
\end{equation}

The next Lie superalgebra in Table~\ref{tab:klsa} is a one-parameter
family of supersymmetric extensions of the kinematical Lie algebra
\hyperlink{KLA16}{$\mathsf{K16}$} in Table~\ref{tab:kla}, whose $\r$-fixing
automorphisms $(A,\mu)$ take the form
\begin{equation}
  A =
  \begin{pmatrix}
    1 & \zero \\ \zero & d
  \end{pmatrix} \qquad\text{and}\qquad \mu =1.
\end{equation}

\subsubsection{Automorphisms of Lie superalgebra $\mathsf{S33}$}
\label{sec:autom-kinem-lie-33}

Here $\hh= \frac12 (1 + \lambda \kk)$, $\bb = \pp = 0$, $c_0 = 0$,
$\bc_1 = \bc_2 = 0$ and $\bc_3 = \kk$.  There are two invariance
conditions:
\begin{equation}
  \qq (1 + \lambda \kk) = (1 + \lambda \kk) \qq \qquad\text{and}\qquad
  d \kk = \qq \kk \qqbar.
\end{equation}
For the second equation we use Lemma~\ref{lem:dk=qkqbar} and for the
first equation we must distinguish between $\lambda =0$ and $\lambda
\neq 0$.  In the latter case, we have that $\qq = q_4 + q_3 \kk$ so
that only the $d = |\qq|^2$ of the lemma survives.  If $\lambda = 0$,
both branches survive.  In summary, for $\lambda \neq 0$, the typical
automorphism $(A,\mu,\qq)$ takes the form
\begin{equation}
  A =
  \begin{pmatrix}
    1 & \zero \\ \zero & |\qq|^2
  \end{pmatrix}, \qquad \mu = 1 \qquad\text{and}\qquad \qq = q_4 + q_3\kk,
\end{equation}
whereas if $\lambda = 0$ we have additional automorphisms of the form
\begin{equation}
  A =
  \begin{pmatrix}
    1 & \zero \\ \zero & -|\qq|^2
  \end{pmatrix}, \qquad \mu = 1 \qquad\text{and}\qquad \qq = q_1\ii + q_2\jj.
\end{equation}

The next Lie superalgebra in Table~\ref{tab:klsa} is the
supersymmetric extension of the kinematical Lie algebra
\hyperlink{KLA17}{$\mathsf{K17}$} in Table~\ref{tab:kla}, whose $\r$-fixing
automorphisms $(A,\mu)$ take the form
\begin{equation}
  A =
  \begin{pmatrix}
    a & \zero \\ c & a^2
  \end{pmatrix} \qquad\text{and}\qquad \mu =a.
\end{equation}

\subsubsection{Automorphisms of Lie superalgebra $\mathsf{S34}$}
\label{sec:autom-kinem-lie-34}

Here $\hh = \frac12 \kk$, $\bb = \pp = 0$, $c_0 = 0$, $\bc_1 = \bc_2 =
0$ and $\bc_3 = \kk$.  The invariance conditions are
\begin{equation}
  a \qq \kk = \kk\qq \qquad\text{and}\qquad a^2\kk = \qq \kk \qqbar.
\end{equation}
Taking norms of the first equation gives $a = \pm 1$ and hence $\pm
\qq\kk = \kk \qq$ and of the second equation $a^2 = |\qq|^2$ and hence
$\qq \kk = \kk \qq$.  This shows that $a = 1$ and hence $|\qq| = 1$,
so that $\qq = e^{\theta\kk}$.  In summary, the typical automorphism
$(A,\mu,\qq)$ takes the form
\begin{equation}
  A =
  \begin{pmatrix}
    1 & \zero \\ c & 1
  \end{pmatrix}, \qquad \mu = 1 \qquad\text{and}\qquad \qq = e^{\theta\kk}.
\end{equation}

The last Lie superalgebra in Table~\ref{tab:klsa} is a one-parameter
family of supersymmetric extensions of the kinematical Lie algebra
\hyperlink{KLA18}{$\mathsf{K18}$} in Table~\ref{tab:kla}, whose $\r$-fixing
automorphisms $(A,\mu)$ take the form
\begin{equation}
  A =
  \begin{pmatrix}
    a & \zero \\ \zero & a^2
  \end{pmatrix} \qquad\text{and}\qquad \mu =1.
\end{equation}

\subsubsection{Automorphisms of Lie superalgebra $\mathsf{S35}$}
\label{sec:autom-kinem-lie-35}

Here $\hh = 1 + \lambda \kk$, $\bb = \pp = 0$, $c_0 = 0$, $\bc_1 =
\bc_2 = 0$ and $\bc_3 = \kk$.  The invariance
conditions~\eqref{eq:aut-s} reduce to
\begin{equation}
  \qq (1 + \lambda \kk) = (1 + \lambda \kk) \qq \qquad\text{and}\qquad
  a^2\kk = \qq\kk\qqbar.
\end{equation}
Taking the norm of the second equation, $a^2 = |\qq|^2$ so that $\qq
\kk = \kk \qq$ and hence $\qq = q_4 + q_3 \kk$.  This also solves the
first equation, independently of the value of $\lambda$.  In summary,
the typical automorphism $(A,\mu,\qq)$ takes the form
\begin{equation}
  A =
  \begin{pmatrix}
    \pm |\qq| & \zero \\ \zero & |\qq|^2
  \end{pmatrix}, \qquad \mu = 1 \qquad\text{and}\qquad \qq = q_4 + q_3\kk.
\end{equation}

\subsubsection{Summary}
\label{sec:summary-1}

Tables~\ref{tab:aut-klsa} and \ref{tab:aut-klsa-extra} summarise the above discussion and lists the
typical automorphisms of each of the Lie superalgebras in Table~\ref{tab:klsa}.

\begin{table}[h!]
  \centering
  \caption{Automorphisms of kinematical Lie superalgebras}
  \label{tab:aut-klsa}
  \begin{tabular}{l|>{$}l<{$}}\toprule
    \multicolumn{1}{c|}{S\#} & \multicolumn{1}{c}{Typical $(A,\mu,\qq) \in \GL(2,\RR) \times \RR^\times \times \HH^\times$}\\
    \toprule
    \hypertarget{SAut1}{1} & \left(\begin{pmatrix} a & \zero \\ c & |\qq|^2 \end{pmatrix}, 1, q_4 + q_3 \kk\right),
   \left(\begin{pmatrix} a & \zero \\ c & -|\qq|^2 \end{pmatrix}, -1, q_1\ii + q_2 \jj\right)\\[10pt]
    \hypertarget{SAut2}{2} & \left(\begin{pmatrix} q_4^2-q_1^2 & 2 q_1 q_4 \\ -2 q_1 q_4 & q_4^2 - q_1^2
    \end{pmatrix}, q_1^2+q_4^2, q_4 + q_1 \ii\right), \left(\begin{pmatrix}
     q_2^2-q_3^2 & 2 q_2 q_3 \\ 2 q_2 q_3 & q_3^2 - q_2^2
    \end{pmatrix}, q_2^2+q_3^2, q_2 \jj + q_3 \kk\right)\\[10pt]
    \hypertarget{SAut3}{3} & \left(\begin{pmatrix} a & \zero \\ c & |\qq|^2 \end{pmatrix}, |\qq|^2, q_4 + q_3 \kk\right),
   \left(\begin{pmatrix} a & \zero \\ c & -|\qq|^2 \end{pmatrix}, |\qq|^2, q_1\ii + q_2 \jj\right)\\[10pt]
    \hypertarget{SAut4}{4} & \left(\begin{pmatrix} a & b \\ c & d \end{pmatrix}, |\qq|^2, \qq \right)\\[10pt]
    \hypertarget{SAut5}{5} & \left(\begin{pmatrix} q_4^2-q_1^2 & 2 q_1 q_4 \\ -2 q_1 q_4 & q_4^2 - q_1^2
    \end{pmatrix}, \mu, q_4 + q_1 \ii\right), \left(\begin{pmatrix}
     q_2^2-q_3^2 & 2 q_2 q_3 \\ 2 q_2 q_3 & q_3^2 - q_2^2
    \end{pmatrix}, \mu, q_2 \jj + q_3 \kk\right)\\[10pt]
    \hypertarget{SAut6}{6} & \left(\begin{pmatrix} a & \zero \\ c & |\qq|^2 \end{pmatrix}, \mu , q_4 + q_3 \kk\right),
   \left(\begin{pmatrix} a & \zero \\ c & -|\qq|^2 \end{pmatrix}, \mu , q_1\ii + q_2 \jj\right)\\[10pt]
    \hypertarget{SAut7}{7} & \left(\begin{pmatrix} |\qq|^2 & \zero \\ c & |\qq|^2 \end{pmatrix}, 1, q_4 + q_3 \kk\right),
   \left(\begin{pmatrix} |\qq|^2 & \zero \\ c & -|\qq|^2 \end{pmatrix}, -1, q_1\ii + q_2 \jj\right)\\[10pt]
    \hypertarget{SAut8}{8} & \left(\begin{pmatrix} a & \zero \\ c & |\qq|^2 \end{pmatrix}, \frac{|\qq|^2}{a}, q_4 + q_3 \kk\right),
   \left(\begin{pmatrix} a & \zero \\ c & -|\qq|^2 \end{pmatrix}, -\frac{|\qq|^2}{a}, q_1\ii + q_2 \jj\right)\\[10pt]
    \hypertarget{SAut9a}{9$_{\gamma\neq 1, \lambda \neq 0}$} & \left(\begin{pmatrix} a & \zero \\ \zero & |\qq|^2 \end{pmatrix}, 1, q_4 + q_3 \kk\right)\\[10pt]
    \hypertarget{SAut9b}{9$_{\gamma= 1, \lambda \neq 0}$} & \left(\begin{pmatrix} a & \zero \\ c & |\qq|^2 \end{pmatrix}, 1, q_4 + q_3 \kk\right)\\[10pt]
    \hypertarget{SAut9c}{9$_{\gamma\neq 1, \lambda= 0}$} & \left(\begin{pmatrix} a & \zero \\ \zero & |\qq|^2 \end{pmatrix}, 1, q_4 + q_3 \kk\right), \left(\begin{pmatrix} a & \zero \\ \zero & -|\qq|^2 \end{pmatrix}, 1, q_1\ii + q_2 \jj\right)\\[10pt]
    \hypertarget{SAut9d}{9$_{\gamma= 1, \lambda = 0}$} &  \left(\begin{pmatrix} a & \zero \\ c & |\qq|^2 \end{pmatrix}, 1, q_4 + q_3 \kk\right), \left(\begin{pmatrix} a & \zero \\ c & -|\qq|^2 \end{pmatrix}, 1, q_1\ii + q_2 \jj\right)\\[10pt]
    \hypertarget{SAut10a}{10$_{\gamma,\lambda\neq0}$} &  \left(\begin{pmatrix} |\qq|^2 & \zero \\ \zero & d \end{pmatrix}, 1, q_4 + q_3 \kk\right)\\[10pt]
    \hypertarget{SAut10b}{10$_{\gamma,\lambda=0}$} &  \left(\begin{pmatrix} |\qq|^2 & \zero \\ \zero & d \end{pmatrix}, 1, q_4 + q_3 \kk\right), \left(\begin{pmatrix} -|\qq|^2 & \zero \\ \zero & d \end{pmatrix}, 1, q_1\ii + q_2 \jj\right)\\[10pt]
    \hypertarget{SAut11a}{11$_{\chi> 0}$} & \left(\begin{pmatrix} q_4^2 - q_2^2 & - 2 q_2 q_4 \\ 2 q_2 q_4 & q_4^2 - q_2^2 \end{pmatrix}, 1, q_4 + q_2\jj\right)\\[10pt]
    \hypertarget{SAut11b}{11$_{\chi= 0}$} & \left(\begin{pmatrix} q_4^2 - q_2^2 & - 2 q_2 q_4 \\ 2 q_2 q_4 & q_4^2 - q_2^2 \end{pmatrix}, 1, q_4 + q_2\jj\right), \left(\begin{pmatrix} q_1^2 - q_3^2 & 2 q_1 q_3 \\ 2 q_1 q_3 & q_3^2 - q_1^2 \end{pmatrix}, -1, q_1\ii + q_3\kk\right)\\[10pt]
    \hypertarget{SAut12a}{12$_{\lambda\neq 0}$} & \left(\begin{pmatrix} |\qq|^2 & \zero \\ c & |\qq|^2 \end{pmatrix}, 1, q_4 + q_3\kk\right)\\[10pt]
    \hypertarget{SAut12b}{12$_{\lambda= 0}$} & \left(\begin{pmatrix} |\qq|^2 & \zero \\ c & |\qq|^2 \end{pmatrix}, 1, q_4 + q_3\kk\right), \left(\begin{pmatrix} -|\qq|^2 & \zero \\ c & -|\qq|^2 \end{pmatrix}, 1, q_1\ii + q_2\jj\right)\\[10pt]
    \hypertarget{SAut13}{13} & \left(\begin{pmatrix} a & b \\ c & d \end{pmatrix}, ad-bc = |\qq|^2, \qq\right)\\[10pt]
    \hypertarget{SAut14}{14} & \left(\begin{pmatrix} 1 & \zero \\ c & |\qq|^2 \end{pmatrix}, |\qq|^2, q_4 + q_3\kk\right), \left(\begin{pmatrix} -1 & \zero \\ c & -|\qq|^2 \end{pmatrix}, |\qq|^2 , q_1\ii + q_2\jj\right)\\[10pt]
    \hypertarget{SAut15}{15} & \left(\begin{pmatrix} \cos2\theta & -\sin2\theta \\ \sin2\theta & \cos2\theta \end{pmatrix}, 1, e^{\theta \kk}\right)\\[10pt]
    \bottomrule
  \end{tabular}
\end{table}

\begin{table}[h!]
  \centering
  \caption{Automorphisms of kinematical Lie superalgebras (continued)}
  \label{tab:aut-klsa-extra}
  \begin{tabular}{l|>{$}l<{$}}\toprule
    \multicolumn{1}{c|}{S\#} & \multicolumn{1}{c}{Typical $(A,\mu,\qq) \in \GL(2,\RR) \times \RR^\times \times \HH^\times$}\\
    \toprule
    \hypertarget{SAut16}{16} &\left(\begin{pmatrix} 1 & \zero \\ \zero & 1 \end{pmatrix},\mu, \pm 1 \right),
         \left(\begin{pmatrix} 1 & \zero \\ \zero & -1 \end{pmatrix}, \mu, \pm\ii \right)\\[10pt]
    \hypertarget{SAut17}{17} & \left(\begin{pmatrix} 1 & \zero \\ \zero & \pm 1 \end{pmatrix}, |\qq|^2, \qq \right)\\[10pt]
    \hypertarget{SAut18}{18} & \left(\begin{pmatrix} 1 & \zero \\ \zero & \pm 1 \end{pmatrix}, 1, e^{\theta\kk} \right)\\[10pt]
    \hypertarget{SAut19}{19} & \left(\begin{pmatrix} 1 & \zero \\ \zero & 1 \end{pmatrix}, 1, e^{\theta\kk} \right)\\[10pt]
    \hypertarget{SAut20}{20} & \left(\begin{pmatrix} 1 & \zero \\ \zero & 1 \end{pmatrix}, |\qq|^2, \qq \right)\\[10pt]
    \hypertarget{SAut21}{21} & \left(\begin{pmatrix} 1 & \zero \\ \zero & \pm 1 \end{pmatrix}, |\qq|^2, \qq \right)\\[10pt]
    \hypertarget{SAut22}{22} & \left(\begin{pmatrix} 1 & \zero \\ \zero & \pm 1 \end{pmatrix}, 1, e^{\theta\kk} \right)\\[10pt]
    \hypertarget{SAut23}{23} & \left(\begin{pmatrix} 1 & \zero \\ \zero & |\qq|^2 \end{pmatrix}, 1, q_4 + q_3\kk\right),
         \left(\begin{pmatrix} 1 & \zero \\ \zero & -|\qq|^2 \end{pmatrix}, -1, q_1\ii + q_2\jj\right)\\[10pt]
    \hypertarget{SAut24}{24} &\left(\begin{pmatrix} 1 & \zero \\ \zero & |\qq|^2 \end{pmatrix}, \mu , q_4 + q_3\kk\right),
         \left(\begin{pmatrix} 1 & \zero \\ \zero & -|\qq|^2 \end{pmatrix}, \mu , q_1\ii + q_2\jj\right)\\[10pt]
    \hypertarget{SAut25}{25} & \left(\begin{pmatrix} 1 & \zero \\ \zero & d \end{pmatrix}, |\qq|^2 , \qq \right)\\[10pt]
    \hypertarget{SAut26}{26} &\left(\begin{pmatrix} 1 & \zero \\ \zero & |\qq|^2 \end{pmatrix}, |\qq|^2 , q_4 + q_3\kk\right),
         \left(\begin{pmatrix} 1 & \zero \\ \zero & -|\qq|^2 \end{pmatrix}, |\qq|^2 , q_1\ii + q_2\jj\right)\\[10pt]
    \hypertarget{SAut27}{27} & \left(\begin{pmatrix} 1 & \zero \\ \zero & d \end{pmatrix}, |\qq|^2 , \qq \right)\\[10pt]
    \hypertarget{SAut28}{28} & \left(\begin{pmatrix} 1 & \zero \\ \zero & d \end{pmatrix}, 1 , e^{\theta\kk}\right)\\[10pt]
    \hypertarget{SAut29}{29} & \left(\begin{pmatrix} \pm|\qq| & \zero \\ c & |\qq|^2 \end{pmatrix}, 1 , q_4 + q_3\kk\right)\\[10pt]
    \hypertarget{SAut30}{30} & \left(\begin{pmatrix} \pm|\qq| & \zero \\ c & |\qq|^2 \end{pmatrix}, \mu , q_4 + q_3\kk\right)\\[10pt]
    \hypertarget{SAut31}{31} & \left(\begin{pmatrix} a & \zero \\ c & a^2 \end{pmatrix}, |\qq|^2 , \qq\right)\\[10pt]
    \hypertarget{SAut32}{32} & \left(\begin{pmatrix} \pm|\qq| & \zero \\ c & |\qq|^2 \end{pmatrix}, |\qq|^2 , q_4 + q_3\kk\right)\\[10pt]
    \hypertarget{SAut33a}{33$_{\lambda \neq 0}$} & \left(\begin{pmatrix} 1 & \zero \\ \zero & |\qq|^2 \end{pmatrix}, 1 , q_4 + q_3\kk\right)\\[10pt]
    \hypertarget{SAut33b}{33$_{\lambda = 0}$} & \left(\begin{pmatrix} 1 & \zero \\ \zero & |\qq|^2 \end{pmatrix}, 1 , q_4 + q_3\kk\right),
                       \left(\begin{pmatrix} 1 & \zero \\ \zero & -|\qq|^2 \end{pmatrix}, 1 , q_1\ii + q_2\jj\right)\\[10pt]
    \hypertarget{SAut34}{34} & \left(\begin{pmatrix} 1 & \zero \\ c & 1 \end{pmatrix}, 1, e^{\theta\kk} \right)\\[10pt]
    \hypertarget{SAut35}{35$_\lambda$} & \left(\begin{pmatrix} \pm|\qq| & \zero \\ \zero & |\qq|^2 \end{pmatrix}, 1 , q_4 + q_3\kk\right)\\[10pt]
    \bottomrule
  \end{tabular}
\end{table}

\section{Homogeneous superspaces}
\label{sec:homog-supersp}

In this section, we classify the simply-connected $(4|4)$-dimensional
homogeneous kinematical and aristotelian superspaces. We start by
classifying the super Lie pairs associated with the kinematical Lie
superalgebras. After determining the super Lie pairs, we select those
super Lie pairs $(\s, \h)$ which are effective in a basis where $\h$
is always the span of $\J$ and $\B$. In this way, the super Lie pair
is uniquely characterised by writing the Lie brackets of $\s$ in that
basis.

Before starting with the classification of super Lie pairs, we first
explain the relationship between super Lie pairs and homogeneous
supermanifolds.  We shall be brief and refer the reader to
\cite{MR2640006}, particularly Section~5, for the details.  Although
the treatment in that paper is phrased in the context of spin
manifolds, the results are more general and apply to the homogeneous
spacetimes under consideration, even in the absence of an invariant
pseudo-riemannian structure.

\subsection{Homogeneous supermanifolds}
\label{sec:homog-superm}

In this paper, we shall adopt the following definition for
supermanifolds (see, e.g., \cite{MR0580292}).

\begin{definition}
  A smooth \textbf{supermanifold} of dimension $(m|n)$ is a pair $(M,\eO)$,
  where the \textbf{body} $M$ is a smooth $m$-dimensional manifold and
  the \textbf{structure sheaf} $\eO$ is a sheaf of supercommutative
  superalgebras extending the sheaf $\eE$ of smooth function of $M$ by
  the subalgebra of nilpotent elements $\eN$; that is, we have an
  exact sequence of sheaves of supercommutative superalgebras:
  \begin{equation}
    \begin{tikzcd}
      0 \arrow[r] & \eN \arrow[r] & \eO \arrow[r] & \eE \arrow[r] & 0,
    \end{tikzcd}
  \end{equation}
  where for every $p \in M$, there is a neighbourhood $p \in U \subset
  M$ such that
  \begin{equation}
    \eO(U) \cong \eE(U) \otimes \wedge[\theta^1,\dots,\theta^n].
  \end{equation}
\end{definition}

All the homogeneous supermanifolds in this paper are \textbf{split}:
$\eO$ is isomorphic to the sheaf of sections of the exterior
algebra bundle of a homogeneous vector bundle $E \to M$; that is,
\begin{equation}
  \eO(U) = \Gamma\left(U, \oplus_{p\geq 0}\wedge^p E\right)
  \qquad\text{with}\qquad \eN(U) = \Gamma\left(U,
    \oplus_{p\geq 1}\wedge^p E\right).
\end{equation}
A celebrated theorem of Batchelor's states that any smooth supermanifold
always admits a splitting; although the splitting is not canonical
\cite{MR536951}.

Lie supergroups can be described as group objects in the category of
supermanifolds, but there is an equivalent description in terms of
Harish-Chandra pairs.  Indeed, there is an equivalence of categories
between Lie supergroups and \textbf{Harish-Chandra pairs}
\cite{MR0580292, MR760837} $(\Kgr,\s)$ consisting of a Lie group
$\Kgr$ and a Lie superalgebra $\s = \s_{\bar 0} \oplus \s_{\bar 1}$
where the Lie algebra of $\Kgr$ is (isomorphic to) $\s_{\bar 0}$ and
where the adjoint action of $\s_{\bar 0}$ on $\s$ lifts to an action
of $\Kgr$ on $\s$ by automorphisms.  By a result of Koszul
\cite{MR760837} (see also \cite[Thm.~2.2]{MR2640006}) the structure
sheaf of the Lie supergroup corresponding to a Harish-Chandra pair
$(\Kgr,\s)$ is the sheaf of smooth functions
$\Kgr \to \wedge^\bullet \s_{\bar 1}$, which can be interpreted as
the sheaf of smooth sections of the trivial vector bundle $\Kgr \times
\wedge^\bullet \s_{\bar 1}$ over $\Kgr$.

Now suppose that $M$ is a simply-connected homogeneous manifold
realising a pair $(\k,\h)$.  Recall that this means that
$M = \Kgr/\Hgr$ where $\Kgr$ is a connected and simply-connected Lie
group with Lie algebra $\k$ and $\Hgr$ is the connected Lie
subgroup of $\Kgr$ generated by $\h$, assumed closed.  Suppose that
$\s = \s_{\bar 0} \oplus \s_{\bar 1}$ is a Lie superalgebra with
$\s_{\bar 0} = \k$.  Then $S := \s_{\bar 1}$ is a representation of
$\k$ and, since $\Kgr$ is simply-connected, it is also a
representation of $\Kgr$ and, by restriction, also a representation of
$\Hgr$.  Let $E := \Kgr \times_{\Hgr} S$ denote the
homogeneous vector bundle over $M$ associated with the representation
$S$ of $\Hgr$.  We define a supermanifold $(M,\eO)$ where $\eO$ is the
sheaf of sections of the exterior bundle $\wedge^\bullet E$.  This
supermanifold is called the \textbf{superisation} of $M$ defined by
the Lie superalgebra $\s$ (\emph{cf.} \cite[Thm.~5.6]{MR2640006}).

Conversely, any homogeneous supermanifold is of this form.  Although
the result is more general, we need only the special case where $\Hgr
\subset \Kgr$ is a closed Lie subgroup.  Then the homogeneous
superisation of $\Kgr/\Hgr$ has as structure sheaf the
$\Hgr$-equivariant smooth functions $\Kgr \to \wedge^\bullet \s_{\bar
  1}$ (\emph{cf.} \cite[§3.3]{MR2640006}), but these are precisely the
smooth sections of the homogeneous vector bundle over $\Kgr/\Hgr$
associated to the representation $\wedge^\bullet \s_{\bar 1}$ of
$\Hgr$.

Therefore to every homogeneous superisation of $\Kgr/\Hgr$ we
may associate a pair $(\s,\h)$ and, conversely, every pair $(\s,\h)$
defines a homogeneous superisation of $\Kgr/\Hgr$.  Let us formally
define these pairs in our present context.

\begin{definition}
  A \textbf{super Lie pair} consists of a pair $(\s, \h)$ where $\s$
  is one of the kinematical Lie superalgebras in Table~\ref{tab:klsa}
  and $\h$ is a Lie subalgebra containing $\r$ and decomposing as
  $\h = \r \oplus V$ under the adjoint action of $\r$, where
  $V \subset \s_{\bar 0}$ is a copy of the vector representation.
  Just as in the non-super case discussed in
  \cite{Figueroa-OFarrill:2017ycu}, we shall refer to such Lie
  subalgebras as \textbf{admissible}.  Two super Lie pairs $(\s, \h)$
  and $(\s, \h')$ are \textbf{isomorphic} if there is an automorphism
  of $\s$ under which $\h$ goes to $\h'$.  We shall say that a super
  Lie pair $(\s,\h)$ is \textbf{geometrically realisable} if and only
  if so is the Lie pair $(\k,\h)$, where $\k = \s_{\bar 0}$.  We say
  that a super Lie pair $(\s,\h)$ is \textbf{effective} if $\h$ does
  not contain an ideal of $\s$.
\end{definition}

We observe that the condition of being geometrically realisable has
nothing to do with supersymmetry, whereas the condition of being
effective does take into account the whole superalgebra.  It is thus
possible, and indeed we will see examples below, that a geometrically
realisable super Lie pair $(\s,\h)$ is effective, but the underlying pair
$(\k,\h)$ is not.  In that case, the vectorial generators in $\h$ act
trivially on the body of the superspace, but nontrivially on the
fermionic coordinates; that is, they generate R-symmetries.

As in the classical theory, there is a one-to-one correspondence
between (isomorphism classes of) effective, geometrically realisable
super Lie pairs and (isomorphism classes of) homogeneous superisations
of homogeneous manifolds.  To the best of our knowledge, this result
is part of the mathematical folklore and we are not aware of any
reference where this result is proved or even stated as such.

\subsection{Admissible super Lie pairs}
\label{sec:slie-pairs}

We are now ready to classify admissible super Lie pairs up to
isomorphism.  We recall these are pairs $(\s,\h)$, where $\s$ is one
of the kinematical Lie superalgebras in Table~\ref{tab:klsa} and $\h$
is a Lie subalgebra $\h \subset \k = \s_{\bar 0}$ which is admissible
in the sense of \cite{Figueroa-OFarrill:2018ilb}; that is, it contains
the rotational subalgebra $\r$ and, as a representation of $\r$,
$\h = \r \oplus V$ where $V \subset \k$ is a copy of the vector
representation.  Two super Lie pairs $(\s,\h)$ and $(\s,\h')$ are
isomorphic if there is an automorphism of $\s$ which maps $\h$
(isomorphically) to $\h'$.  As in
\cite[§3]{Figueroa-OFarrill:2018ilb}, our strategy in classifying
admissible super Lie pairs up to isomorphism will be to take each
kinematical Lie superalgebra $\s$ in Table~\ref{tab:klsa} in turn,
determine the admissible subalgebras $\h$ and study the action of the
automorphisms in Tables~\ref{tab:aut-klsa} and
\ref{tab:aut-klsa-extra} on the space of admissible subalgebras in
order to select one representative from each orbit.  In particular,
every admissible super Lie pair $(\s,\h)$ defines a unique admissible
Lie pair $(\k,\h)$ which, if effective and geometrically realisable,
is associated with a unique simply-connected kinematical homogeneous
spacetime $\Kgr/\Hgr$.  That being the case, we may think of the super
Lie pair $(\s,\h)$ as a homogeneous kinematical superspacetime which
superises $\Kgr/\Hgr$.

Without loss of generality -- since an admissible subalgebra $\h$
contains $\r$ -- the vectorial complement $V$ can be taken to be the
span of $\alpha B_i + \beta P_i$, $i=1,2,3$, for some
$\alpha,\beta \in \RR$ not both zero, since the spans of
$\{J_i, \alpha B_i + \beta P_i\}$ and of
$\{J_i, \alpha B_i + \beta P_i + \gamma J_i\}$ coincide for all
$\gamma \in \RR$.  We will often use the shorthand
$V = \alpha \B + \beta \P$.  The determination of the possible
admissible subalgebras can be found in
\cite[§§3.1-2]{Figueroa-OFarrill:2018ilb}, but we cannot simply import
the results of that paper wholesale because here we are only allowed
to act with automorphisms of $\s$ and not just of $\k$.

As in that paper, we will eventually change basis in the Lie
superalgebra $\s$ so that the admissible subalgebra $\h$ is spanned by
$\J$ and $\B$.  Hence in determining the possible super Lie pairs, we
will keep track of the required change of basis, ensuring, where
possible, that $(\s,\h)$ is reductive; that is, such that $H, P_i,
Q_a$ (defined by equation~\eqref{eq:quat-basis-s}) span a subspace $\m
\subset \s$ complementary to $\h$ and such that $[\h,\m]\subset \m$.
This is equivalent to requiring that the span $\m_{\bar 0}$ of $H,
P_i$ satisfies $[\h,\m_{\bar 0}] \subset \m_{\bar 0}$, since the $Q_i$
span $\s_{\bar 1}$ and $[\h, \s_{\bar 1} ] \subset \s_{\bar 1}$ by
virtue of $\s$ being a Lie superalgebra.

It follows by inspection of
\cite[§§3.1-2]{Figueroa-OFarrill:2018ilb} that the Lie superalgebras
$\s$ whose automorphisms are listed in Table~\ref{tab:aut-klsa} are
extensions of kinematical Lie algebras $\k$ for which \emph{any}
vectorial subspace $V = \alpha \B + \beta \P$ defines an admissible
subalgebra $\h = \r \oplus V \subset \k$.  It is then a simple matter
to determine the orbits of the action of the automorphisms listed in
Table~\ref{tab:aut-klsa} on the space of vectorial subspaces and hence
to arrive at a list of possible inequivalent super Lie pairs
$(\s, \h)$ for such $\s$.

It also follows by inspection of
\cite[§§3.1-2]{Figueroa-OFarrill:2018ilb} that, of the remaining Lie
superalgebras (i.e., those whose automorphisms are listed in
Table~\ref{tab:aut-klsa-extra}), most are extensions of kinematical Lie
algebras possessing a unique vectorial subspace $V$ for which
$\h = \r \oplus V$ is an admissible subalgebra.  The exceptions are
those Lie superalgebras
\hyperlink{KLSA23}{$\mathsf{S23}$}--\hyperlink{KLSA28}{$\mathsf{S28}$}
and \hyperlink{KLSA33}{$\mathsf{S33}_\lambda$}, which are extensions
of the kinematical Lie algebras \hyperlink{KLA14}{$\mathsf{K14}$} and
\hyperlink{KLA16}{$\mathsf{K16}$}, respectively, for which there are
precisely two vectorial subspaces leading to admissible subalgebras.

Let us concentrate first on the Lie superalgebras
\hyperlink{KLSA1}{$\mathsf{S1}$}--\hyperlink{KLSA15}{$\mathsf{S15}$},
whose automorphisms are listed in Table~\ref{tab:aut-klsa}.  As
mentioned above, for $V$ any vectorial subspace, $\h = \r
\oplus V$ is an admissible subalgebra.  We need to determine the
orbits of the action of the automorphisms in
Table~\ref{tab:aut-klsa}.  Since $V = \alpha \B + \beta \P$, this is
equivalent to studying the action of the matrix part $A$ of the
automorphism $(A,\mu,\qq)$ on nonzero
vectors $(\alpha,\beta) \in \RR^2$.  In fact, since $(\alpha,\beta)$
and $(\lambda\alpha,\lambda\beta)$ for $0 \neq \lambda \in \RR$ denote
the same vectorial subspace, we must study the action of the subgroup
of $\GL(2,\RR)$ defined by the matrices $A$ in the automorphism group
on the projective space $\RP^1$.  The map $(A,\mu,\qq) \mapsto A$
defines a group homomorphism from the automorphism group of a Lie
superalgebra $\s$ to $\GL(2,\RR)$.  We will let $\Agr$ denote the
image of this homomorphism: it is a subgroup of $\GL(2,\RR)$ and it is
the action of $\Agr$ on $\RP^1$ that we need to investigate.  Of
course, $\Agr$ depends on $\s$, even though we choose not to overload
the notation by making this dependence explicit.

It follows by inspection of Table~\ref{tab:aut-klsa}, that for $\s$
any of the Lie superalgebras \hyperlink{SAut2}{$\mathsf{S2}$},
\hyperlink{SAut4}{$\mathsf{S4}$}, \hyperlink{SAut5}{$\mathsf{S5}$},
\hyperlink{SAut11a}{$\mathsf{S11}_{\chi\geq 0}$},
\hyperlink{SAut13}{$\mathsf{S13}$} and
\hyperlink{SAut15}{$\mathsf{S15}$}, the subgroup
$\Agr \subset \GL(2,\RR)$ acts transitively on $\RP^1$ and hence for
such Lie superalgebras there is a unique admissible subalgebra spanned
by $\J$ and $\B$. 

In contrast, if $\s$ is any of the Lie superalgebras
\hyperlink{SAut1}{$\mathsf{S1}$}, \hyperlink{SAut3}{$\mathsf{S3}$},
\hyperlink{SAut6}{$\mathsf{S6}$}, \hyperlink{SAut7}{$\mathsf{S7}$},
\hyperlink{SAut8}{$\mathsf{S8}$},
\hyperlink{SAut9b}{$\mathsf{S9}_{\gamma=1,\lambda\in\RR}$},
\hyperlink{SAut12a}{$\mathsf{S12}_{\lambda\in\RR}$} and
\hyperlink{SAut14}{$\mathsf{S14}$}, the subgroup $\Agr \subset
\GL(2,\RR)$ acts with two orbits on $\RP^1$.  For example, consider
the Lie superalgebra \hyperlink{SAut1}{$\mathsf{S1}$}, for which any $A
\in \Agr$ takes the form
\begin{equation}
  \begin{pmatrix}
    a & \zero \\ c & d
  \end{pmatrix} \qquad\text{for some $a,c,d \in \RR$ with $a,d\neq 0$,}
\end{equation}
and act as
\begin{equation}
  \begin{pmatrix}
    \alpha \\ \beta 
  \end{pmatrix}\mapsto   \begin{pmatrix}
    a & \zero \\ c & d
  \end{pmatrix}   \begin{pmatrix}
    \alpha \\ \beta 
  \end{pmatrix} =   \begin{pmatrix}
    a \alpha \\ d \beta + c \alpha
  \end{pmatrix}.
\end{equation}
If $\alpha \neq 0$, we can choose $c = -d\beta/\alpha$ to bring
$(\alpha,\beta)$ to $(a\alpha, 0)$ which is projectively equivalent to
$(1,0)$.  On the other hand, if $\alpha = 0$, then we cannot change
that via automorphisms and hence we have $(0,\beta)$, which is
projectively equivalent to $(0,1)$.  In summary, we have two
inequivalent admissible subalgebras with vectorial subspaces $V=\B$
and $V=\P$.  The same result holds for the other Lie superalgebras in
this list.

For the cases where $V=\P$ we change basis in the Lie superalgebra
$\s$ so that the admissible subalgebra $\h$ is spanned by $\J$ and
$\B$.  This results in different brackets, which we now proceed to
list.

\begin{table}[h!]
  \centering
  \caption{Super Lie pairs (with $V=\P$)}
  \label{tab:slp-vp-1}
  \setlength{\extrarowheight}{2pt}
  \rowcolors{2}{blue!10}{white}
    \begin{tabular}{l|*{3}{>{$}l<{$}}*{3}{|>{$}c<{$}}}\toprule
      \multicolumn{1}{c|}{S\#} & \multicolumn{3}{c|}{$\k$ brackets} & \multicolumn{1}{c|}{$\hh$} & \multicolumn{1}{c|}{$\pp$} & \multicolumn{1}{c}{$[\sQ(s),\sQ(s)]$}\\
      \toprule
      \hyperlink{KLSA1}{1} & & & & \tfrac12 \kk & & -\sB(s\kk\sbar) \\
      \hyperlink{KLSA3}{3} & & & & & & |s|^2 H - \sB(s\kk\sbar) \\
      \hyperlink{KLSA6}{6} & & & & & & -\sB(s\kk\sbar) \\
      \hyperlink{KLSA7}{7} & [H,\P]=-\B & & & \kk & & -\sB(s\kk\sbar) \\
      \hyperlink{KLSA8}{8} & [H,\P]=-\B & & & & & -\sB(s\kk\sbar) \\
      \hyperlink{KLSA9}{9$_{\gamma=1,\lambda\in\RR}$} & [H,\B]=\B & [H,\P]=\P & & \tfrac12 (1 + \lambda \kk) & & -\sB(s\kk\sbar) \\
      \hyperlink{KLSA12}{12$_{\lambda\in\RR}$} & [H,\B]=\B & [H,\P] = \B + \P & & \tfrac12 (1 + \lambda \kk) & & -\sB(s\kk\sbar) \\
      \hyperlink{KLSA14}{14} & [H,\P] = \B & [\B,\P] = H & [\P,\P] = - \J & & \tfrac12 \kk & |s|^2 H + \sB(s\kk\sbar) \\
      \bottomrule
    \end{tabular}
\end{table}

Finally, if $\s$ is any of the Lie superalgebras
\hyperlink{SAut9a}{$\mathsf{S9}_{\gamma\neq 1,\lambda\in\RR}$} and
\hyperlink{SAut10a}{$\mathsf{S10}_{\gamma,\lambda\in\RR}$}, the
subgroup $\Agr \subset \GL(2,\RR)$ acts with three orbits.  Indeed,
the matrices $A \in \Agr$ are now diagonal and of the form
\begin{equation}
  \begin{pmatrix}
    a & \zero \\ \zero & d
  \end{pmatrix},
\end{equation}
where at least one of $a,d$ can take \emph{any} nonzero value.  If
$(\alpha,\beta)$ is such that $\alpha = 0$ or $\beta = 0$, we cannot
alter this via automorphisms and hence projectively we have either
$(1,0)$ or $(0,1)$.  If $\alpha\beta \neq 0$, then we can always bring
it to $(1,1)$ or $(-1,-1)$ via an automorphism, but these are
projectively equivalent.  In summary, we have three orbits,
corresponding to $V=\B$, $V = \P$ and $V = \B + \P$.

When $V = \P$, the Lie brackets of
\hyperlink{KLSA9}{$\mathsf{S9}_{\gamma\neq 1,\lambda\in\RR}$} in the
new basis are given by
\begin{equation}
 [H,\B] = \B, \quad [H,\P]=\gamma\P,\quad [H,\sQ(s)]=\sQ(\tfrac12 s
 (1 + \lambda\kk)) \quad\text{and}\quad [\sQ(s),\sQ(s)] =
 -\sB(s\kk\sbar),
\end{equation}
and those of \hyperlink{KLSA10}{$\mathsf{S10}_{\gamma,\lambda\in\RR}$}
by
\begin{equation}
  [H,\B] = \B, \quad [H,\P]=\gamma\P,\quad [H,\sQ(s)]=\sQ(\tfrac12 s
  (\gamma + \lambda\kk)) \quad\text{and}\quad [\sQ(s),\sQ(s)] =
  -\sP(s\kk\sbar).
\end{equation}

On the other hand, when $V = \B + \P$, the Lie brackets of
\hyperlink{KLSA9}{$\mathsf{S9}_{\gamma\neq 1,\lambda\in\RR}$} in the
new basis are given by
\begin{equation}
  \begin{aligned}[m]
    [H,\B] &= -\P\\
    [H,\P] &= \gamma\B + (1+\gamma)\P\\
  \end{aligned}
  \qquad\qquad
  \begin{aligned}[m]
    [H,\sQ(s)] &= \sQ(\tfrac12 s (1 + \lambda\kk))\\
    [\sQ(s),\sQ(s)] &= \tfrac{1}{1-\gamma}(\gamma \sB(s\kk\sbar) + \sP(s\kk\sbar)),
  \end{aligned}
\end{equation}
and those of \hyperlink{KLSA10}{$\mathsf{S10}_{\gamma,\lambda\in\RR}$}
by
\begin{equation}
  \begin{aligned}[m]
    [H,\B] &= -\P\\
    [H,\P] &= \gamma\B + (1+\gamma)\P\\
  \end{aligned}
  \qquad\qquad
  \begin{aligned}[m]
    [H,\sQ(s)] &= \sQ(\tfrac12 s (\gamma + \lambda\kk))\\
    [\sQ(s),\sQ(s)] &= \tfrac{1}{\gamma-1}(\sB(s\kk\sbar) + \sP(s\kk\sbar)).
  \end{aligned}
\end{equation}

Now we turn to the Lie superalgebras whose automorphisms are listed in
Table~\ref{tab:aut-klsa-extra}.  If $\s$ is one such Lie superalgebra,
not every vectorial subspace leads to an admissible subalgebra.  From
the results in \cite[§§3.1-2]{Figueroa-OFarrill:2018ilb} we have that
Lie superalgebras
\hyperlink{SAut16}{$\mathsf{S16}$}--\hyperlink{SAut22}{$\mathsf{S22}$}
admit a unique admissible subalgebra with $V = \B$, whereas for the
Lie superalgebras
\hyperlink{SAut29}{$\mathsf{S29}$}--\hyperlink{SAut32}{$\mathsf{S32}$},
\hyperlink{SAut34}{$\mathsf{S34}$} and
\hyperlink{SAut35}{$\mathsf{S35}_{\lambda\in\RR}$} also admit a unique
admissible subalgebra with $V = \P$.  Finally, the Lie superalgebras
\hyperlink{SAut23}{$\mathsf{S23}$}--\hyperlink{SAut28}{$\mathsf{S28}$}
and \hyperlink{SAut33a}{$\mathsf{S33}_{\lambda\in\RR}$} admit precisely
two admissible subalgebras with $V= \B$ and $V= \P$, which cannot be
related by automorphisms.

\begin{table}[h!]
  \centering
  \caption{More super Lie pairs (with $V=\P$)}
  \label{tab:slp-vp-2}
  \setlength{\extrarowheight}{2pt}
  \rowcolors{2}{blue!10}{white}
    \begin{tabular}{l|*{3}{>{$}l<{$}}*{3}{|>{$}c<{$}}}\toprule
      \multicolumn{1}{c|}{S\#} & \multicolumn{3}{c|}{$\k$ brackets} & \multicolumn{1}{c|}{$\hh$} & \multicolumn{1}{c|}{$\pp$} & \multicolumn{1}{c}{$[\sQ(s),\sQ(s)]$}\\
      \toprule
      \hyperlink{KLSA23}{23} & [\P,\P] = \P & & & \kk & & -\sB(s\kk\sbar) \\
      \hyperlink{KLSA24}{24} & [\P,\P] = \P & & & & & - \sB(s\kk\sbar) \\
      \hyperlink{KLSA25}{25} & [\P,\P] = \P & & & & & |s|^2 H \\
      \hyperlink{KLSA26}{26} & [\P,\P] = \P & & & & & |s|^2 H - \sB(s\kk\sbar) \\
      \hyperlink{KLSA27}{27} & [\P,\P]= \P & & & & \tfrac12  & |s|^2 H \\
      \hyperlink{KLSA28}{28} & [\P,\P]= \P & & & \tfrac12 \kk & \tfrac12 & |s|^2 H -\sP(s\kk\sbar) \\
      \hyperlink{KLSA29}{29} & [\P,\P]= \B & & & \kk & & -\sB(s\kk\sbar) \\
      \hyperlink{KLSA30}{30} & [\P,\P] = \B & & & & & -\sB(s\kk\sbar) \\
      \hyperlink{KLSA31}{31} & [\P,\P] = \B & & & & & |s|^2 H \\
      \hyperlink{KLSA32}{32} & [\P,\P] = \B & & & & & |s|^2 H - \sB(s\kk\sbar) \\
      \hyperlink{KLSA33}{33$_{\lambda\in\RR}$} & [H,\B] = \B & [\P,\P] = \P & & \tfrac12 (1+\lambda \kk) & & - \sB(s\kk\sbar) \\
      \hyperlink{KLSA34}{34} & [H,\P] = -\B & [\P,\P] = \B & & \tfrac12 \kk & & - \sB(s\kk\sbar) \\
      \hyperlink{KLSA35}{35$_{\lambda\in\RR}$} & [H,\P] = \P & [H,\B] =2\B & [\P,\P] = \B & 1+\lambda \kk & & - \sB(s\kk\sbar) \\
      \bottomrule
    \end{tabular}
\end{table}

Table~\ref{tab:super-lie-pairs} summarises the above results.  For
each Lie superalgebra $\s$ in Table~\ref{tab:klsa} it lists the
admissible subalgebras $\h$ and hence the possible super Lie pairs
$(\s,\h)$.  The notation for $\h$ is simply the generators of the
vectorial subspace $V \subset \h$, where the span of
$\alpha B_a + \beta P_a$ is abbreviated as $\alpha \B + \beta \P$.
The blue entries correspond to effective super Lie pairs, whereas the
green and greyed out correspond to non-effective super Lie pairs: the
green ones giving rise to aristotelian superspaces upon quotienting by
ideal.  In Section~\ref{sec:class-arist-lie}, we classified
aristotelian Lie superspaces by classifying their corresponding
aristotelian Lie superalgebras (see Table~\ref{tab:alsa}) and in
Section~\ref{sec:arist-super-lie} we exhibit the precise
correspondence between the aristotelian non-effective super Lie pairs
and the aristotelian superspaces (see
Table~\ref{tab:aristo-correspondence}).

\begin{table}[h!]
  \centering
  \caption{Summary of super Lie pairs}
  \label{tab:super-lie-pairs}
  \resizebox{\textwidth}{!}{
    \setlength{\extrarowheight}{2pt}
    \begin{tabular}{l|l*{3}{|>{$}c<{$}}}\toprule
      \multicolumn{1}{c|}{$\s$} & \multicolumn{1}{c|}{$\k$} & \multicolumn{3}{c}{$V \subset \h$}\\
      \toprule
      \hyperlink{KLSA1}{$\mathsf{S1}$} & \hyperlink{KLA1}{$\mathsf{K1}$} & \ari{\B} & \non{\P} & \\
      \hyperlink{KLSA2}{$\mathsf{S2}$} & \hyperlink{KLA1}{$\mathsf{K1}$} & \ari{\B} & & \\
      \hyperlink{KLSA3}{$\mathsf{S3}$} & \hyperlink{KLA1}{$\mathsf{K1}$} & \ari{\B} & \ari{\P} & \\
      \hyperlink{KLSA4}{$\mathsf{S4}$} & \hyperlink{KLA1}{$\mathsf{K1}$} & \ari{\B} & & \\
      \hyperlink{KLSA5}{$\mathsf{S5}$} & \hyperlink{KLA1}{$\mathsf{K1}$} & \ari{\B} & & \\
      \hyperlink{KLSA6}{$\mathsf{S6}$} & \hyperlink{KLA1}{$\mathsf{K1}$} & \ari{\B} & \non{\P} & \\
      \hyperlink{KLSA7}{$\mathsf{S7}$} & \hyperlink{KLA2}{$\mathsf{K2}$} & \eff{\B} & \non{\P} & \\
      \hyperlink{KLSA8}{$\mathsf{S8}$} & \hyperlink{KLA2}{$\mathsf{K2}$} & \eff{\B} & \non{\P} & \\
      \hyperlink{KLSA9}{$\mathsf{S9}_{\gamma\in[-1,1),\lambda\in\RR}$} & \hyperlink{KLA3}{$\mathsf{K3}_\gamma$} & \ari{\B} & \non{\P} & \eff{\B + \P}\\
      \hyperlink{KLSA9}{$\mathsf{S9}_{\gamma=1,\lambda\in\RR}$} &  \hyperlink{KLA3}{$\mathsf{K3}_{\gamma=1}$} & \ari{\B} & \non{\P} &\\
      \hyperlink{KLSA10}{$\mathsf{S10}_{\gamma\in[-1,1),\lambda\in\RR}$} & \hyperlink{KLA3}{$\mathsf{K3}_\gamma$} & \non{\B} & \ari{\P} & \eff{\B + \P} \\
      \hyperlink{KLSA11}{$\mathsf{S11}_{\chi\geq0}$} & \hyperlink{KLA4}{$\mathsf{K4}_\chi$} & \eff{\B} & & \\
      \bottomrule
    \end{tabular}
    \hspace{2cm}
    \begin{tabular}{l|l*{2}{|>{$}c<{$}}}\toprule
      \multicolumn{1}{c|}{$\s$} & \multicolumn{1}{c|}{$\k$} & \multicolumn{2}{c}{$V \subset \h$}\\
      \toprule
      \hyperlink{KLSA12}{$\mathsf{S12}_{\lambda\in\RR}$} & \hyperlink{KLA5}{$\mathsf{K5}$} & \eff{\B} & \non{\P} \\
      \hyperlink{KLSA13}{$\mathsf{S13}$} & \hyperlink{KLA6}{$\mathsf{K6}$} & \eff{\B} &  \\
      \hyperlink{KLSA14}{$\mathsf{S14}$} & \hyperlink{KLA8}{$\mathsf{K8}$} & \eff{\B} & \eff{\P } \\
      \hyperlink{KLSA15}{$\mathsf{S15}$} & \hyperlink{KLA11}{$\mathsf{K11}$} & \eff{\B} &  \\
      \hyperlink{KLSA16}{$\mathsf{S16}$} & \hyperlink{KLA12}{$\mathsf{K12}$} & \ari{\B} &  \\
      \hyperlink{KLSA17}{$\mathsf{S17}$} & \hyperlink{KLA12}{$\mathsf{K12}$} & \eff{\B} &  \\
      \hyperlink{KLSA18}{$\mathsf{S18}$} & \hyperlink{KLA12}{$\mathsf{K12}$} & \eff{\B} & \\ 
      \hyperlink{KLSA19}{$\mathsf{S19}$} & \hyperlink{KLA13}{$\mathsf{K13}$} & \ari{\B} & \\ 
      \hyperlink{KLSA20}{$\mathsf{S20}$} & \hyperlink{KLA13}{$\mathsf{K13}$} & \ari{\B} & \\ 
      \hyperlink{KLSA21}{$\mathsf{S21}$} & \hyperlink{KLA13}{$\mathsf{K13}$} & \eff{\B} & \\ 
      \hyperlink{KLSA22}{$\mathsf{S22}$} & \hyperlink{KLA13}{$\mathsf{K13}$} & \eff{\B} & \\ 
      \hyperlink{KLSA23}{$\mathsf{S23}$} & \hyperlink{KLA14}{$\mathsf{K14}$} & \ari{\B} & \non{\P} \\
      \bottomrule
    \end{tabular}
    \hspace{2cm}
    \begin{tabular}{l|l*{2}{|>{$}c<{$}}}\toprule
      \multicolumn{1}{c|}{$\s$} & \multicolumn{1}{c|}{$\k$} & \multicolumn{2}{c}{$V \subset \h$}\\
      \toprule
      \hyperlink{KLSA24}{$\mathsf{S24}$} & \hyperlink{KLA14}{$\mathsf{K14}$} & \ari{\B} & \non{\P} \\
      \hyperlink{KLSA25}{$\mathsf{S25}$} & \hyperlink{KLA14}{$\mathsf{K14}$} & \ari{\B} & \ari{\P} \\
      \hyperlink{KLSA26}{$\mathsf{S26}$} & \hyperlink{KLA14}{$\mathsf{K14}$} & \ari{\B} & \ari{\P} \\
      \hyperlink{KLSA27}{$\mathsf{S27}$} & \hyperlink{KLA14}{$\mathsf{K14}$} & \eff{\B} & \ari{\P} \\
      \hyperlink{KLSA28}{$\mathsf{S28}$} & \hyperlink{KLA14}{$\mathsf{K14}$} & \eff{\B} & \ari{\P} \\
      \hyperlink{KLSA29}{$\mathsf{S29}$} & \hyperlink{KLA15}{$\mathsf{K15}$} & \non{\P} & \\ 
      \hyperlink{KLSA30}{$\mathsf{S30}$} & \hyperlink{KLA15}{$\mathsf{K15}$} & \non{\P} & \\ 
      \hyperlink{KLSA31}{$\mathsf{S31}$} & \hyperlink{KLA15}{$\mathsf{K15}$} & \ari{\P} & \\ 
      \hyperlink{KLSA32}{$\mathsf{S32}$} & \hyperlink{KLA15}{$\mathsf{K15}$} & \ari{\P} & \\ 
      \hyperlink{KLSA33}{$\mathsf{S33}_{\lambda\in\RR}$} & \hyperlink{KLA16}{$\mathsf{K16}$} & \ari{\B} & \non{\P} \\
      \hyperlink{KLSA34}{$\mathsf{S34}$} & \hyperlink{KLA17}{$\mathsf{K17}$} & \non{\P} & \\
      \hyperlink{KLSA35}{$\mathsf{S35}_{\lambda\in\RR}$} & \hyperlink{KLA18}{$\mathsf{K18}$} & \non{\P} &\\
      \bottomrule    
    \end{tabular}
  }
  \caption*{The blue pairs (e.g., \eff{\scriptsize\B}) are effective; the
    green pairs (e.g., \ari{\scriptsize\B}) though not effective, give rise to
    aristotelian superspaces; whereas the greyed out pairs (e.g.,
    \non{\scriptsize\B}) are not effective and will not be considered further.}
\end{table}

\subsection{Effective super Lie pairs}
\label{sec:effective-super-lie}

Recall that a super Lie pair $(\s,\h)$ is said to be
\emph{effective} if $\h$ does not contain an ideal of $\s$.  Since
$\h \subset \k$ and contains the rotational subalgebra, which has
nonvanishing brackets with $\Q$, the only possible ideal of $\s$
contained in $\h$ would be the vectorial subspace $V \subset \h$.  It
is then a simple matter to inspect the super Lie pairs determined in
the previous section and select those for which $V$ is not an ideal of
$\s$.  Those super Lie pairs have been highlighted in blue in
Table~\ref{tab:super-lie-pairs}.  We now take each such super Lie pair
in turn, change basis if needed so that $V$ is spanned by $\B$, and
then list the resulting brackets in that basis.  Every such super Lie
pair $(\s,\h)$ determines a Lie pair $(\k,\h)$.  If the Lie pair
$(\k,\h)$ is effective (and geometrically realisable), then $(\s,\h)$
describes a homogeneous superisation of one of the spatially-isotropic
homogeneous spacetimes in \cite{Figueroa-OFarrill:2018ilb}.  We remark
that there are effective super Lie pairs $(\s,\h)$ for which the
underlying Lie pair $(\k,\h)$ is not effective.  In those cases, there
are no boosts on the body of the superspacetime, but instead there
are R-symmetries in the odd coordinates.

As usual, in writing the Lie brackets of $\s$ below we do not include
any bracket involving $\J$, which are given in
equation~\eqref{eq:klsa-brackets-quat} and instead give any non-zero
additional brackets.

\subsubsection{Galilean superspaces}
\label{sec:super-g}

Galilean spacetime is described by $(\k,\h)$ where $\k$ has the
additional bracket $[H,\B] = - \P$.  There are two possible
superisations $(\s,\h)$, with brackets
\begin{equation}
  [H, \sQ(s)] =
  \begin{cases}
    \sQ(s\kk) \\ 0
  \end{cases} \qquad\text{and}\qquad
  [\sQ(s), \sQ(s)] = - \sP(s\kk\sbar).
\end{equation}
These are associated with Lie superalgebras \hyperlink{KLSA7}{$\mathsf{S7}$}
and \hyperlink{KLSA8}{$\mathsf{S8}$} in Table~\ref{tab:klsa}.

\subsubsection{Galilean de Sitter superspace}
\label{sec:super-dsg}

Galilean de Sitter spacetime is described by $(\k,\h)$ where $\k$ has the
additional brackets $[H,\B] = - \P$ and $[H,\P] = -\B$.  There are two
one-parameter family of superisations $(\s,\h)$, with brackets
\begin{equation}
  [H, \sQ(s)] =\sQ(\tfrac12 s (\pm 1+\lambda\kk)) \qquad\text{and}\qquad
  [\sQ(s), \sQ(s)] = - \tfrac12 (\sB(s\kk\sbar) \mp \sP(s\kk\sbar))
\end{equation}
for $\lambda \in \RR$.  They are associated with Lie superalgebras
\hyperlink{KLSA9}{$\mathsf{S9}$$_{\gamma=-1,\lambda}$} and
\hyperlink{KLSA10}{$\mathsf{S10}$$_{\gamma=-1,\lambda}$}, respectively.

\subsubsection{Torsional galilean de Sitter superspaces}
\label{sec:super-tdsg}

Torsional galilean de Sitter spacetime is described by $(\k,\h)$ where
$\k$ has the additional brackets $[H,\B] = - \P$ and
$[H,\P] = \gamma \B + (1+ \gamma) \P$, where $\gamma\in(-1,1)$.  There
are two one-parameter family of superisations $(\s,\h)$, with brackets
\begin{equation}
  [H, \sQ(s)] =\sQ(\tfrac12 s (1+\lambda\kk)) \qquad\text{and}\qquad
  [\sQ(s), \sQ(s)] = \tfrac1{1-\gamma} (\gamma \sB(s\kk\sbar) + \sP(s\kk\sbar))
\end{equation}
and
\begin{equation}
  [H, \sQ(s)] =\sQ(\tfrac12 s (\gamma+\lambda\kk)) \qquad\text{and}\qquad
  [\sQ(s), \sQ(s)] = \tfrac1{\gamma-1} (\sB(s\kk\sbar) + \sP(s\kk\sbar))
\end{equation}
for $\lambda \in \RR$.  The associated Lie superalgebras are
\hyperlink{KLSA9}{$\mathsf{S9}$$_{\gamma,\lambda}$} and
\hyperlink{KLSA10}{$\mathsf{S10}$$_{\gamma,\lambda}$}, respectively.

For $\gamma=1$, with additional brackets $[H,\B] = -\P$ and $[H,\P] =
\B + 2 \P$, there is a one-parameter family of superisations, with brackets
\begin{equation}
  [H, \sQ(s)] =\sQ(\tfrac12 s (1+\lambda\kk)) \qquad\text{and}\qquad
  [\sQ(s), \sQ(s)] = \sB(s\kk\sbar) + \sP(s\kk\sbar).
\end{equation}
The associated Lie superalgebras are \hyperlink{KLSA12}{$\mathsf{S12}_{\lambda}$}.

\subsubsection{Galilean anti de Sitter superspace}
\label{sec:super-adsg}

Galilean anti de Sitter spacetime is described by $(\k,\h)$ where $\k$ has the
additional brackets $[H,\B] = -\P$ and $[H,\P] = \B$.  It admits a
superisation $(\s,\h)$, with brackets 
\begin{equation}
  [H, \sQ(s)] =\sQ(\tfrac12 s \jj) \qquad\text{and}\qquad
  [\sQ(s), \sQ(s)] = - \sB(s\ii\sbar) + \sP(s\kk\sbar),
\end{equation}
which corresponds to the Lie superalgebra \hyperlink{KLSA11}{$\mathsf{S11}_{\chi = 0}$}, after changing basis the sign of $\P$.

\subsubsection{Torsional galilean anti de Sitter superspace}
\label{sec:super-tadsg}

Torsional galilean anti de Sitter spacetime is described by $(\k,\h)$ where $\k$ has the
additional brackets $[H,\B] = \chi \B + \P$ and $[H,\P] = \chi \P -
\B$, where $\chi > 0$.  There is a unique superisation $(\s,\h)$, with brackets
\begin{equation}
  [H, \sQ(s)] =\sQ(\tfrac12 s (\chi + \jj)) \qquad\text{and}\qquad
  [\sQ(s), \sQ(s)] = - \sB(s\ii\sbar) - \sP(s\kk\sbar).
\end{equation}
For uniformity, we change basis so that $[H,\B] = -\P$ as for all
galilean spacetimes.  Then the resulting super Lie pair $(\s,\h)$ is
determined by the brackets $[H,\B] = -\P$, $[H,\P] = (1+\chi^2)\B +
2\chi \P$ and, in addition,
\begin{equation}
  [H, \sQ(s)] =\sQ(\tfrac12 s (\chi + \jj)) \qquad\text{and}\qquad
  [\sQ(s), \sQ(s)] = \sB(s\kk(\chi + \jj)\sbar) + \sP(s\kk\sbar),
\end{equation}
corresponding to the Lie superalgebra \hyperlink{KLSA11}{$\mathsf{S11}$$_\chi$}.

\subsubsection{Carrollian superspace}
\label{sec:super-c}

Carrollian spacetime is described by $(\k,\h)$ where $\k$ has the
additional brackets $[\B,\P] = H$.  It admits a superisation
$(\s,\h)$, with brackets
\begin{equation}
  [\sQ(s), \sQ(s)] = |s|^2 H,
\end{equation}
which corresponds to the Lie superalgebra \hyperlink{KLSA13}{$\mathsf{S13}$}.

\subsubsection{Minkowski superspace}
\label{sec:super-m}

Minkowski superspace arises as a superisation of Minkowski
spacetime, described by $(\k,\h)$ with brackets $[H,\B] = -\P$,
$[\B,\P] = H$ and $[\B,\B] = -\J$ and in addition
\begin{equation}
  [\sB(\beta),\sQ(s)] = \sQ(\tfrac12\beta s \kk)
  \qquad\text{and}\qquad
  [\sQ(s),\sQ(s)] = |s|^2 H - \sP(s\kk\sbar).
\end{equation}
This is, of course, the Poincaré superalgebra \hyperlink{KLSA14}{$\mathsf{S14}$}.

\subsubsection{Carrollian anti de Sitter superspace}
\label{sec:super-adsc}

Carrollian anti de Sitter spacetime is described as $(\k,\h)$ where
the $\k$ brackets are given by $[H,\P] = \B$, $[\B,\P] = H$ and
$[\P,\P] = -\J$.  It admits a unique superisation $(\s,\h)$ with
brackets (we have rotated $\kk$ to $\ii$)
\begin{equation}
  [\sP(\pi), \sQ(s)] = \sQ(\tfrac12 \pi s\ii) \qquad\text{and}\qquad
  [\sQ(s),\sQ(s)] = |s|^2 H + \sB(s\ii\sbar).
\end{equation}
We remark that just as with carrollian anti de Sitter and Minkowski
spacetimes, which are both homogeneous spacetimes of the Poincaré
group, their superisations have isomorphic supersymmetry algebras:
namely, the Poincaré superalgebra \hyperlink{KLSA14}{$\mathsf{S14}$}.

\subsubsection{Anti de Sitter superspace}
\label{sec:super-ads}

Anti de Sitter spacetime is described kinematically as $(\k,\h)$ with
brackets
\begin{equation}
  [H,\B] =-\P, \qquad [H,\P] = \B, \qquad [\B,\P] = H, \qquad [\B,\B]
  = -\J \qquad\text{and}\qquad [\P,\P] = -\J.
\end{equation}
It admits a unique superisation $(\s,\h)$, with additional brackets
(where we have rotated $(\ii,\jj,\kk) \mapsto (\kk,\ii,\jj)$ for uniformity)
\begin{gather}
    [H,\sQ(s)] = \sQ(\tfrac12 s \jj), \qquad [\sB(\beta), \sQ(s)] =
    \sQ(\tfrac12 \beta s \kk), \qquad [\sP(\pi),\sQ(s)] = \sQ(\tfrac12
    \pi s\ii) \nonumber \\
    \qquad\text{and}\qquad [\sQ(s),\sQ(s)] = |s|^2 H +
    \sJ(s\jj\sbar) + \sB(s\ii\sbar) - \sP(s\kk\sbar).
\end{gather}
The associated Lie superalgebra is \hyperlink{KLSA15}{$\mathsf{S15}$}, which
is isomorphic to $\osp(1|4)$.

\subsubsection{Super-spacetimes extending $\RR \times S^3$}
\label{sec:super-rxS3}

These correspond to the effective super Lie pairs associated with the
Lie superalgebras \hyperlink{KLSA21}{$\mathsf{S21}$} and
\hyperlink{KLSA22}{$\mathsf{S22}$}.  The super Lie pairs $(\s,\h)$ are
effective, but the underlying Lie pair $(\k,\h)$ is not.  Indeed, the
brackets of $\k$ are now $[\B,\B] = \B$ and $[\P,\P]= \J - \B$, from
where we see that $\B$ spans an ideal of $\k$; although not one of
$\s$, due to the brackets
\begin{equation}
  [\sB(\beta), \sQ(s)] = \sQ(\tfrac12\beta s) \qquad\text{and}\qquad
  [\sQ(s),\sQ(s)] = |s|^2 H,
\end{equation}
for $\s$ the Lie superalgebra \hyperlink{KLSA21}{$\mathsf{S21}$} or
\begin{equation}
  [H, \sQ(s)] = \sQ(\tfrac12 s \kk), \qquad [\sB(\beta), \sQ(s)] =
  \sQ(\tfrac12\beta s) \qquad\text{and}\qquad [\sQ(s),\sQ(s)] = |s|^2
  H - \sB(s\kk\sbar),
\end{equation}
for $\s$ the Lie superalgebra \hyperlink{KLSA22}{$\mathsf{S22}$}.  In both
superspaces, $\B$ do not generate boosts but R-symmetries.  The
underlying spacetime in both cases is the Einstein static universe
$\RR \times S^3$.

\subsubsection{Super-spacetimes extending $\RR \times H^3$}
\label{sec:super-rxH3}

These correspond to the effective super Lie pairs associated with the
Lie superalgebras \hyperlink{KLSA17}{$\mathsf{S17}$} and
\hyperlink{KLSA18}{$\mathsf{S18}$}.  The super Lie pairs $(\s,\h)$ are
effective, but the underlying Lie pair $(\k,\h)$ is not.  Indeed, the
brackets of $\k$ are $[\B,\B] = \B$ and $[\P,\P] = \B - \J$, so that
$\B$ span an ideal $\v\subset \k$.  The resulting aristotelian
spacetime $(\k/\v,\r)$ is the hyperbolic version of the Einstein
static universe \hyperlink{A23m}{$\RR \times H^3$}.

For $\s$ the Lie superalgebra \hyperlink{KLSA17}{$\mathsf{S17}$}, the brackets are
\begin{equation}
  [\sB(\beta), \sQ(s)] = \sQ(\tfrac12\beta s) \qquad\text{and}\qquad
  [\sQ(s),\sQ(s)] = |s|^2 H,
\end{equation}
so that $\B$ does not span an ideal of $\s$.  In other words, $\B$ do
not generate boosts in the underlying homogeneous spacetime, but
rather R-symmetries.

A similar story holds for $\s$ the Lie superalgebra
\hyperlink{KLSA18}{$\mathsf{S18}$}, with the additional brackets
\begin{equation}
  [H, \sQ(s)] = \sQ(\tfrac12 s \kk), \qquad [\sB(\beta), \sQ(s)] =
  \sQ(\tfrac12\beta s) \qquad\text{and}\qquad [\sQ(s),\sQ(s)] = |s|^2
  H - \sB(s\kk\sbar).
\end{equation}
Again, $\B$ are to be interpreted as R-symmetries.

\subsubsection{Super-spacetimes extending the static aristotelian spacetime}
\label{sec:super-S}

This corresponds to the Lie superalgebras \hyperlink{KLSA27}{$\mathsf{S27}$}
and \hyperlink{KLSA28}{$\mathsf{S28}$}.  In either case the resulting super
Lie pair $(\s,\h)$ is effective, but the underlying Lie pair $(\k,\h)$
is not since $[\B,\B] = \B$ spans an ideal of $\k$.  The homogeneous
spacetime associated with the non-effective $(\k,\h)$ is the
aristotelian static spacetime \hyperlink{A21}{$\zS$}.

As in the previous cases, the generators $\B$ do not act as boosts but
rather as R-symmetries, as evinced by the brackets:
\begin{equation}
  [\sB(\beta), \sQ(s)] = \sQ(\tfrac12\beta s) \qquad\text{and}\qquad
  [\sQ(s),\sQ(s)] = |s|^2 H.  
\end{equation}
for $\s$ the Lie superalgebra \hyperlink{KLSA27}{$\mathsf{S27}$}, or
\begin{equation}
  [H, \sQ(s)] = \sQ(\tfrac12 s\kk), \qquad
  [\sB(\beta), \sQ(s)] = \sQ(\tfrac12\beta s) \qquad\text{and}\qquad
  [\sQ(s),\sQ(s)] = |s|^2 H - \sB(s\kk\sbar).
\end{equation}
for $\s$ the Lie superalgebra \hyperlink{KLSA28}{$\mathsf{S28}$}.

\subsection{Aristotelian homogeneous superspaces}
\label{sec:arist-super-lie}

The super Lie pairs $(\s,\h)$ in green in
Table~\ref{tab:super-lie-pairs} are such that the vectorial subspace
$V \subset \h$ is an ideal $\v$ of $\s$.  Quotienting $\s$ by this
ideal yields a Lie superalgebra $\sa \cong \s/\v$ with
$\a = \sa_{\bar 0}$ an aristotelian Lie algebra
(see~\cite[App.~A]{Figueroa-OFarrill:2018ilb} for a classification).
The resulting aristotelian super Lie pair $(\sa,\r)$ is effective by
construction and geometrically realisable.  It is then a simple matter
to identify the aristotelian Lie superalgebra to which each of those
non-effective super Lie pairs in Table~\ref{tab:super-lie-pairs}
leads.  We summarise this in Table~\ref{tab:aristo-correspondence},
which exhibits the correspondence between aristotelian super Lie pairs
in Table~\ref{tab:super-lie-pairs} and aristotelian Lie superalgebras
in Table~\ref{tab:alsa}.  We identify the super Lie pair $(\s,\h)$ by
the label for $\s$ as in Table~\ref{tab:klsa} and the ideal
$\v \subset \h$.

\begin{table}[h!]
  \centering
  \caption{Correspondence between non-effective super Lie pairs and
    aristotelian superalgebras}
  \label{tab:aristo-correspondence}
  \resizebox{\textwidth}{!}{
    \rowcolors{2}{blue!10}{white}
    \begin{tabular}{l|>{$}l<{$}|l}\toprule
      \multicolumn{1}{c|}{$\s$} & \multicolumn{1}{c|}{$\v$} & \multicolumn{1}{c}{$\sa$}\\\midrule
      \hyperlink{KLSA1}{$\mathsf{S1}$} & \B & \hyperlink{ALSA36}{$\mathsf{S36}$} \\
      \hyperlink{KLSA2}{$\mathsf{S2}$} & \B & \hyperlink{ALSA39}{$\mathsf{S39}$} \\
      \hyperlink{KLSA3}{$\mathsf{S3}$} & \B & \hyperlink{ALSA39}{$\mathsf{S39}$} \\
      \hyperlink{KLSA3}{$\mathsf{S3}$} & \P & \hyperlink{ALSA38}{$\mathsf{S38}$} \\
      \hyperlink{KLSA4}{$\mathsf{S4}$} & \B & \hyperlink{ALSA38}{$\mathsf{S38}$} \\
      \hyperlink{KLSA5}{$\mathsf{S5}$} & \B & \hyperlink{ALSA37}{$\mathsf{S37}$} \\
      \hyperlink{KLSA6}{$\mathsf{S6}$} & \B & \hyperlink{ALSA37}{$\mathsf{S37}$}\\
      \hyperlink{KLSA9}{$\mathsf{S9}_{\gamma\in[-1,1),\lambda\in\RR}$} & \B & \hyperlink{ALSA40}{$\mathsf{S40}_\lambda$} \\
      \hyperlink{KLSA9}{$\mathsf{S9}_{\gamma=1,\lambda\in\RR}$} & \B &  \hyperlink{ALSA40}{$\mathsf{S40}_\lambda$} \\
      \bottomrule
    \end{tabular}
    \hspace{1cm}
    \begin{tabular}{l|>{$}l<{$}|l} \toprule
      \multicolumn{1}{c|}{$\s$} & \multicolumn{1}{c|}{$\v$} & \multicolumn{1}{c}{$\sa$}\\\midrule
      \hyperlink{KLSA10}{$\mathsf{S10}_{\gamma\in[-1,0)\cup(0,1),\lambda\in\RR}$} & \P &  \hyperlink{ALSA40}{$\mathsf{S40}_\lambda$} \\
      \hyperlink{KLSA10}{$\mathsf{S10}_{\gamma=0,\lambda\neq 0}$} & \P & \hyperlink{ALSA36}{$\mathsf{S36}$} \\
      \hyperlink{KLSA10}{$\mathsf{S10}_{\gamma=0,\lambda= 0}$} & \P & \hyperlink{ALSA37}{$\mathsf{S37}$} \\
      \hyperlink{KLSA16}{$\mathsf{S16}$} & \B & \hyperlink{ALSA43}{$\mathsf{S43}$} \\
      \hyperlink{KLSA19}{$\mathsf{S19}$} & \B & \hyperlink{ALSA42}{$\mathsf{S42}$} \\
      \hyperlink{KLSA20}{$\mathsf{S20}$} & \B & \hyperlink{ALSA41}{$\mathsf{S41}$} \\
      \hyperlink{KLSA23}{$\mathsf{S23}$} & \B & \hyperlink{ALSA36}{$\mathsf{S36}$} \\
      \hyperlink{KLSA24}{$\mathsf{S24}$} & \B & \hyperlink{ALSA37}{$\mathsf{S37}$} \\
      \bottomrule
    \end{tabular}
    \hspace{1cm}
    \begin{tabular}{l|>{$}l<{$}|l} \toprule
      \multicolumn{1}{c|}{$\s$} & \multicolumn{1}{c|}{$\v$} & \multicolumn{1}{c}{$\sa$}\\\midrule
      \hyperlink{KLSA25}{$\mathsf{S25}$} & \B & \hyperlink{ALSA38}{$\mathsf{S38}$} \\
      \hyperlink{KLSA25}{$\mathsf{S25}$} & \P & \hyperlink{ALSA38}{$\mathsf{S38}$} \\
      \hyperlink{KLSA26}{$\mathsf{S26}$} & \B & \hyperlink{ALSA39}{$\mathsf{S39}$} \\
      \hyperlink{KLSA26}{$\mathsf{S26}$} & \P & \hyperlink{ALSA38}{$\mathsf{S38}$} \\
      \hyperlink{KLSA27}{$\mathsf{S27}$} & \P & \hyperlink{ALSA41}{$\mathsf{S41}$} \\
      \hyperlink{KLSA28}{$\mathsf{S28}$} & \P & \hyperlink{ALSA42}{$\mathsf{S42}$} \\
      \hyperlink{KLSA31}{$\mathsf{S31}$} & \P & \hyperlink{ALSA38}{$\mathsf{S38}$} \\
      \hyperlink{KLSA32}{$\mathsf{S32}$} & \P & \hyperlink{ALSA38}{$\mathsf{S38}$} \\
      \hyperlink{KLSA33}{$\mathsf{S33}_{\lambda\in\RR}$} & \B & \hyperlink{ALSA40}{$\mathsf{S40}_\lambda$} \\
      \bottomrule
    \end{tabular}
  }
\end{table}

\subsection{Summary}
\label{sec:summary-3}

Table~\ref{tab:superspaces} lists the homogeneous superspaces we have
classified in this paper. Each superspacetime is a superisation of an
underlying spatially-isotropic, homogeneous (kinematical or
aristotelian) spacetime, which we list in Table~\ref{tab:spacetimes},
which is borrowed from \cite{Figueroa-OFarrill:2018ilb} (see also
\cite{Figueroa-OFarrill:2019sex}), to which we refer the reader for a
detailed discussion of these spacetimes. Let us recall that
Table~\ref{tab:spacetimes} is divided into five sections,
corresponding to the different invariant structures which the
homogeneous spacetimes admit, as recalled in the introduction.  We
have a similar division of Table~\ref{tab:superspaces}: with the
superisations of spacetimes admitting a lorentzian, galilean,
carrollian, aristotelian (with R-symmetries) and aristotelian (without
R-symmetries) structures, respectively.  All spacetimes admit
superisations with the exception of the riemannian spaces, de Sitter
spacetime ($\hyperlink{S2}{\zdS}_4$) and two of the carrollian
spacetimes: carrollian de sitter ($\hyperlink{S14}{\zdSC}$) and the
carrollian light-cone ($\hyperlink{S16}{\zLC}$).

\begin{table}[h!]
  \centering
  \caption{Simply-connected spatially-isotropic homogeneous superspaces}
  \label{tab:superspaces}
  \setlength{\extrarowheight}{2pt}
  \rowcolors{2}{blue!10}{white}
  \begin{tabular}{l|l|l|l*{4}{|>{$}c<{$}}}\toprule
    \multicolumn{1}{c|}{SM\#} & \multicolumn{1}{c|}{M} & \multicolumn{1}{c|}{$\s$} & \multicolumn{1}{c|}{$\k$ (or $\a$)} & \multicolumn{1}{c|}{$\hh$}& \multicolumn{1}{c|}{$\bb$} & \multicolumn{1}{c|}{$\pp$} & \multicolumn{1}{c}{$[\sQ(s),\sQ(s)]$} \\
    \toprule
    \hypertarget{SM1}{1} & \hyperlink{S1}{$\MM^4$} & \hyperlink{KLSA14}{$\mathsf{S14}$} & \hyperlink{KLA8}{$\mathsf{K8}$} & & \tfrac12 \kk & & |s|^2 H - \sP(s\kk\sbar) \\
    \hypertarget{SM2}{2} & \hyperlink{S3}{$\zAdS_4$} & \hyperlink{KLSA15}{$\mathsf{S15}$} & \hyperlink{KLA11}{$\mathsf{K11}$} & \tfrac12 \jj & \tfrac12 \kk & \tfrac12 \ii & |s|^2 H + \sJ(s\jj\sbar) + \sB(s\ii\sbar) - \sP(s\kk\sbar) \\
    \midrule
    \hypertarget{SM3}{3} & \hyperlink{S7}{$\zG$} & \hyperlink{KLSA7}{$\mathsf{S7}$} & \hyperlink{KLA2}{$\mathsf{K2}$} & \kk & & & -\sP(s\kk\sbar) \\
    \hypertarget{SM4}{4} & \hyperlink{S7}{$\zG$} & \hyperlink{KLSA8}{$\mathsf{S8}$} & \hyperlink{KLA2}{$\mathsf{K2}$} & & & & -\sP(s\kk\sbar)  \\
    \hypertarget{SM5}{5$_{\lambda\in\RR}$} & \hyperlink{S8}{$\zdSG$} & \hyperlink{KLSA9}{$\mathsf{S9}_{-1,\lambda}$}& \hyperlink{KLA3}{$\mathsf{K3}_{-1}$} & \tfrac12 (1 + \lambda \kk) & & & -\tfrac12 (\sB(s\kk\sbar) - \sP(s\kk\sbar)) \\      
    \hypertarget{SM6}{6$_{\lambda\in\RR}$} & \hyperlink{S8}{$\zdSG$} & \hyperlink{KLSA10}{$\mathsf{S10}_{-1,\lambda}$} & \hyperlink{KLA3}{$\mathsf{K3}_{-1}$} & \tfrac12 (-1 + \lambda \kk) & & & -\tfrac12 (\sB(s\kk\sbar) + \sP(s\kk\sbar))  \\
    \hypertarget{SM7}{7$_{\gamma\in(-1,1),\lambda\in\RR}$} & \hyperlink{S9}{$\ztdSG_\gamma$} & \hyperlink{KLSA9}{$\mathsf{S9}_{\gamma,\lambda}$} & \hyperlink{KLA3}{$\mathsf{K3}_\gamma$} & \tfrac12 (1 + \lambda \kk) & & & \tfrac{1}{1-\gamma}(\gamma\sB(s\kk\sbar) + \sP(s\kk\sbar)) \\
    \hypertarget{SM8}{8$_{\gamma\in(-1,1),\lambda\in\RR}$} & \hyperlink{S9}{$\ztdSG_\gamma$} &  \hyperlink{KLSA10}{$\mathsf{S10}_{\gamma,\lambda}$} & \hyperlink{KLA3}{$\mathsf{K3}_\gamma$} & \tfrac12 (\gamma + \lambda \kk) & & & \tfrac1{\gamma-1}(\sB(s\kk\sbar) + \sP(s\kk\sbar)) \\
    \hypertarget{SM9}{9$_{\lambda\in\RR}$} & \hyperlink{S9}{$\ztdSG_{\gamma=1}$} & \hyperlink{KLSA12}{$\mathsf{S12}_\lambda$} & \hyperlink{KLA3}{$\mathsf{K3}_1$} & \tfrac12 (1 + \lambda \kk) & & & \sB(s\kk\sbar) + \sP(s\kk\sbar) \\
    \hypertarget{SM10}{10} & \hyperlink{S10}{$\zAdSG$} & \hyperlink{KLSA11}{$\mathsf{S11}_0$} & \hyperlink{KLA4}{$\mathsf{K4}_0$} & \tfrac12 \jj & & & -\sB(s\ii\sbar) + \sP(s\kk\sbar) \\
    \hypertarget{SM11}{11$_{\chi>0}$} & \hyperlink{S11}{$\ztAdSG_\chi$} & \hyperlink{KLSA11}{$\mathsf{S11}_\chi$} & \hyperlink{KLA4}{$\mathsf{K4}_\chi$} & \tfrac12 (\chi + \jj) & & & \sB(s\kk(\chi + \jj)\sbar) + \sP(s\kk\sbar)  \\
    \midrule
    \hypertarget{SM12}{12} & \hyperlink{S13}{$\zC$} & \hyperlink{KLSA13}{$\mathsf{S13}$} & \hyperlink{KLA6}{$\mathsf{K6}$} & & & & |s|^2 H \\
    \hypertarget{SM13}{13} & \hyperlink{S15}{$\zAdSC$} & \hyperlink{KLSA14}{$\mathsf{S14}$} & \hyperlink{KLA8}{$\mathsf{K8}$} & & & \tfrac12 \ii & |s|^2 H + \sB(s\ii\sbar) \\
    \midrule
    \hypertarget{SM14}{14} & \hyperlink{A23m}{$\RR \times H^3$} & \hyperlink{KLSA17}{$\mathsf{S17}$} & \hyperlink{KLA12}{$\mathsf{K12}$} & & \tfrac12 & & |s|^2 H \\
    \hypertarget{SM15}{15} & \hyperlink{A23m}{$\RR \times H^3$} & \hyperlink{KLSA18}{$\mathsf{S18}$} & \hyperlink{KLA12}{$\mathsf{K12}$} & \tfrac12 \kk & \tfrac12 & & |s|^2 H - \sB(s \kk \sbar) \\
    \hypertarget{SM16}{16} & \hyperlink{A23p}{$\RR \times S^3$} & \hyperlink{KLSA21}{$\mathsf{S21}$} & \hyperlink{KLA13}{$\mathsf{K13}$} & & \tfrac12 & & |s|^2 H \\
    \hypertarget{SM17}{17} & \hyperlink{A23p}{$\RR \times S^3$} & \hyperlink{KLSA22}{$\mathsf{S22}$} & \hyperlink{KLA13}{$\mathsf{K13}$} & \tfrac12 \kk & \tfrac12 & & |s|^2 H - \sB(s\kk\sbar) \\
    \hypertarget{SM18}{18} & \hyperlink{A21}{$\zS$} & \hyperlink{KLSA27}{$\mathsf{S27}$} & \hyperlink{KLA14}{$\mathsf{K14}$} & & \tfrac12 & & |s|^2 H \\
    \hypertarget{SM19}{19} & \hyperlink{A21}{$\zS$} & \hyperlink{KLSA28}{$\mathsf{S28}$} & \hyperlink{KLA14}{$\mathsf{K14}$} & \tfrac12 \kk & \tfrac12 & & |s|^2 H - \sB(s\kk\sbar) \\
    \midrule
    \hypertarget{SM20}{20} &  \hyperlink{A21}{$\zS$} & \hyperlink{ALSA36}{$\mathsf{S36}$} & \hyperlink{ALA1}{$\mathsf{A1}$} & \kk & - & & - \sP(s\kk\sbar) \\
    \hypertarget{SM21}{21} &  \hyperlink{A21}{$\zS$} & \hyperlink{ALSA37}{$\mathsf{S37}$} & \hyperlink{ALA1}{$\mathsf{A1}$} & & -& & - \sP(s\kk\sbar)  \\
    \hypertarget{SM22}{22} &  \hyperlink{A21}{$\zS$} & \hyperlink{ALSA38}{$\mathsf{S38}$} & \hyperlink{ALA1}{$\mathsf{A1}$} & & - & & |s|^2 H  \\
    \hypertarget{SM23}{23} & \hyperlink{A21}{$\zS$} & \hyperlink{ALSA39}{$\mathsf{S39}$} & \hyperlink{ALA1}{$\mathsf{A1}$} & & - & & |s|^2 H - \sP(s\kk\sbar)  \\
    \hypertarget{SM24}{24$_{\lambda\in\RR}$} &  \hyperlink{A22}{$\zTS$} & \hyperlink{ALSA40}{$\mathsf{S40}_\lambda$} & \hyperlink{ALA2}{$\mathsf{A2}$} & \tfrac12(1 + \lambda \kk) & - & & -\sP(s\kk\sbar)  \\
    \hypertarget{SM25}{25} & \hyperlink{A23p}{$\RR\times S^3$} & \hyperlink{ALSA41}{$\mathsf{S41}$} & \hyperlink{ALA3p}{$\mathsf{A3}_+$} & & - & \tfrac12 & |s|^2 H  \\
    \hypertarget{SM26}{26} & \hyperlink{A23p}{$\RR\times S^3$} & \hyperlink{ALSA42}{$\mathsf{S42}$} & \hyperlink{ALA3p}{$\mathsf{A3}_+$} &\kk & - & \tfrac12 & |s|^2 H - \sJ(s\kk\sbar) - \sP(s\kk\sbar)  \\
    \hypertarget{SM27}{27} & \hyperlink{A23m}{$\RR\times H^3$} & \hyperlink{ALSA43}{$\mathsf{S43}$} & \hyperlink{ALA3m}{$\mathsf{A3}_-$} & & - & \tfrac12\ii &  \sJ(s\jj\sbar) - \sP(s\kk\sbar)  \\
    \bottomrule
  \end{tabular}
  \caption*{The first column is our identifier for the superspace,
    whereas the second column is the underlying homogeneous spacetime
    it superises.  The next two columns are the isomorphism classes of
    kinematical Lie superalgebra and kinematical Lie algebra,
    respectively.  The next columns specify the brackets
    of $\s$ not of the form $[\J,-]$ in a basis where $\h$ is spanned
    by $\J$ and $\B$.  As explained in Section~\ref{sec:quat-form},
    supercharges $\sQ(s)$ are parametrised by $s \in \HH$, whereas
    $\sJ(\omega)$, $\sB(\beta)$ and $\sP(\pi)$ are parametrised by
    $\omega,\beta,\pi \in \Im\HH$. The brackets are given by
    $[H,\sQ(s)] = \sQ(s\hh)$, $[\sB(\beta),\sQ(s)]=\sQ(\beta s \bb)$
    and $[\sP(\pi),\sQ(s)] = \sQ(\pi s \pp)$, for some
    $\hh,\bb,\pp\in\HH$. The table is divided into five sections from
    top to bottom: lorentzian, galilean, carrollian, aristotelian with
    R-symmetries and aristotelian.}
\end{table}
 
\subsection{Low-rank invariants}
\label{sec:low-rank-invariants}

In this section, we exhibit the low-rank invariants of the homogeneous
superspaces in Table~\ref{tab:superspaces}, all of which are
reductive.  Indeed, a homogeneous supermanifold with super Lie pair
$(\s,\h)$, where $\h \subset \k = \s_{\bar 0}$, is reductive if and
only if so is the underlying homogeneous manifold $(\k, \h)$.  This is
because if $\k = \h \oplus \m$ is a reductive split, then so is
$\s = \h \oplus (\m \oplus S)$, with $S = \s_{\bar 1}$: the bracket
$[\h,\m] \subset \m$ because $(\k,\h)$ is reductive and the bracket
$[\h, S] \subset S$ because $\h \in \s_{\bar 0}$ and
$S = \s_{\bar 1}$. In \cite{Figueroa-OFarrill:2018ilb} it is shown
that all the homogeneous spacetimes in Table~\ref{tab:spacetimes} are
reductive with the exception of the carrollian light-cone $\zLC$,
which in any case does not admit any $(4|4)$-dimensional
superisation. Hence all the superspaces in Table~\ref{tab:superspaces}
are reductive.

Let $(\s,\h)$ be the super Lie pair associated with one of the
homogeneous superspaces in Table~\ref{tab:superspaces}.  We will write
$\s = \h \oplus \m$, where we have promoted $\m$ to a vector
superspace $\m = \m_{\bar 0} \oplus \m_{\bar 1}$, with $\k = \h
\oplus \m_{\bar 0}$ a reductive split and $\m_{\bar 1} = \s_{\bar 1} =
S$.

Invariant tensors on the simply-connected superspace with super Lie
pair $(\s,\h)$ are in one-to-one correspondence with $\h$-invariant
tensors on $\m$.  Since $\h$ contains the rotational subalgebra $\r
\cong \so(3)$, $\h$-invariant tensors are in particular also
rotationally invariant.  It is not difficult to write down the
rotationally invariant tensors of low order.

As an $\r$-module, $\m = \RR \oplus V \oplus S$, where $\RR$ is the
trivial one-dimensional representation, $V$ is the vector
three-dimensional representation and $S$ is the spinor
four-dimensional representation.  Under the isomorphism $\r = \sp(1) =
\Im \HH$, $\m = \RR \oplus \Im \HH \oplus 
\HH$, where the integrated action of a unit-norm quaternion $u \in
\Sp(1)$ on $(h, p, s) \in \m$ is given by
\begin{equation}
  u \cdot (h, p, s) = (h, u p \bar u, u s).
\end{equation}

Let $H, P_i, Q_a$ denote a basis for $\m$, where $P_i$ and $Q_a$ have
been defined in equation \eqref{eq:quat-basis-s}.  We let $\eta,
\pi^i, \theta^a$ denote the canonically dual basis for $\m^*$.
There is a rotationally invariant line in $\m$: namely, the span of
$H$, which lives in $\m_{\bar 0}$.  Dually, there is a rotationally invariant
line in $\m^*$, which is the span of $\eta$.  These are all the
rotationally invariant tensors of rank $1$.

Let us now consider rank $2$.  As a representation of $\Sp(1)$, $\m
\otimes \m$ has the following invariants.  First of all, we have
$H^2$, which is the only invariant featuring $H$.  Another invariant
is $P^2 := \sum_i P_i \otimes P_i$, which corresponds to the
$\Sp(1)$-invariant inner product  $\left<-,-\right> : \Im\HH \times
\Im\HH \to \RR$ given by $\left<\alpha,\beta\right> =
\Re(\alpha \bar\beta) = - \Re(\alpha\beta)$.  If $q \in \HH$ is any
quaternion, the real bilinear form
\begin{equation}
  \omega_q : \HH \to \HH \to \RR \qquad\text{defined by}\qquad
  \omega_q(s_1,s_2) = \Re(s_1 q \sbar_2)
\end{equation}
is $\Sp(1)$-invariant: symmetric if $q$ is real and symplectic if $q$
is imaginary (and nonzero).  This gives rise to four
$\Sp(1)$-invariants quadratic in $\Q$: $\sum_a Q_a \otimes Q_a$ and
the triplet $\sum_{a,b} I_{ab} Q_a \otimes Q_b$,
$\sum_{a,b} J_{ab} Q_a \otimes Q_b$ and
$\sum_{a,b} K_{ab} Q_a \otimes Q_b$, where $I,J,K$ are the matrices
representing right-multiplication by the quaternions $\ii$, $\jj$,
$\kk$; that is,
\begin{equation}
  \sQ(s\ii) = \sum_{a,b=1}^4 Q_a I_{ab} s_b, \qquad   \sQ(s\jj) =
  \sum_{a,b=1}^4 Q_a J_{ab} s_b \qquad\text{and}\qquad
  \sQ(s\kk) = \sum_{a,b=1}^4 Q_a K_{ab} s_b.
\end{equation}

Similarly there are several rotational invariants in
$\m^* \otimes \m^*$: $\eta^2$ and, in addition, the symmetric tensors
$\pi^2$ and $\theta^2$, and the triplet of symplectic forms
$\omega_I$, $\omega_J$ and $\omega_K$, defined as follows:
\begin{equation}
  \begin{split}
    \pi^2(\sP(\alpha'),\sP(\alpha)) &= \Re(\alpha' \bar\alpha) = - \Re(\alpha'\alpha)\\
    \theta^2(\sQ(s'), \sQ(s)) &= \Re(s'\sbar)\\
    \omega_I(\sQ(s'), \sQ(s)) &= \Re(s'\ii\sbar)\\
    \omega_J(\sQ(s'), \sQ(s)) &= \Re(s'\jj\sbar)\\
    \omega_K(\sQ(s'), \sQ(s)) &= \Re(s'\kk\sbar).
  \end{split}
\end{equation}

To investigate the invariant tensors on $(\s,\h)$ we need to
investigate the action of $\B$ on the tensors.  For the classical
invariants (i.e., those not involving $Q_a$ or $\theta^a$), we may
consult \cite{Figueroa-OFarrill:2018ilb}: the lorentzian metric (and
the corresponding cometric) are invariant for the lorentzian
spacetimes, the clock one-form and spatial cometric for the galilean
spacetimes, the carrollian vector and the spatial metric for the
carrollian spacetimes.  The generators $\B$ act trivially on
aristotelian spacetimes, so the rotationally invariant tensors are the
invariant tensors.  For the invariants involving $Q_a$ or $\theta^a$,
we need to examine how $\B$ acts on $S$.

As can be gleaned from Table~\ref{tab:superspaces}, $\B$ acts
trivially on $\Q$ in most cases.  The exceptions are Minkowski and AdS
superspaces and the aristotelian superspaces where $\B$ acts via
R-symmetries. Hence in all other superspaces, the four rotational
invariants in $\m_{\bar 1} \otimes \m_{\bar 1}$ defined above and
$\theta^2$, $\omega_I$, $\omega_J$ and $\omega_K$ in
$\m^*_{\bar 1} \otimes \m^*_{\bar 1}$ are $\h$-invariant.  This
situation continues to hold for the aristotelian superspaces with
R-symmetry, namely
\hyperlink{SM14}{$\mathsf{SM14}$}--\hyperlink{SM19}{$\mathsf{SM19}$}.
Indeed, one can show that all the rotational invariants which are
quadratic in $\Q$ or in the $\theta^a$ are also R-symmetry invariant.
Indeed, the R-symmetry generator $B_i$ acts on $\m_{\bar 1}$ in the
same way as the infinitesimal rotation generator $J_i$.

Hence it is only for \hyperlink{SM1}{Minkowski} and
\hyperlink{SM2}{$\zAdS$} superspaces that the $\h$-invariants do not agree
with the $\r$-invariants.  For both of these superspaces, $\h \cong
\so(3,1)$, acting in the same way on the spinors:
\begin{equation}
  [\sB(\beta), \sQ(s)] = \sQ(\tfrac12 \beta s \kk).
\end{equation}
It is a simple calculation to see that the following are
$\h$-invariant: 
$\sum_{a,b} I_{ab} Q_a \otimes Q_b$,
$\sum_{a,b} J_{ab} Q_a \otimes Q_b$, $\omega_I$ and $\omega_J$.

Since $\h$ is isomorphic to the Lorentz subalgebra,
we recover the well-known fact that there are two independent
Lorentz-invariant symplectic structures on the Majorana spinors.  This
does not contradict the fact that the Majorana spinor representation
$S$ of $\so(3,1)$ is irreducible as a \emph{real} representation,
since its complexification (the Dirac spinor representation)
decomposes as a direct sum of the two Weyl spinor representations,
each one having a Lorentz-invariant symplectic structure.

\section{Limits between superspaces}
\label{sec:limits-betw-supersp}

In this section, we exhibit some limits between the superspaces in
Table~\ref{tab:superspaces} and interpret them in terms of
contractions of the underlying Lie superalgebras.

As we will show, a limit between two superspaces induces a limit of
the underlying homogeneous spacetimes.  These were determined in
\cite{Figueroa-OFarrill:2018ilb}.  Our discussion will closely follow 
that in \cite[§5]{Figueroa-OFarrill:2018ilb}.  There contractions of a
Lie algebra $\g = (V, \phi)$, where $V$ is a finite-dimensional real
vector space and $\phi: \wedge^2 V \to V$ is a linear map satisfying
the Jacobi identity, were defined as limits of curves in the space of
Lie brackets.  If $g: (0, 1] \to GL(V)$, mapping $t \mapsto g_t$, is a
continuous curve with $g_1 = \id_V$, we can define a curve of
isomorphic Lie algebras $(V,\phi_t)$, where
\begin{equation}
  \phi_t(X,Y) := \left(g^{-1}_t\cdot\phi \right)(X,Y) = g^{-1}_t \left(\phi(g_t X,
  g_t Y)\right).
\end{equation}
If the limit $\phi_0 = \lim_{t\to 0} \phi$ exists, it defines a Lie algebra
$\g_0 = (V, \phi_0)$ which is then a contraction of $\g=(V,\phi_1)$.

In the current case, we will contract Lie superalgebras $\s = (V,
\phi)$, where $V$ is now a real finite-dimensional super vector space
and $\phi: \wedge^2 V \to V$ is a linear map, where
$\wedge^2$ is defined in the super sense, satisfying the super-Jacobi
identity.  We will define contractions of $\s$ in a completely
analogous manner.

\subsection{Contractions of the AdS superalgebra}
\label{sec:ads-limits}

We begin with the superalgebra for the AdS superspace
\hyperlink{SM2}{$\mathsf{SM2}$}, whose generators $\J$, $\B$, $\P$,
$H$ and $\Q$ satisfy the following brackets (in shorthand notation):
\begin{equation}
  \begin{aligned}[m]
    [\J, \J] &= \J \\
    [\J, \B] &= \B \\
    [\J, \P] &= \P \\
    [\J, \Q] &= \Q
  \end{aligned}
  \qquad\qquad
  \begin{aligned}[m]
    [H, \B] &= -\P \\
    [H, \P] &= \B \\
    [\B, \P] &= H \\
    [\B, \B] &= -\J \\
    [\P, \P] &= -\J
  \end{aligned} \qquad\qquad
  \begin{aligned}[m]
    [H, \Q] &= \Q \\
    [\B, \Q] &= \Q \\
    [\P, \Q] &= \Q \\
    [\Q, \Q] &= H + \J + \B - \P.
  \end{aligned}
\end{equation}
Consider the following three-parameter family of
linear transformations $g_{\kappa, c, \tau}$ defined by
\begin{equation}
g_{\kappa, c, \tau}\cdot \J = \J, \qquad g_{\kappa, c, \tau}
\cdot\B = \tfrac{\tau}{c} \B, \qquad g_{\kappa, c, \tau}\cdot
\P = \tfrac{\kappa}{c} \P, \qquad g_{\kappa, c, \tau}\cdot H =
\tau\kappa H, \qquad g_{\kappa, c, \tau}\cdot\Q = \tfrac{\kappa\tau}{c}\Q. 
\end{equation}
The action on the even generators is as in
\cite[§5]{Figueroa-OFarrill:2018ilb} and the action on $\Q$ is chosen to ensure that the
bracket $[\Q, \Q]$ has well-defined limits as $\kappa \to 0$, $c \to
\infty$ or $\tau\to 0$.

The brackets involving $\J$ remain unchanged for 
the above transformations and the remaining brackets
become
\begin{equation}
  \begin{aligned}[m]
    [H, \B] &= -\tau^2 \P \\
    [H, \P] &= \kappa^2 \B \\
    [\B, \P] &= \tfrac{1}{c^2} H
  \end{aligned} \qquad\qquad
  \begin{aligned}[m]
    [\B, \B] &= -\tfrac{\tau^2}{c^2}\J \\
    [\P, \P] &= -\tfrac{\kappa^2}{c^2}\J\\
    [H, \Q] &= \kappa\tau\Q
  \end{aligned} \qquad\qquad
  \begin{aligned}[m]
    [\B, \Q] &= \tfrac{\tau}{c}\Q \\
    [\P, \Q] &= \tfrac{\kappa}{c}\Q \\
    [\Q, \Q] &= \tfrac{1}{c} H + \tfrac{\kappa\tau}{c}\J + \kappa\B - \tau\P.
  \end{aligned}
\end{equation}
We now want to take the limits $\kappa \to 0$,
$c \to \infty$, and $\tau \to 0$ in turn,
corresponding to the flat, non-relativistic, and ultra-relativistic
limits, respectively. Notice that the limits of the brackets between
the even generators will produce the same Lie algebra contractions as
in \cite{Figueroa-OFarrill:2018ilb}. Thus we cannot have a limit from
one superspace to another unless there exists a limit between their
underlying homogeneous spacetimes.

Taking the flat limit $\kappa \to 0$, we are left with
\begin{equation}
[H, \B] = -\tau^2 \P, \quad [\B, \P] = \tfrac{1}{c^2} H, 
\quad [\B, \B] = -\tfrac{\tau^2}{c^2}\J, \quad
[\B, \Q] = \tfrac{\tau}{c}\Q \quad \text{and} \quad
 [\Q, \Q] = \tfrac{1}{c} H - \tau\P.
\end{equation}
For $\tfrac{\tau}{c}\neq 0$, this is the Poincaré superalgebra
(\hyperlink{KLSA14}{$\mathsf{S14}$}).  Thus, we obtain the limit
$\hyperlink{SM2}{\mathsf{SM2}} \to \hyperlink{SM1}{\mathsf{SM1}}$.
Subsequently taking the non-relativistic limit $c \to \infty$,
the brackets reduce to
\begin{equation}
[H, \B] = -\tau^2 \P  \qquad\text{and}\qquad [\Q, \Q] = - \tau\P.
\end{equation}
For $\tau \neq 0$, this shows us that we have the limit
$\hyperlink{SM1}{\mathsf{SM1}} \to
\hyperlink{SM4}{\mathsf{SM4}}$.

Alternatively, we could have taken the ultra-relativistic limit 
$\tau \to 0$, which, for $c \neq 0$, gives us the
Carroll superalgebra (\hyperlink{KLSA13}{$\mathsf{S13}$}):
\begin{equation}
[\B, \P] = \tfrac{1}{c^2} H  \qquad\text{and}\qquad [\Q, \Q] =
\tfrac{1}{c} H.
\end{equation}
Thus, we have $\hyperlink{SM1}{\mathsf{SM1}}\to
\hyperlink{SM12}{\mathsf{SM12}}$.

Returning to the $\zAdS$ superalgebra
(\hyperlink{KLSA15}{$\mathsf{S15}$}) and taking the non-relativistic
limit $c \to \infty$, we find
\begin{equation}
[H, \B] = -\tau^2 \P, \qquad [H, \P] = \kappa^2 \B,
\qquad [H, \Q] = \kappa\tau\Q \qquad\text{and}\qquad
[\Q, \Q] = \kappa\B - \tau\P.
\end{equation}
For $\tau\kappa \neq 0$, this is $\hyperlink{KLSA11}{\mathsf{S11}_0}$
(under a suitable basis change).  Therefore, we have 
$\hyperlink{SM2}{\mathsf{SM2}} \to \hyperlink{SM10}{\mathsf{SM10}}$.
Because these limits commute, we may now take
the flat limit to arrive at \hyperlink{SM4}{$\mathsf{SM4}$}.  

Finally, we may take the ultra-relativistic limit of
$\zAdS$ (\hyperlink{KLSA15}{$\mathsf{S15}$}).  This limit leaves the
brackets
\begin{equation}
[H, \P] = \kappa^2 \B, \quad [\B, \P] = \tfrac{1}{c^2} H,
\quad [\P, \P] = -\tfrac{\kappa^2}{c^2}\J, \quad
[\P, \Q] = \tfrac{\kappa}{c}\Q \quad\text{and}\quad
[\Q, \Q] = \tfrac{1}{c} H + \kappa\B,
\end{equation}
for $\tfrac{\kappa}{c} \neq 0$. Thus, we arrive at
\hyperlink{SM13}{$\mathsf{SM13}$}. Subsequently taking the flat limit,
we find \hyperlink{SM12}{$\mathsf{SM12}$}, as expected.

We can also take limits from the superspaces discussed above to
non-effective super Lie pairs, which will have associated 
aristotelian superspaces.  Since all of the above superspaces
have either \hyperlink{SM4}{$\mathsf{SM4}$} or \hyperlink{SM12}{$\mathsf{SM12}$}
as a limit, we will only show the limits to aristotelian 
superspaces coming form these two cases.  Beginning 
with \hyperlink{SM4}{$\mathsf{SM4}$}, we can use the 
transformation
\begin{equation}
g_t\cdot\B = t\B, \qquad g_t\cdot H = H, \qquad g_t\cdot \P = \P
\qquad\text{and}\qquad g_t\cdot\Q = \Q
\end{equation}
and the limit $t\to 0$ to obtain \hyperlink{SM21}{$\mathsf{SM21}$}.
Using the same transformation and limit, we can also start with 
\hyperlink{SM12}{$\mathsf{SM12}$} and find \hyperlink{SM22}{$\mathsf{SM22}$}.

\subsection{Remaining galilean superspaces}
\label{sec:lim-galilean}

We have shown that we obtain the other lorentzian and 
two carrollian superspaces as limits of the $\zAdS$ superspace
\hyperlink{SM2}{$\mathsf{SM2}$}: namely, Minkowski
(\hyperlink{SM1}{$\mathsf{SM1}$}), Carroll
(\hyperlink{SM12}{$\mathsf{SM12}$}) and carrollian anti de Sitter
(\hyperlink{SM13}{$\mathsf{SM13}$}) superspaces.  In addition, we also
obtain two superisations of galilean spacetimes: a superisation
\hyperlink{SM4}{$\mathsf{SM4}$} of the flat galilean spacetime and the
superisation \hyperlink{SM10}{$\mathsf{SM10}$} of galilean anti de
Sitter spacetime.  But what about the superisations of other galilean
spacetimes?

\subsubsection{Flat galilean superspaces}
\label{sec:lim-g}

From \hyperlink{SM2}{$\mathsf{SM2}$} we obtained the galilean
superspace \hyperlink{SM4}{$\mathsf{SM4}$}.  There is a second
superisation \hyperlink{SM3}{$\mathsf{SM3}$} of the flat galilean
homogeneous spacetime, from which we can also reach
\hyperlink{SM4}{$\mathsf{SM4}$}.  Indeed, using the transformations
\begin{equation}
  g_t\cdot\B = t\B, \qquad g_t\cdot H = t H, 
  \qquad g_t\cdot \P = t\P \qquad\text{and}\qquad g_t\cdot \Q =
  \sqrt{t} \Q,
\end{equation}
on the Lie superalgebra for \hyperlink{SM3}{$\mathsf{SM3}$}, and taking the
limit $t\to 0$, we find the Lie superalgebra for 
\hyperlink{SM4}{$\mathsf{SM4}$}.  Thus, we have $\hyperlink{SM3}{\mathsf{SM3}}
\to \hyperlink{SM4}{\mathsf{SM4}}$.

Beginning with \hyperlink{SM3}{$\mathsf{SM3}$}, we may also 
consider the transformation
\begin{equation}
  g_t\cdot\B = t\B, \qquad g_t\cdot H =  H, \qquad g_t\cdot \P = t\P
  \qquad\text{and}\qquad g_t\cdot \Q = \sqrt{t} \Q,
\end{equation}
and the limit $t\to 0$.  This procedure will give us a non-effective
super Lie pair corresponding to \hyperlink{SM20}{$\mathsf{SM20}$}.

\subsubsection{Galilean de Sitter superspaces}
\label{sec:lim-dsg}

The superspaces \hyperlink{SM5}{$\mathsf{SM5}_\lambda$} and
\hyperlink{SM6}{$\mathsf{SM6}_\lambda$} arise as the $\gamma \to -1$
limit of \hyperlink{SM7}{$\mathsf{SM7}_{\gamma, \lambda}$} and
\hyperlink{SM8}{$\mathsf{SM8}_{\gamma, \lambda}$}, respectively.  This
fact has already been noted in Section~\ref{sec:super-dsg}.
Section~\ref{sec:super-tdsg} demonstrated that
\hyperlink{SM9}{$\mathsf{SM9}_\lambda$} is the $\gamma \to 1$ limit of
\hyperlink{SM7}{$\mathsf{SM7}_{\gamma, \lambda}$} and
\hyperlink{SM8}{$\mathsf{SM8}_{\gamma, \lambda}$}.

The superalgebras associated with these five superspaces take the 
general form
\begin{equation} \label{eq:galilean-de-sitter-brackets}
  \begin{aligned}[m]
    [H, \sB(\beta)] &= - \sP(\beta) \\
    [H, \sP(\pi)] &= \gamma \sB(\pi) + (1+\gamma) \sP(\pi)
  \end{aligned} \qquad
  \begin{aligned}[m]
    [H, \sQ(s)] &= \tfrac{1}{2} \sQ(s(\eta + \lambda\kk)) \\
    [\sQ(s), \sQ(s)] &= \rho \sB(s\kk\sbar) + \sigma \sP(s\kk\sbar)
  \end{aligned} 
\end{equation}
for some $\eta, \rho, \sigma \in \mathbb{R}$, where $\gamma \in
[-1,1]$ and $\lambda \in \RR$ are the parameters of the Lie
superalgebras.  Using the transformations
\begin{equation}
g_t\cdot \B = \B, \qquad g_t\cdot H = tH, \qquad g_t\cdot \P = t \P
\qquad\text{and}\qquad g_t\cdot \Q = \sqrt{\omega t} \Q,
\end{equation}
where $\omega \in \mathbb{R}$, and taking the limit $t \to 0$, 
the above brackets become
\begin{equation}
[H, \sB(\beta)] = -\sP(\beta) \qquad\text{and}\qquad [\sQ(s), \sQ(s)]
= \omega\sigma \sP(s\kk\sbar).
\end{equation}
Therefore, by choosing $\omega = -\sigma^{-1}$, we can always
recover \hyperlink{SM4}{$\mathsf{SM4}$}.

There is a second superisation of the flat galilean homogeneous
spacetime, namely \hyperlink{SM3}{$\mathsf{SM3}$}. There does not
seem to be any Lie-superalgebra contraction that gives
\hyperlink{SM3}{$\mathsf{SM3}$}, but as we will see below, there are 
non-contracting limits (involving taking $\lambda \to \pm \infty$)
which take the superspaces \hyperlink{SM5}{$\mathsf{SM5}_\lambda$},
\hyperlink{SM6}{$\mathsf{SM6}_\lambda$},
\hyperlink{SM7}{$\mathsf{SM7}_{\gamma, \lambda}$},
\hyperlink{SM8}{$\mathsf{SM8}_{\gamma, \lambda}$} and
\hyperlink{SM9}{$\mathsf{SM9}_\lambda$} to
\hyperlink{SM3}{$\mathsf{SM3}$}.

\subsubsection{Galilean anti de Sitter superspaces}

The superspace \hyperlink{SM10}{$\mathsf{SM10}$} is, by definition, 
the $\chi \to 0$ limit of \hyperlink{SM11}{$\mathsf{SM11}_\chi$}.
These algebras take the form
\begin{equation}
  \begin{aligned}[m]
    [H, \sB(\beta)] &= - \sP(\beta) \\
    [H, \sP(\pi)] &= (1+\chi^2) \sB(\pi) + \chi \sP(\pi)
  \end{aligned} \qquad\qquad
  \begin{aligned}[m]
    [H, \sQ(s)] &= \tfrac{1}{2} \sQ(s(\chi + \jj)) \\
    [\sQ(s), \sQ(s)] &= - \sB(s\ii\sbar) - \sP(s\kk\sbar),
  \end{aligned}
\end{equation}
where $\chi \geq 0$ is the parameter of the Lie superalgebra.  Using
the same transformations as in the galilean de Sitter case, but with
$\omega = 1$, we find
\begin{equation}
  [H, \sB(\beta)] = \sP(\beta) \qquad\text{and}\qquad [\sQ(s), \sQ(s)] = -
  \sP(s\kk\sbar).
\end{equation}
Thus, we find \hyperlink{SM4}{$\mathsf{SM4}$} as a limit of both
\hyperlink{SM10}{$\mathsf{SM10}$} and \hyperlink{SM11}{$\mathsf{SM11}_\chi$}.

We cannot obtain \hyperlink{SM3}{$\mathsf{SM3}$} as a limit
of these superspaces as \hyperlink{SM3}{$\mathsf{SM3}$} has collinear
$\hh$ and $\bc_3$, whereas \hyperlink{SM10}{$\mathsf{SM10}$} and 
\hyperlink{SM11}{$\mathsf{SM11}_\chi$} have orthogonal $\hh$ and $\bc_3$.

\subsubsection{Non-contracting limits}
\label{sec:gal-non-contracting-lim}

In \cite{Figueroa-OFarrill:2018ilb} it was shown that
$\lim_{\chi\to\infty} \hyperlink{S11}{\zAdSG_\chi} =
\hyperlink{S9}{\ztdSG_1}$, but this limit is not induced by a Lie algebra
contraction since the Lie algebras are non-isomorphic for different
values of $\chi$.  Does this limit extend to the superspaces?

Beginning with \hyperlink{SM11}{$\mathsf{SM11}_\chi$}, change basis such
that
\begin{equation}
  H' = \chi^{-1} H, \qquad \B' = \B, \qquad \P' = \chi^{-1} \P \qquad
  \text{and} \qquad
  \Q' = \chi^{-1/2} \Q,
\end{equation}
under which the brackets become
\begin{equation}
  \begin{aligned}[m]
    [H', \sB'(\beta)] &= -\sP'(\beta) \\
    [H', \sP'(\pi)] &= 2 \sP'(\pi) + (1+\chi^{-2})\sB'(\pi)
  \end{aligned} \qquad\qquad
  \begin{aligned}[m]
    [H', \sQ'(s)] &= \tfrac{1}{2\chi} \sQ'(s(\chi + \jj)) \\
    [\sQ'(s), \sQ'(s)] &= - \chi^{-1} \sB'(s\ii\sbar) + \sB'(s\kk\sbar) +
    \sP(s\kk\sbar).
  \end{aligned}
\end{equation}
Taking the limit $\chi\to\infty$, we find
\begin{equation}
  \begin{aligned}[m]
    [H', \sB'(\beta)] &= -\sP'(\beta) \\
    [H', \sP'(\pi)] &= 2 \sP'(\pi) + \sB'(\pi)
  \end{aligned} \qquad\qquad
  \begin{aligned}[m]
    [H', \sQ'(s)] &= \tfrac{1}{2} \sQ'(s) \\
    [\sQ'(s), \sQ'(s)] &= - \sB'(s\kk\sbar) +
    \sP(s\kk\sbar).
  \end{aligned}
\end{equation}
This Lie superalgebra is precisely that for \hyperlink{SM9}{$\mathsf{SM9_0}$}.
Thus, we inherit this limit from the underlying homogeneous 
spacetimes.

The superspaces \hyperlink{SM5}{$\mathsf{SM5}_\lambda$},
\hyperlink{SM6}{$\mathsf{SM6}_\lambda$},
\hyperlink{SM7}{$\mathsf{SM7}_{\gamma, \lambda}$},
\hyperlink{SM8}{$\mathsf{SM8}_{\gamma, \lambda}$} and
\hyperlink{SM9}{$\mathsf{SM9}_\lambda$} all have an additional
parameter $\lambda$ and we can ask what happens if we take the limit
$\lambda \to \pm \infty$ in these cases.  This is again a
non-contracting limit, since the Lie superalgebras with different
values of $\lambda \in \RR$ are not isomorphic.

Using the general form of the brackets stated
in~\eqref{eq:galilean-de-sitter-brackets} above, consider a change of
basis
\begin{equation}
  \B'=\B, \qquad H' = 2 \lambda^{-1} H, \qquad \P' = 2 \lambda^{-1}
  \P \qquad\text{and}\qquad \Q' = \lambda^{-\tfrac{1}{2}} \Q.
\end{equation}
In our new basis, the brackets become
\begin{equation} 
  \begin{aligned}[m]
    [H', \sB'(\beta)] &= - \sP'(\beta) \\
    [H', \sP'(\pi)] &= 4\lambda^{-2}\gamma\sB'(\pi) + 2\lambda^{-1} (1+\gamma) \sP'(\pi)
  \end{aligned} \qquad
  \begin{aligned}[m]
    [H', \sQ'(s)] &= \sQ'(s(\lambda^{-1}\eta + \kk)) \\
    [\sQ'(s), \sQ'(s)] &= \lambda^{-1} \rho \sB'(s\kk\sbar) + \tfrac{\sigma}{2} \sP'(s\kk\sbar).
  \end{aligned} 
\end{equation}
Taking either $\lambda \to \infty$ or $\lambda \to -\infty$, we find
\begin{equation}
[H', \sB'(\beta)] = - \sP'(\beta), \qquad [H', \sQ'(s)] = \sQ'(s\kk), 
\qquad [\sQ'(s), \sQ'(s)] = \tfrac{\sigma}{2} \sP'(s\kk\sbar).
\end{equation}
Rescaling both $\B'$ and $\P'$ by $\tfrac{\sigma}{2}$, we recover the Lie
superalgebra for \hyperlink{SM3}{$\mathsf{SM3}$}.

Figure~\ref{fig:super-limits} below illustrates the different
superspaces and the limits between them.  The families
\hyperlink{SM5}{$\mathsf{SM5}_\lambda$},
\hyperlink{SM6}{$\mathsf{SM6}_\lambda$},
\hyperlink{SM7}{$\mathsf{SM7}_{\gamma, \lambda}$},
\hyperlink{SM8}{$\mathsf{SM8}_{\gamma, \lambda}$} and
\hyperlink{SM9}{$\mathsf{SM9}_\lambda$} fit together into a
two-dimensional space which also includes
\hyperlink{SM3}{$\mathsf{SM3}$} as their common limits
$\lambda \to \pm \infty$ and which can be described as follows.  If we
fix $\lambda \in \RR$, then
\begin{equation}
  \lim_{\gamma\to 1} \hyperlink{SM7}{\mathsf{SM7}_{\gamma, \lambda}}
  =\hyperlink{SM9}{\mathsf{SM9}_\lambda} \qquad\text{whereas}\qquad
  \lim_{\gamma \to -1} \hyperlink{SM7}{\mathsf{SM7}_{\gamma,
      \lambda}} = \hyperlink{SM5}{\mathsf{SM5}_\lambda}. 
\end{equation}
Similarly, again fixing $\lambda \in \RR$, we have
\begin{equation}
  \lim_{\gamma\to 1} \hyperlink{SM8}{\mathsf{SM8}_{\gamma, \lambda}}
  =\hyperlink{SM9}{\mathsf{SM9}_\lambda} \qquad\text{whereas}\qquad
  \lim_{\gamma \to -1} \hyperlink{SM8}{\mathsf{SM8}_{\gamma,
      \lambda}} = \hyperlink{SM6}{\mathsf{SM6}_\lambda}.
\end{equation}
This gives rise to the following two-dimensional parameter spaces:
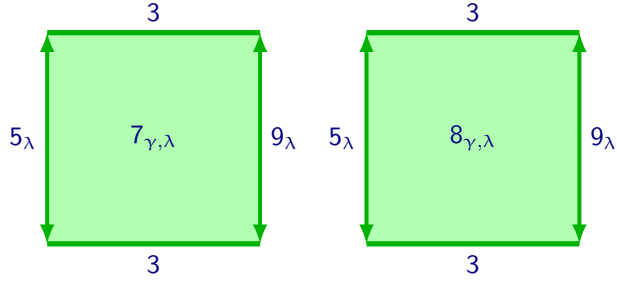
\begin{center}
  \begin{tikzpicture}[>=latex, x=1.0cm,y=1.0cm,scale=0.7]
    %
    %
    %
    %
    \coordinate (bl1) at (-2,0);
    \coordinate (br1) at (2,0);
    \coordinate (tl1) at (-2,4);
    \coordinate (tr1) at (2,4);
    \coordinate (bl2) at (4,0);
    \coordinate (br2) at (8,0);
    \coordinate (tl2) at (4,4);
    \coordinate (tr2) at (8,4);
    %
    %
    \fill [color=green!30!white] (bl1) -- (tl1) -- (tr1) -- (br1) -- (bl1);
    \fill [color=green!30!white] (bl2) -- (tl2) -- (tr2) -- (br2) -- (bl2);
    %
    %
    \node at (0,2) {$\hyperlink{SM7}{\mathsf{7}_{\gamma,\lambda}}$};
    \node at (6,2) {$\hyperlink{SM8}{\mathsf{8}_{\gamma,\lambda}}$};
    %
    %
    \draw [<->,line width=1.5pt,color=green!70!black] (bl1) -- (tl1) node [midway,left] {$\hyperlink{SM5}{\mathsf{5}_\lambda}$};
    \draw [<->,line width=1.5pt,color=green!70!black] (br1) -- (tr1) node [midway,right] {$\hyperlink{SM9}{\mathsf{9}_\lambda}$};
    \draw [<->,line width=1.5pt,color=green!70!black] (bl2) -- (tl2) node [midway,left] {$\hyperlink{SM6}{\mathsf{5}_\lambda}$};
    \draw [<->,line width=1.5pt, color=green!70!black] (br2) -- (tr2) node [midway,right] {$\hyperlink{SM9}{\mathsf{9}_\lambda}$};
    \draw [-, line width=2pt, color=green!70!black] (bl1) -- (br1) node [midway,below] {$\hyperlink{SM3}{\mathsf{3}}$}; 
    \draw [-, line width=2pt, color=green!70!black] (bl2) -- (br2) node [midway,below] {$\hyperlink{SM3}{\mathsf{3}}$}; 
    \draw [-, line width=2pt, color=green!70!black] (tl1) -- (tr1) node [midway,above] {$\hyperlink{SM3}{\mathsf{3}}$}; 
    \draw [-, line width=2pt, color=green!70!black] (tl2) -- (tr2) node [midway,above] {$\hyperlink{SM3}{\mathsf{3}}$}; 
  \end{tikzpicture}
\end{center}

We then flip the square on the right horizontally and glue the two
squares along their common $\hyperlink{SM9}{\mathsf{9}_\lambda}$ edge
to obtain the following picture
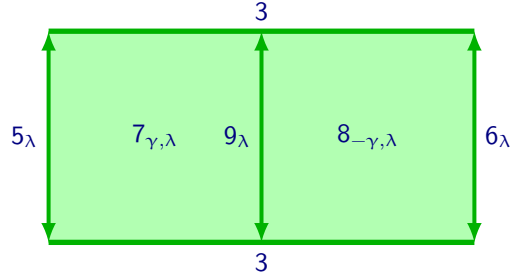
\begin{center}
  \begin{tikzpicture}[>=latex,  x=1.0cm,y=1.0cm,scale=0.7]
    %
    %
    %
    %
    \coordinate (bl) at (-2,0);
    \coordinate (br) at (6,0);
    \coordinate (bm) at (2,0);    
    \coordinate (tl) at (-2,4);
    \coordinate (tr) at (6,4);
    \coordinate (tm) at (2,4);
    %
    %
    \fill [color=green!30!white] (bl) -- (tl) -- (tr) -- (br) -- (bl);
    %
    %
    \node at (0,2) {$\hyperlink{SM7}{\mathsf{7}_{\gamma,\lambda}}$};
    \node at (4,2) {$\hyperlink{SM8}{\mathsf{8}_{-\gamma,\lambda}}$};
    %
    %
    \draw [<->,line width=1.5pt,color=green!70!black] (bl) -- (tl) node [midway,left] {$\hyperlink{SM5}{\mathsf{5}_\lambda}$};
    \draw [<->,line width=1.5pt,color=green!70!black] (br) -- (tr) node [midway,right] {$\hyperlink{SM6}{\mathsf{6}_\lambda}$};
    \draw [<->,line width=1.5pt,color=green!70!black] (bm) -- (tm) node [midway,left] {$\hyperlink{SM9}{\mathsf{9}_\lambda}$};
    \draw [-, line width=2pt, color=green!70!black] (bl) -- (br) node [midway,below] {$\hyperlink{SM3}{\mathsf{3}}$}; 
    \draw [-, line width=2pt, color=green!70!black] (tl) -- (tr) node [midway,above] {$\hyperlink{SM3}{\mathsf{3}}$}; 
  \end{tikzpicture}
\end{center}

We now glue the top and bottom edges to arrive at the following
cylinder:

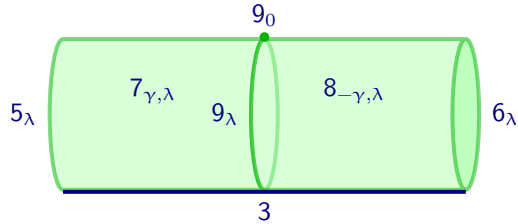
\begin{center}
  \begin{tikzpicture}[>=stealth, aspect=1.5,x=1.0cm,y=1.0cm,scale=0.7,color=green!70!black,line width=1.5pt]
    %
    %
    %
    %
    \node [name=cyl1, draw, cylinder, cylinder uses custom fill, 
    cylinder body fill=green!30!white, minimum height=3cm, minimum
    width=2cm,opacity=0.5] {};
    \node [name=cyl2, draw, cylinder, cylinder uses custom fill, cylinder end fill=green!50!white,
    cylinder body fill=green!30!white, minimum height=3cm, minimum width=2cm,above=0pt of cyl1.before top, anchor=after bottom,opacity=0.5] {};
    %
    %
    \draw [color=blue!50!black, line width=1.5pt] (cyl1.before bottom)
    -- (cyl2.after top) node [midway, below]{\hyperlink{SM3}{$\mathsf{3}$}};
    %
    %
    \coordinate [label=left:{\hyperlink{SM5}{$\mathsf{5}_\lambda$}}] (5) at (cyl1.bottom);
    \coordinate [label=right:{\hyperlink{SM6}{$\mathsf{6}_\lambda$}}] (6) at (cyl2.top);
    \coordinate [label=above:{\hyperlink{SM9}{$\mathsf{9}_0$}}] (90) at (cyl1.before top); 
    \coordinate [label=left:{\hyperlink{SM9}{$\mathsf{9}_\lambda$}}] (9) at (cyl2.bottom); 
    \coordinate [label=:{$\hyperlink{SM7}{\mathsf{7}_{\gamma,\lambda}}$}] (7) at (cyl1.center);
    \coordinate [label=:{$\hyperlink{SM8}{\mathsf{8}_{-\gamma,\lambda}}$}] (8) at (cyl2.center);    
    %
    %
    \foreach \point in {90}
    \filldraw [color=green!70!black,fill=green!70!black] (\point) circle (1.5pt);
  \end{tikzpicture}
\end{center}

Finally, we collapse the ``edge'' labelled
\hyperlink{SM3}{$\mathsf{3}$} to a point, arriving at the object in
Figure~\ref{fig:super-limits}.

\subsection{Aristotelian limits}
\label{sec:aristo-lim}

There are two kinds of superisations of aristotelian spacetimes: the
ones where $\B$ acts as R-symmetries and the ones where $\B$ acts
trivially.  We treat them in turn.

\subsubsection{Aristotelian superspaces with R-symmetry}
\label{sec:r-sym-lim}
 
The homogeneous spacetimes \hyperlink{A23m}{$\RR\times H^3$} and \hyperlink{A23p}{$\RR\times S^3$} 
underlying the homogeneous superspaces \hyperlink{SM14}{$\mathsf{SM14}$}
-  \hyperlink{SM17}{$\mathsf{SM17}$} have \hyperlink{A21}{$\zS$} as their limit.  
Therefore, we could expect \hyperlink{SM14}{$\mathsf{SM14}$} - 
\hyperlink{SM17}{$\mathsf{SM17}$} to have either \hyperlink{SM18}{$\mathsf{SM18}$} or
\hyperlink{SM19}{$\mathsf{SM19}$} as limits.  The relevant contraction uses
the transformation
\begin{equation}
g_t\cdot\B = \B, \qquad g_t\cdot H = H \qquad\text{and} \qquad
g_t\cdot\P = t\P.
\end{equation}
Taking the limit $t \to 0$, the $[\P, \P]$ bracket vanishes
leaving all other brackets unchanged. Thus, we find
$\hyperlink{SM14}{\mathsf{SM14}}
\to\hyperlink{SM18}{\mathsf{SM18}}$,
$\hyperlink{SM16}{\mathsf{SM16}} \to
\hyperlink{SM18}{\mathsf{SM18}}$, $\hyperlink{SM15}{\mathsf{SM15}}
\to \hyperlink{SM19}{\mathsf{SM19}}$ and
$\hyperlink{SM17}{\mathsf{SM17}} \to \hyperlink{SM19}{\mathsf{SM19}}$.

Taking into account the form of $\hh$, and the $[\Q, \Q]$ bracket
for each of these superspaces, we notice that each homogeneous
spacetime has two superspaces associated with it.  One for which
\begin{equation}
\bb = \tfrac{1}{2} \qquad\text{and}\qquad [\sQ(s), \sQ(s)] = |s|^2 H,
\end{equation}
and one for which
\begin{equation}
\bb = \tfrac{1}{2}, \qquad \hh = \tfrac{1}{2}\kk \qquad\text{and}\qquad
[\sQ(s), \sQ(s)] = |s|^2 H - \sB(s\kk\sbar).
\end{equation}
Using transformations which act as
\begin{equation}
g_t\cdot H = tH, \qquad g_t\cdot\Q = \sqrt{t}\Q
\end{equation}
and trivially on $\J, \B,$ and $\P$, we find the brackets of the 
latter superspaces described by
\begin{equation}
\bb = \tfrac{1}{2}, \qquad \hh = \tfrac{t}{2} \kk, \qquad \text{and}
\qquad [\sQ(s), \sQ(s)] = |s|^2 H - t\sB(s\kk\sbar).
\end{equation}
Therefore, taking the limit $t\to 0$, we find the former
superspaces.  Thus, we get the limits $\hyperlink{SM15}{\mathsf{SM15}}
\to \hyperlink{SM14}{\mathsf{SM14}}$, $\hyperlink{SM17}{\mathsf{SM17}}
\to\hyperlink{SM16}{\mathsf{SM16}}$ and $\hyperlink{SM19}{\mathsf{SM19}}
\to\hyperlink{SM18}{\mathsf{SM18}}$.

All of the above superspaces have \hyperlink{SM18}{$\mathsf{SM18}$} 
as a limit.  Therefore, we will only consider the limits of this 
superspace to those aristotelian superspaces without R-symmetry.
Letting
\begin{equation}
g_t\cdot\B = t\B, \qquad  g_t\cdot H = H, \qquad g_t\cdot\P = \P,
\qquad g_t\cdot\Q = \Q,
\end{equation}
and taking the limit $t\to 0$, we arrive at a non-effective super 
Lie pair corresponding to \hyperlink{SM22}{$\mathsf{SM22}$}.

 \subsubsection{Aristotelian superspaces without R-symmetry}
\label{sec:w-o-r-sym-lim}

The aristotelian homogeneous spacetimes \hyperlink{A23p}{$\RR\times S^3$},
\hyperlink{A23m}{$\RR\times H^3$}, and \hyperlink{A22}{$\zTS$} have \hyperlink{A21}{$\zS$}
as their limit; therefore, we would expect their superisations to have 
have one or more of \hyperlink{SM20}{$\mathsf{SM20}$}-\hyperlink{SM23}{$\mathsf{SM23}$} 
as limits.  For \hyperlink{A22}{$\zTS$} to have \hyperlink{A21}{$\zS$} as its 
limit, we require the transformation
\begin{equation}
g_t\cdot \B = \B, \qquad g_t\cdot H = tH \qquad\text{and}\qquad g_t\cdot \P = \P. 
\end{equation}
Wanting to ensure $[\Q, \Q] \neq 0$, and that the limit $t\to 0$
is well-defined, we need $g_t\cdot \Q = \sqrt{t} \Q$.  Taking this limit,
we find $\hyperlink{SM24}{\mathsf{SM24}_\lambda}\to\hyperlink{SM21}{\mathsf{SM21}}$.

To get \hyperlink{A21}{$\zS$}  from \hyperlink{A23p}{$\RR\times S^3$},
we need the transformation
\begin{equation}
g_t\cdot \B = \B, \qquad g_t\cdot H = H \qquad\text{and}\qquad g_t\cdot \P = t\P. 
\end{equation}
Using this transformation and taking the limit $t\to 0$, we
find $\hyperlink{SM25}{\mathsf{SM25}}\to\hyperlink{SM22}{\mathsf{SM22}}$.
However, the limit is not well-defined for \hyperlink{SM26}{$\mathsf{SM26}$}
due to $\P$ in the expression for $[\Q, \Q]$.  In this case, we additionally
require $g_t\cdot\Q = \sqrt{t} \Q$.  Then $\hyperlink{SM26}{\mathsf{SM26}}
\to \hyperlink{SM20}{\mathsf{SM20}}$.  Another choice of 
transformation,
\begin{equation}
g_t\cdot \B = \B, \qquad g_t\cdot H = tH, \qquad g_t\cdot \P = t\P
\qquad\text{and}\qquad g_t\Q = \sqrt{t}\Q,
\end{equation}
for \hyperlink{SM26}{$\mathsf{SM26}$}, gives \hyperlink{SM23}{$\mathsf{SM23}$}  in the 
limit $t\to 0$.  Thus, we also have $\hyperlink{SM26}{\mathsf{SM26}}
\to\hyperlink{SM23}{\mathsf{SM23}}$.

Finally, to get \hyperlink{A21}{$\zS$}  from \hyperlink{A23m}{$\RR\times H^3$},
we use the transformation
\begin{equation}
g_t\cdot \B = \B, \qquad g_t\cdot H = H, \qquad g_t\cdot \P = t\P. 
\end{equation}
To ensure the limit $t\to 0$ is well-defined, we subsequently need
$g_t\cdot\Q = \sqrt{t} \Q$.  This transformation with the limit gives
$\hyperlink{SM27}{\mathsf{SM27}}\to\hyperlink{SM21}{\mathsf{SM21}}$.

There are only two underlying aristotelian homogeneous spacetimes
which have more than one superisation.  These are
\hyperlink{A21}{$\zS$} and \hyperlink{A23p}{$\RR\times S^3$}.
In the latter case, we find the superisation
\hyperlink{SM25}{$\mathsf{SM25}$} as the limit of \hyperlink{SM26}{$\mathsf{SM26}$}
using the transformation
\begin{equation}
g_t\cdot \B = \B, \qquad g_t\cdot H = tH, \qquad g_t\cdot \P = \P
\qquad\text{and}\qquad g_t\cdot \Q = \sqrt{t} \Q,
\end{equation}
and taking $t \to 0$.  In the former case, the superisations
\hyperlink{SM22}{$\mathsf{SM22}$} and \hyperlink{SM21}{$\mathsf{SM21}$} can be
found as limits of \hyperlink{SM23}{$\mathsf{SM23}$} using the 
transformations
\begin{equation}
g_t\cdot \B = \B, \qquad g_t\cdot H = tH, \qquad g_t\cdot \P = \P 
\qquad\text{and}\qquad g_t\cdot \Q = \sqrt{t} \Q,
\end{equation}
and
\begin{equation}
g_t\cdot \B = \B, \qquad g_t\cdot H = H, \qquad g_t\cdot \P = t\P
\qquad\text{and}\qquad g_t\cdot \Q = \sqrt{t} \Q,
\end{equation}
respectively.  We also have 
\begin{equation}
g_t\cdot \B = \B, \qquad g_t\cdot H = tH, \qquad g_t\cdot \P = \P
\qquad\text{and}\qquad g_t\cdot \Q = \Q,
\end{equation}
giving the limit $\hyperlink{SM20}{\mathsf{SM20}}\to\hyperlink{SM21}{\mathsf{SM21}}$.

\subsubsection{A non-contracting limit}
\label{sec:aristo-non-contracting-lim}

Use the following change of basis on the Lie superalgebra for
\hyperlink{SM24}{$\mathsf{SM24}_\lambda$},
\begin{equation}
\B'=\B, \qquad H' = 2 \lambda^{-1} H, \qquad \P' = \P, 
\qquad \Q' = \Q.
\end{equation}
The brackets then become
\begin{equation}
[H', \sP(\pi)'] = 2 \lambda^{-1} \sP(\pi)', \qquad 
[H', \sQ'(s)] = \sQ'(s(\lambda^{-1} + \kk)), \qquad 
[\sQ'(s), \sQ'(s)] = -\sP'(s\kk\sbar).
\end{equation}
Taking the limits $\lambda \to \pm \infty$, we find the superspace
\hyperlink{SM20}{$\mathsf{SM20}$}.  Therefore, the line of superspaces
\hyperlink{SM24}{$\mathsf{SM24}_\lambda$} compactifies to a circle
with \hyperlink{SM20}{$\mathsf{SM20}$} as the point at infinity.

\subsection{Summary}
\label{sec:limit-summary}

The picture resulting from the above discussion is given in
Figure~\ref{fig:super-limits}.  Except for
$\hyperlink{SM4}{\mathsf{SM3}} \to \hyperlink{SM4}{\mathsf{SM4}}$, the
limits from the families \hyperlink{SM5}{$\mathsf{SM5}_\lambda$},
\hyperlink{SM6}{$\mathsf{SM6}_\lambda$},
\hyperlink{SM7}{$\mathsf{SM7}_{\gamma, \lambda}$},
\hyperlink{SM8}{$\mathsf{SM8}_{\gamma, \lambda}$},
\hyperlink{SM9}{$\mathsf{SM9}_\lambda$} and
\hyperlink{SM11}{$\mathsf{SM11}_\chi$} to
\hyperlink{SM4}{$\mathsf{SM4}$} are not shown explicitly in order to
improve readability.  Neither is the limit between
\hyperlink{SM24}{$\mathsf{SM24}_\lambda$} and
\hyperlink{SM21}{$\mathsf{SM21}$} shown.

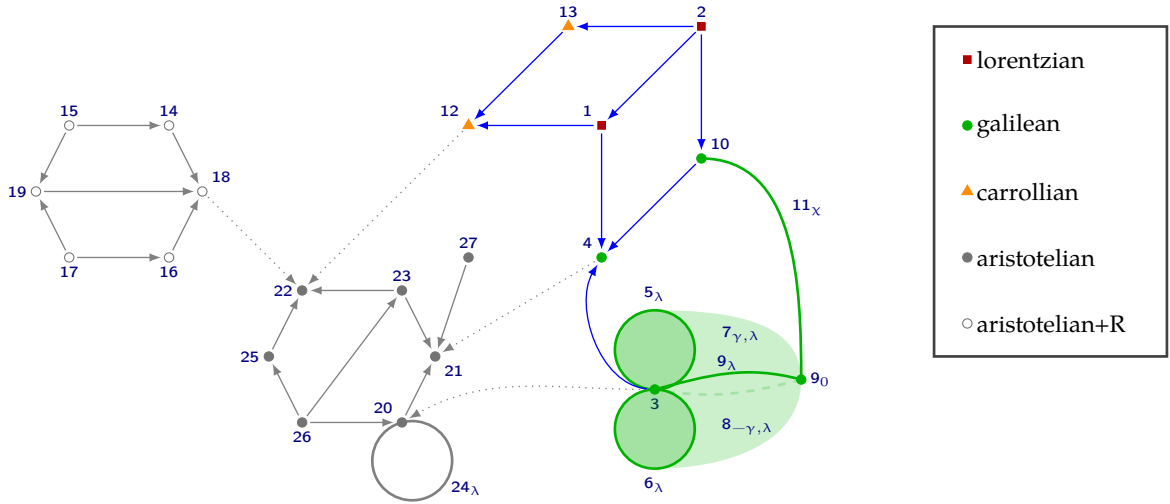
\begin{figure}[h!]
  \centering
  \begin{tikzpicture}[scale=1.75,>=latex, shorten >=3pt, shorten <=3pt, x=1.0cm,y=1.0cm]
    %
    %
    %
    %
    %
    \coordinate [label=above left:{\hyperlink{SM1}{\tiny $\mathsf{1}$}}] (1) at (6.75, 0.75);
    \coordinate [label=above:{\hyperlink{SM2}{\tiny $\mathsf{2}$}}] (2) at (7.5, 1.5);
    \coordinate [label=above left:{\hyperlink{SM4}{\tiny $\mathsf{4}$}}] (4) at  (6.75, -0.25);
    \coordinate [label=above right:{\hyperlink{SM10}{\tiny $\mathsf{10}$}}] (10) at (7.5, 0.5); 
    \coordinate [label=above left:{\hyperlink{SM12}{\tiny $\mathsf{12}$}}] (12) at (5.75,0.75);
    \coordinate [label=above:{\hyperlink{SM13}{\tiny $\mathsf{13}$}}]  (13) at (6.5,1.5);
    \coordinate [label=below:{\hyperlink{SM3}{\tiny $\mathsf{3}$}}] (3) at  (7.15, -1.25); 
    \coordinate [label=above:{\hyperlink{SM5}{\tiny $\mathsf{5}_{\lambda}$}}] (5) at (7.15, -.65);
    \coordinate [label=below:{\hyperlink{SM6}{\tiny $\mathsf{6}_{\lambda}$}}] (6) at (7.15, -1.85);
    \coordinate [label=right:{\hyperlink{SM9}{\tiny $\mathsf{9}_{0}$}}] (9) at (8.25, -1.175);
    \coordinate [label=above:{\hyperlink{SM14}{\tiny $\mathsf{14}$}}]  (14) at (3.5,0.75);
    \coordinate [label=above:{\hyperlink{SM15}{\tiny $\mathsf{15}$}}]  (15) at (2.75,0.75);    
    \coordinate [label=below:{\hyperlink{SM16}{\tiny $\mathsf{16}$}}]  (16) at (3.5,-0.25);
    \coordinate [label=below:{\hyperlink{SM17}{\tiny $\mathsf{17}$}}]  (17) at (2.75,-0.25);
    \coordinate [label=above right:{\hyperlink{SM18}{\tiny $\mathsf{18}$}}]  (18) at (3.75,.25);    
    \coordinate [label=left:{\hyperlink{SM19}{\tiny $\mathsf{19}$}}]  (19) at (2.5,.25);    
    \coordinate [label=above left:{\hyperlink{SM20}{\tiny $\mathsf{20}$}}] (20) at (5.25, -1.5);
    \coordinate [label=below right:{\hyperlink{SM21}{\tiny $\mathsf{21}$}}] (21) at (5.5, -1);
    \coordinate [label=left:{\hyperlink{SM22}{\tiny $\mathsf{22}$}}] (22) at (4.5, -.5);
    \coordinate [label=above:{\hyperlink{SM23}{\tiny $\mathsf{23}$}}] (23) at (5.25, -.5);
    \coordinate [label=left:{\hyperlink{SM25}{\tiny $\mathsf{25}$}}]  (25) at (4.25, -1); 
    \coordinate [label=below:{\hyperlink{SM26}{\tiny $\mathsf{26}$}}]  (26) at (4.5, -1.5);
    \coordinate [label=right:{\hyperlink{SM24}{\tiny $\mathsf{24}_{\lambda}$}}]  (24) at (5.55,-2); 
    \coordinate [label=above:{\hyperlink{SM27}{\tiny $\mathsf{27}$}}]  (27) at (5.75, -0.25);
    \coordinate [label=above:{\hyperlink{SM11}{\tiny $\mathsf{11}_\chi$}}] (11) at (8.3, 0); 
    %
    %
    \path [fill=green!70!black,opacity=0.2, line width=.1mm] (9) to [in=0,out=90] (5) arc (90:270:0.3) to  [opacity=0,in=165, out=15] (9);
    \path [fill=green!70!black,opacity=0.2, line width=.1mm] (9) to [in=0,out=270] (6) arc (270:90:0.3) to  [opacity=0,in=165, out=15] (9);
    \draw[>=latex, shorten >=0pt, shorten <=0pt, line width=1pt, color=green!70!black, fill=green!70!black,fill opacity=.2] (3) arc (90:450:0.3);
    \draw[>=latex, shorten >=0pt, shorten <=0pt, line width=1pt, color=green!70!black, fill=green!70!black,fill opacity=.2] (3) arc (-90:270:0.3); 
    \draw [>=latex, shorten >=0pt, shorten <=0pt, line width=1pt, color=green!70!black] (10) to [in=90,out=0] (9);
    \draw [>=latex, shorten >=0pt, shorten <=0pt, line width=1pt, color=green!70!black] (9) to [in=15,out=165] (3);
    \draw [>=latex, shorten >=0pt, shorten <=0pt, line width=1pt, dashed, opacity = 0.2, color=green!70!black] (9) to [in=350,out=195] (3);
    \draw[>=latex, shorten >=0pt, shorten <=0pt, line width=1pt, color=gray,rotate=-30] (20) arc (-225:135:0.3); 
    %
    %
    \draw [->,line width=0.5pt,dotted,color=gray] (18) -- (22);
    \draw [->,line width=0.5pt,dotted,color=gray] (12) -- (22);
    \draw [->,line width=0.5pt,dotted,color=gray] (4) -- (21);
    \draw [->,line width=0.5pt,dotted,color=gray] (3) to [in=30, out=180] (20);
    %
    %
    \draw [->,line width=0.5pt,color=blue] (2) -- (13);
    \draw [->,line width=0.5pt,color=blue] (2) -- (1);
    \draw [->,line width=0.5pt,color=blue] (2) -- (10);
    \draw [->,line width=0.5pt,color=blue] (13) -- (12);
    \draw [->,line width=0.5pt,color=blue] (1) -- (12);
    \draw [->,line width=0.5pt,color=blue] (1) -- (4);
    \draw [->,line width=0.5pt,color=blue] (10) -- (4); 
    \draw [->,line width=0.5pt,color=blue] (3) to [in=-120,out=175] (4);
    %
    %
    \draw [->,line width=0.5pt,color=gray] (15) to (19); 
    \draw [->,line width=0.5pt,color=gray] (17) to (19); 
    \draw [->,line width=0.5pt,color=gray] (14) to (18); 
    \draw [->,line width=0.5pt,color=gray] (16) to (18); 
    \draw [->,line width=0.5pt,color=gray] (15) to (14); 
    \draw [->,line width=0.5pt,color=gray] (19) to (18); 
    \draw [->,line width=0.5pt,color=gray] (17) to (16); 
    %
    %
    \draw [->,line width=0.5pt,color=gray] (27) to (21); 
    \draw [->,line width=0.5pt,color=gray] (20) to (21);    
    \draw [->,line width=0.5pt,color=gray] (26) to (20); 
    \draw [->,line width=0.5pt,color=gray] (26) to (25); 
    \draw [->,line width=0.5pt,color=gray] (25) to (22); 
    \draw [->,line width=0.5pt,color=gray] (23) to (22); 
    \draw [->,line width=0.5pt,color=gray] (23) to (21); 
    \draw [->,line width=0.5pt,color=gray] (26) to (23);  
    %
    %
    \coordinate [label=below:{\hyperlink{SM7}{\tiny $\mathsf{7}_{\gamma, \lambda}$}}] (7) at (7.8, -0.7);
    \coordinate [label=below:{\hyperlink{SM8}{\tiny $\mathsf{8}_{-\gamma, \lambda}$}}] (8) at (7.85,-1.4);
    \coordinate [label=above:{\hyperlink{SM9}{\tiny $\mathsf{9}_{\lambda}$}}] (9_2) at (7.7,-1.17);
    %
    %
    \foreach \point in {3, 4, 9, 10}
    \filldraw [color=green!70!black,fill=green!70!black] (\point) circle (1pt);
    \foreach \point in {1, 2}
   \node[lorentzian] at (\point) {};
    \foreach \point in {12,13}
    \node[carrollian] at (\point) {}; 
    \foreach \point in {14, 15, 16, 17, 18, 19}
    \draw [color=gray!90!black] (\point) circle (1pt);
    \foreach \point in {20, 21, 22, 23, 25, 26, 27}
    \draw [color=gray!90!black,fill=gray!90!black] (\point) circle (1pt);
    %
    %
    \begin{scope}[xshift=-1.5cm]
    \draw [line width=1pt,color=gray!50!black] (10.75,-1) rectangle (12.5,1.5);
    \node[lorentzian] at (11,1.25) {};
    \draw (11,1.25) node[color=black,anchor=west] {\small lorentzian};
    \filldraw [color=green!70!black,fill=green!70!black] (11,0.75) circle (1pt) node[color=black,anchor=west] {\small galilean};
    \node[carrollian] at (11,0.25) {};
    \draw (11,0.25) node[color=black,anchor=west] {\small carrollian};
    \draw [color=gray!90!black,fill=gray!90!black] (11,-0.25) circle (1pt) node[color=black,anchor=west] {\small aristotelian};       
    \draw [color=gray!90!black] (11, -0.75) circle (1pt) node[color=black,anchor=west] {\small aristotelian+R}; 
    \end{scope}
  \end{tikzpicture}
  \caption{Homogeneous superspaces and their limits.\\
    (Numbers are hyperlinked to the corresponding superspaces in
    Table~\ref{tab:superspaces}.)}
  \label{fig:super-limits}
\end{figure}

For comparison, we extract from
\cite[Fig.3]{Figueroa-OFarrill:2018ilb} the subgraph corresponding to
spacetimes which admit superisations and show it in
Figure~\ref{fig:sub-limits}.  There are arrows between these two
pictures: taking a superspace to its corresponding spacetime, but
making this explicit seems beyond our combined artistic abilities.

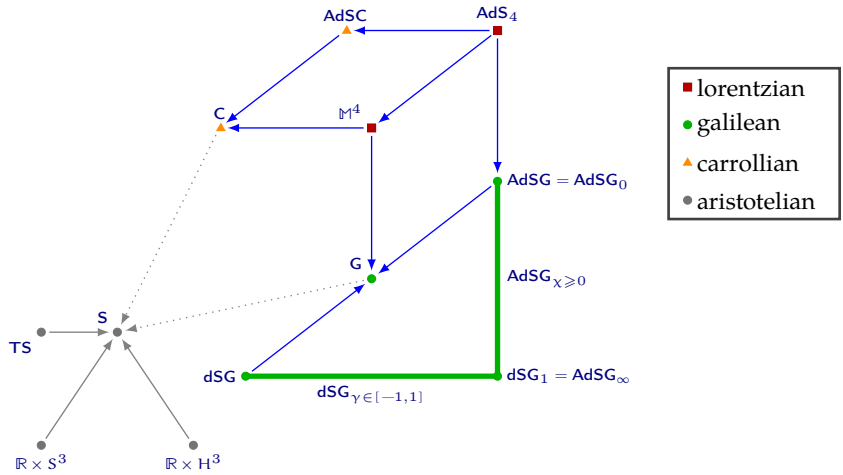
\begin{figure}[h!]
  \centering
  \begin{tikzpicture}[scale=1,>=latex, shorten >=3pt, shorten <=3pt, x=1.0cm,y=1.0cm]
    %
    %
    %
    %
    \coordinate [label=left:{\hyperlink{S8}{\tiny $\zdSG$}}] (dsg) at (5.688048519056286, -2.5838170592960186);
    \coordinate [label=right:{\hyperlink{S9}{\tiny $\ztdSG_1 =\ztAdSG_\infty$}}] (dsgone) at (9, -2.5838170592960186);
    \coordinate [label=above:{\hyperlink{S13}{\tiny $\zC$}}] (c) at (5.344024259528143, 0.7080914703519907);
    \coordinate [label=above left:{\hyperlink{S1}{\tiny $\MM^4$}}] (m) at (7.344024259528143,0.7080914703519907);
    \coordinate [label=above left:{\hyperlink{S7}{\tiny $\zG$}}] (g) at  (7.344024259528143, -1.2919085296480093);
    \coordinate [label=above:{\hyperlink{S3}{\tiny $\zAdS_4$}}] (ads) at (9,2);
    \coordinate [label=right:{\hyperlink{S10}{\tiny $\zAdSG =\ztAdSG_0$}}] (adsg) at (9,0);
    \coordinate [label=above:{\hyperlink{S15}{\tiny $\zAdSC$}}]  (adsc) at (7,2);
    \coordinate [label=above left:{\hyperlink{A21}{\tiny $\zS$}}] (s) at (4, -2);
    \coordinate [label=below left:{\hyperlink{A22}{\tiny $\zTS$}}]  (ts) at (3, -2); 
    \coordinate [label=below:{\hyperlink{A23p}{\tiny $\RR\times S^3$}}]  (esu) at (3, -3.5); 
    \coordinate [label=below:{\hyperlink{A23m}{\tiny $\RR\times H^3$}}]  (hesu) at (5, -3.5); 
    %
    %
    \coordinate [label=below:{\hyperlink{S9}{\tiny $\ztdSG_{\gamma\in[-1,1]}$}}] (tdsg) at (7.344024259528143, -2.5838170592960186);
    \coordinate [label=right:{\hyperlink{S11}{\tiny $\ztAdSG_{\chi\geq0}$}}] (tadsg) at (9, -1.2919085296480093);
    %
    %
    \draw [->,line width=0.5pt,dotted,color=gray] (c) -- (s);
    \draw [->,line width=0.5pt,dotted,color=gray] (g) -- (s);
    %
    %
    \draw [->,line width=0.5pt,color=blue] (adsc) -- (c);
    \draw [->,line width=0.5pt,color=blue] (ads) -- (m);
    \draw [->,line width=0.5pt,color=blue] (adsg) -- (g);
    \draw [->,line width=0.5pt,color=blue] (dsg) -- (g);
    \draw [->,line width=0.5pt,color=blue] (m) -- (c);
    \draw [->,line width=0.5pt,color=blue] (m) -- (g);
    \draw [->,line width=0.5pt,color=blue] (ads) -- (adsc);
    \draw [->,line width=0.5pt,color=blue] (ads) -- (adsg);
    %
    %
    \draw [->,line width=0.5pt,color=gray] (ts) to (s); 
    \draw [->,line width=0.5pt,color=gray] (esu) to (s); 
    \draw [->,line width=0.5pt,color=gray] (hesu) to (s); 
    %
    %
    \begin{scope}[>=latex, shorten >=0pt, shorten <=0pt, line width=2pt, color=green!70!black]
      \draw (adsg) --(dsgone);
      \draw (dsg) -- (dsgone);
    \end{scope}
    \foreach \point in {g,adsg,dsg,dsgone}
    \filldraw [color=green!70!black,fill=green!70!black] (\point) circle (1.5pt);
    \foreach \point in {ads,m}
    \filldraw [color=red!70!black,fill=red!70!black] (\point) ++(-1.5pt,-1.5pt) rectangle ++(3pt,3pt);
    \foreach \point in {adsc,c}
    \filldraw [color=DarkOrange,fill=DarkOrange] (\point) ++(-1pt,-1pt) -- ++(3pt,0pt) -- ++(-1.5pt,2.6pt) -- cycle;
    \foreach \point in {s,ts,esu,hesu}
    \filldraw [color=gray!90!black] (\point) circle (1.5pt);
    %
    %
    \begin{scope}[xshift=0.5cm]
    \draw [line width=1pt,color=gray!50!black] (10.75,-0.5) rectangle (13,1.5);
    \filldraw [color=red!70!black,fill=red!70!black] (11,1.25) ++(-1.5pt,-1.5pt) rectangle ++(3pt,3pt) ; 
    \draw (11,1.25) node[color=black,anchor=west] {\small lorentzian}; 
    \filldraw [color=green!70!black,fill=green!70!black] (11,0.75) circle (1.5pt) node[color=black,anchor=west] {\small galilean};
    \filldraw [color=DarkOrange,fill=DarkOrange] (11,0.25) ++(-1.5pt,-1pt) -- ++(3pt,0pt) -- ++(-1.5pt,2.6pt) -- cycle;
    \draw (11,0.25) node[color=black,anchor=west] {\small carrollian};
    \filldraw [color=gray!90!black] (11,-0.25) circle (1.5pt) node[color=black,anchor=west] {\small aristotelian};       
    \end{scope}
  \end{tikzpicture}
  \caption{Limits between superisable spacetimes}
  \label{fig:sub-limits}
\end{figure}

Nevertheless, interpreting Figures~\ref{fig:super-limits} and
\ref{fig:sub-limits} as posets, with arrows defining the partial
order, the map taking a superspace to its underlying spacetime is
surjective by construction (we consider only superisable spacetimes)
and order preserving, as shown at the start of this section.  As can
be gleaned from Table~\ref{tab:superspaces}, the fibres of this map
are often quite involved, clearly showing the additional ``internal''
structure in the superspace which allows for more than one possible
superisation of a spacetime.

We should mention that despite appearances, superspaces
$\hyperlink{SM3}{\mathsf{SM3}}$ and $\hyperlink{SM4}{\mathsf{SM4}}$
share the same underlying spacetime: namely, the galilean spacetime
$\hyperlink{S7}{\zG}$.  Notice that superspaces
$\hyperlink{SM21}{\mathsf{SM21}}$ and
$\hyperlink{SM22}{\mathsf{SM22}}$, which are ``terminal'' in the
partial order, correspond to the static aristotelian spacetime
$\hyperlink{A21}{\mathsf{\zS}}$.  With the exception of
$\lim_{\chi\to \infty} \hyperlink{SM11}{\mathsf{SM11}_\chi} =
\hyperlink{SM9}{\mathsf{SM9}_0}$, all other non-contracting limits
between superspaces induce limits between the underlying spacetimes
which arise from contractions of the kinematical Lie algebras: the
limits $|\lambda|\to \infty$ of
\hyperlink{SM5}{$\mathsf{SM5}_\lambda$} and
\hyperlink{SM6}{$\mathsf{SM6}_\lambda$} induce the contraction
$\hyperlink{S8}{\zdSG} \to \hyperlink{S7}{\zG}$, whereas the limits
$|\lambda|\to \infty$ of
\hyperlink{SM7}{$\mathsf{SM7}_{\gamma, \lambda}$},
\hyperlink{SM8}{$\mathsf{SM8}_{\gamma, \lambda}$} and
\hyperlink{SM9}{$\mathsf{SM9}_\lambda$} induce the contractions
$\hyperlink{S9}{\zdSG_\gamma} \to \hyperlink{S7}{\zG}$, where
$\gamma = 1$ for \hyperlink{SM9}{$\mathsf{SM9}_\lambda$}.

\section{Conclusions}
\label{sec:conclusions}

In this paper, we have answered the question: \emph{What are
  the possible super-kinematics?} by classifying ($N{=}1$ $d{=}4$)
kinematical Lie superalgebras and their corresponding superspaces.

The Lie superalgebras were classified by solving the Jacobi identities
in a quaternionic reformulation, which made the computations no harder
than multiplying quaternions and paying close attention to the action
of automorphisms in order to ensure that there is no repetition in our
list.  Since we are interested in supersymmetry, we focussed on Lie
superalgebras where the supercharges were not abelian: i.e., we demand
that $[\Q,\Q] \neq 0$ and, subject to that condition, we classified
Lie superalgebras which extend either kinematical or aristotelian Lie
algebras.  The results are contained in Tables~\ref{tab:klsa} and
\ref{tab:alsa}, respectively.

There are two salient features of these classifications.  Firstly, not
every kinematical Lie algebra admits a supersymmetric extension: in
some cases because of our requirement that $[\Q,\Q] \neq 0$, but in
other cases (e.g., $\so(5)$, $\so(4,1)$,...) because the
four-dimensional spinor representation of $\so(3)$ does not extend to
a representation of these Lie algebras.

Secondly, some kinematical Lie algebras admit more than one
non-isomorphic supersymmetric extension. For example, the galilean Lie
algebra admits two supersymmetric extensions, but only one of them
($\hyperlink{KLSA8}{\mathsf{S8}}$) can be obtained as a contraction
of $\osp(1|4)$.  By far most of the Lie superalgebras in our
classification cannot be so obtained and hence are not listed in
previous classifications.  Nevertheless, our ``moduli space'' of Lie
superalgebras is connected, if not always by contractions.  For
example, the other supersymmetric extension of the galilean algebra
($\hyperlink{KLSA7}{\mathsf{S7}}$) can be obtained as a
non-contracting limit of some of the multi-parametric families of Lie
superalgebras in the limit as one of the parameters goes to
$\pm \infty$, in effect compactifying one of the directions in the
parameter space into a circle.

We classified the corresponding superspaces via their super Lie pairs
$(\s,\h)$, where $\s$ is a kinematical Lie superalgebra and $\h$ an
admissible subalgebra.  Every such pair ``superises'' a pair
$(\k,\h)$, where $\k = \s_{\bar 0}$ is a kinematical Lie algebra.  As
shown in \cite{Figueroa-OFarrill:2018ilb}, effective and geometrically
realisable pairs $(\k,\h)$ are in bijective correspondence with
simply-connected homogeneous spacetimes, and hence the super Lie pairs 
$(\s,\h)$ are in bijective correspondence with superisations of such
spacetimes.  These are listed in Table~\ref{tab:superspaces}.

There are several salient features of that table.  Firstly, many
spacetimes admit more than one inequivalent superisation.  Whereas
Minkowski and AdS spacetimes admit a unique ($N{=}1$) superisation,
and so too do the (superisable) carrollian spacetimes, many of the
galilean spacetimes admit more than one and in some cases even a
circle of superisations.

Secondly, there are effective super Lie pairs $(\s,\h)$ for which the
underlying pair $(\k,\h)$ is not effective.  This means that the
``boosts'' act trivially on the underlying spacetime, but nontrivially
in the superspace: in other words, the ``boosts'' are actually
R-symmetries.  Since $(\k,\h)$ is not effective, this means that it
describes an aristotelian spacetime and this gives rise to the class
of aristotelian superspaces with R-symmetry.

Thirdly, there are three superspaces in our list which also appear in
\cite{deMedeiros:2016srz}: namely, Minkowski
($\hyperlink{SM1}{\mathsf{SM1}}$) and AdS
($\hyperlink{SM2}{\mathsf{SM2}}$) superspaces, but also the
aristotelian superspace $\hyperlink{SM26}{\mathsf{SM26}}$, whose
underlying manifold appears in \cite{deMedeiros:2016srz} as the
lorentzian Lie group $\RR \times \SU(2)$ with a bi-invariant
metric.

Lastly, just like Minkowski ($\hyperlink{S1}{\MM^4}$) and carrollian
AdS ($\hyperlink{S15}{\zAdSC}$) spacetimes are homogeneous under the
Poincaré group, their (unique) superisations
($\hyperlink{SM1}{\mathsf{SM1}}$ and
$\hyperlink{SM13}{\mathsf{SM13}}$, respectively) are homogeneous under
the Poincaré supergroup.  This suggests a sort of correspondence or
duality, which we hope to explore in future work.

There are a number of natural extensions to the results in this paper,
which we list in no particular order.  It would be interesting to
classify extended $N{>}1$ superalgebras and superspaces in four
dimensions and also kinematical/aristotelian superalgebras and
superspaces in other dimensions: particularly in three-dimensions due
to their use in Chern--Simons theories (see, e.g.,
\cite{Matulich:2019cdo}).  In the three-dimensional case, it
would be important to determine the possible central charges and also
the existence of invariant inner products.  It would also be
interesting to classify superconformal algebras along the lines of
\cite{Figueroa-OFarrill:2018ygf}, which at least in four dimensions
would be amenable to the quaternionic formalism employed in this
paper.  There has been a great deal of work on Schrödinger
superalgebras, departing from the pioneering work in
\cite{Duval:1993hs}.

As shown in Tables \ref{tab:klsa} and \ref{tab:alsa}, many of these
Lie superalgebras are graded and hence can serve as the starting
ingredient to explore its filtered deformations, as advocated in
\cite{deMedeiros:2016srz,deMedeiros:2018ooy}; perhaps allowing
us to go from the homogeneous models classified in this paper to more
general superspaces.

The underlying spacetimes of the superspaces in
Table~\ref{tab:superspaces} are reductive and hence possess a
canonical invariant connection.  It is a natural question to ask
whether the kinematical superalgebras admit an interpretation as
Killing superalgebras in the spacetimes; that is, whether they are
generated by ``spinor'' fields relative to some connection modifying
the canonical invariant connection.  In fact, as proved in
\cite[§5]{MR2640006} in the context of spin manifolds, this is indeed
the case (see Definition~5.3 in \cite{MR2640006} for the notion of a
generalised Killing spinor).

Finally, along the lines of \cite[§4]{MR2640006}, we could investigate
the invariant connections in the superspaces in
Table~\ref{tab:superspaces}, by determining the space of Nomizu
maps, as was done in \cite{Figueroa-OFarrill:2019sex} for the
homogeneous spacetimes.

\section*{Acknowledgments}

We are very grateful to Andrea Santi for answering many of our
questions about homogeneous supermanifolds and Stefan Prohazka for
comments on a previous version of this paper and which we hope to have
incorporated in a way that has hopefully improved the exposition.  We
are also grateful to an anonymous referee for pointing out that we had
missed some explanation concerning Table~\ref{tab:spacetimes}.  JMF
was visiting the University of Stavanger during the final stages of
writing this paper and he's grateful to Paul de Medeiros for lending
his ear.  He's also grateful to Paul and Sigbjørn Hervik for the
invitation to visit and the hospitality.

\appendix

\section{Lorentzian superspaces}
\label{app:lorentzian}

In this appendix we give the definitions of the lorentzian
superspaces, in a way that is as agnostic as possible about
conventions.  These are precisely the superspaces which also appear in
\cite{deMedeiros:2016srz}, since their supersymmetry algebras are
filtered deformations of subalgebras of the Poincaré superalgebra.

\subsection{Minkowski superspace}
\label{sec:minkowski-superspace}

The ur-example is, of course, Minkowski superspace
(\hyperlink{SM1}{$\mathsf{SM1}$}), which is a homogeneous space of the
Poincaré supergroup and can be described by a pair $(\s,\h)$ as
follows. The kinematical Lie superalgebra $\s$ is the $N{=}1$ Poincaré
superalgebra, which is defined as follows. Let $(V,\eta)$ be a
lorentzian (``mostly minus'') four-dimensional vector space and let
$\so(V)$ denote the skew-symmetric endomorphisms of $V$; that is,
linear maps $\varphi: V \to V$ such that
$\eta(\varphi(v),w) = - \eta(v,\varphi(w))$ for all $v,w\in V$. Let
$\Cl(V)$ denote the corresponding Clifford algebra, with Clifford
relation $v \cdot v = -\eta(v,v) \id$, for all $v \in V$. As a real
associative algebra, $\Cl(V) \cong \End(S)$, where $S$ is a real
four-dimensional irreducible Clifford module. It is also an
irreducible representation (``Majorana spinors'') of
$\so(V) \subset \Cl(V)$; although its complexification (``Dirac
spinors'') decomposes into positive- and negative-chirality
irreducible representations (``Weyl spinors''). On $S$ there is a
symplectic inner product $\left<-,-\right>$ satisfying
\begin{equation}\label{eq:spinorIP}
  \left<v \cdot s_1, s_2\right> = - \left<s_1, v \cdot s_2\right>,
\end{equation}
for all $s_1,s_2 \in S$ and $v \in V$, where $\cdot$ denotes the
Clifford action.  This implies that $\left<-,-\right>$ is
$\so(V)$-invariant.  We define a $\ZZ$-graded vector space
$\s = \s_0 \oplus \s_{-1} \oplus \s_{-2}$, with $\s_0 = \so(V)$,
$\s_{-1} = S$ and $\s_{-2} = V$.  Let
$\s_{\bar 0} = \s_0 \oplus \s_{-2}$ and $\s_{\bar 1} = \s_1$ and we
define on the vector superspace $\s = \s_{\bar 0} \oplus \s_{\bar 1}$
the structure of a Lie superalgebra as follows.  The Lie algebra
structure on $\s_{\bar 0}$ is the Poincaré algebra:
\begin{equation}
  [(A,v), (B,w)] = (AB - BA, A(w) - B(v)),
\end{equation}
or equivalently,
\begin{equation}
  [A, B] = AB - BA, \quad [A,v] = A(v) = - [v,A]\qquad\text{and}\qquad
  [v,w] = 0,
\end{equation}
for $A,B \in \so(V)$ and $v,w\in V$.  We make $\s_{\bar 1}$ into an
$\s_{\bar 0}$-module by declaring $\so(V)$ to act via the spinor
representation and $V$ to act trivially.  Finally, if $s_1,s_2\in
\s_{\bar1}$, their bracket $[s_1,s_2] \in V$ is defined to be the
vector such that, for all $v \in V$,
\begin{equation}
  \eta([s_1,s_2], v) = \left<s_1, v  \cdot s_2\right>,
\end{equation}
which is symmetric by equation~\eqref{eq:spinorIP} and the fact that
$\left<-,-\right>$ is symplectic.  The bracket defines a symmetric
bilinear map $\s_{\bar 1} \times \s_{\bar 1} \to \s_{\bar 0}$ or,
equivalently,  a linear map $\bigodot^2 \s_{\bar 1} \to V \subset
\s_{\bar 0}$ from the symmetric tensor square of $\s_{\bar 1}$.  This
map is surjective and, moreover, $\so(V)$-equivariant
because $\eta$, $\left<-,-\right>$ are $\so(V)$-invariant and Clifford
action is $\so(V)$-equivariant.  The Jacobi identity $[[s,s],s] = 0$
is trivially satisfied because $[s,s] \in V$ and $V$ acts trivially on
$S$.  This defines the Poincaré superalgebra $\s$.  The admissible
subalgebra $\h = \so(V)$ is the Lie subalgebra of Lorentz
transformations, and Minkowski superspace is described by the pair
$(\s,\h)$.  The pair $(\s_{\bar 0},\h)$ defines a homogeneous
spacetime, which is none other than Minkowski spacetime
$\hyperlink{S1}{\MM^4}$.

\subsection{Anti de Sitter superspace}
\label{sec:anti-de-sitter}

The second well-known example is anti de Sitter superspace
(\hyperlink{SM2}{$\mathsf{SM2}$}), whose associated kinematical Lie
superalgebra is isomorphic to $\osp(1|4)$ and whose construction we
now review.  The spin representation of $\so(3,2)$ defines an
isomorphism $\so(3,2) \to \sp(4,\RR)$.  This means that the spinor
representation $S$ is real, symplectic and four-dimensional.  Let
$\left<-,-\right>$ denote the $\so(3,2)$-invariant symplectic inner
product on $S$:
$\left<X \cdot s_1, s_2\right> = - \left<s_1, X \cdot s_2\right>$ for
all $s_1,s_2 \in S$ and $X \in \so(3,2)$.  Let $\kappa$ denote the
Killing form on $\so(3,2)$, which is nondegenerate because $\so(3,2)$
is simple.  Define a bilinear map $[-,-]: S \times S \to \so(3,2)$ by
declaring $[s_1,s_2] \in \so(3,2)$ to be the unique element whose
inner product (relative to the Killing form) with any $X \in \so(3,2)$
is given by
\begin{equation}
  \kappa([s_1,s_2],X) = \left<s_1, X\cdot s_2\right>,
\end{equation}
which is symmetric by the $\so(3,2)$-invariance of the symplectic
structure.  Define a vector superspace $\s = s_{\bar 0}
\oplus \s_{\bar 1}$, with $\s_{\bar 0} = \so(3,2)$ and $\s_{\bar 1} =
S$ and an even bracket on $\s$ by taking it to be the Lie bracket on
$\s_{\bar 0}$, the action of $\so(3,2)$ on $S$  and the above map
$\bigodot^2 S \to \so(3,2)$.  The Jacobi identity follows from the
fact that $\s_{\bar 0}$ is a Lie algebra, $\s_{\bar 1}$ is an
$\s_{\bar 0}$-module, the bracket $\bigodot^2 \s_{\bar 1} \to \s_{\bar
  0}$ is $\s_{\bar 0}$-equivariant (since $\kappa$ and
$\left<-,-\right>$ are $\s_{\bar 0}$-invariant) and because the only
$\so(3,2)$-equivariant linear map $\bigodot^3 S \to S$ is the zero map.
Notice that $[S,S]$ is a nonzero ideal of $\so(3,2)$, but since
$\so(3,2)$ is simple, this is all of $\so(3,2)$.  The resulting simple
Lie superalgebra is isomorphic to $\osp(1|4)$.  We may take for the admissible
subalgebra $\h$ the stabiliser in $\so(3,2)$ of any timelike vector in
$\RR^{3,2}$, which is isomorphic to $\so(3,1) \subset \so(3,2)$.  The
pair $(\s,\h) = (\osp(1|4), \so(3,1))$ defines a homogeneous
superspace whose underlying homogeneous spacetime $(\s_{\bar 0},\h) =
(\so(3,2),\so(3,1))$ is of course anti de Sitter spacetime
$\hyperlink{S3}{\zAdS_4}$.

\subsection{Einstein static superspace}
\label{sec:einst-stat-supersp}

The third and final example of a lorentzian superspace in our
classification is the aristotelian superspace
\hyperlink{SM26}{$\mathsf{SM26}$}, which is one of the superisations
of the Einstein static universe \hyperlink{A23p}{$\RR \times S^3$}.
We shall be brief and refer to \cite[Thm.~14]{deMedeiros:2016srz} for
the details, particularly equation~(98) in that paper, except that
what we call $\h$ in that paper is not the admissible subalgebra as in
this paper, but actually the rotational subalgebra $\r$.  The notation
is as in the case of the Minkowski superspace treated above:
$(V,\eta)$ a ``mostly minus'' lorentzian four-dimensional vector space
and $S$ the real four-dimensional irreducible $\Cl(V)$-module.  Pick a
nonzero timelike vector $\varphi \in V$, whose stabiliser in $\so(V)$
is the rotational subalgebra $\r$.  Define $\widetilde\psi : V \to \r$
by $\widetilde\psi(v) = 2 \imath_v \imath_\varphi \vol\in \wedge^2 V
\cong \so(V)$.  Since for $w \in V$, $\widetilde\psi(v) w = 2 \imath_w
\imath_v \imath_\varphi \vol$, we see that $\widetilde\psi(v) \varphi
= 0$ and hence $\widetilde\psi(v) \in \r$ for all $v \in V$ as
claimed.  Now let $A,B \in \r$, $v,w \in V$ and $s \in S$.  The Lie
brackets $[A,B]$, $[A,s]$, $[A,v]$ and $[s,s]$ are exactly as in the
Poincaré superalgebra, whereas
\begin{equation}
  [v,w] = \widetilde\psi(v)w - \widetilde\psi(w)v
  \qquad\text{and}\qquad [v,s] = -\tfrac12 (v \cdot \varphi + 3
  \varphi \cdot v) \cdot s + \widetilde\psi(v) s.
\end{equation}
Let us choose a pseudo-orthonormal basis $(\be_0, \be_1, \be_2,
\be_3)$ for $V$ and take $\varphi = \be_0$.  Then we have that
$\widetilde\psi(\be_0) = 0$ and so $[\be_0,\be_i]=0$ for all
$i=1,2,3$.  It follows that $\widetilde\psi(\be_i) = -2 \epsilon_{ijk}
\be_j \wedge \be_k$ and hence $[\be_i, \be_j] = 4 \epsilon_{ijk}
\be_k$.  Acting on $s \in S$, $[\be_0, s] = 2 \vol s$.
Calculating from the above formula, $[\be_i, s] = - \epsilon_{ijk}
(\be_j \wedge \be_k) \cdot s$.  Letting $P_i = \tfrac12 \be_i - J_i$,
we find that $[P_i,P_j] = \epsilon_{ijk} J_k$ and that $P_i$ and $J_i$
act in the same way on $S$.  Choosing $\be_0 = -2 H$ and $\be_i = 2
(J_i + P_i)$, we find that (rescaling $s$) the $[s,s]$ bracket is
precisely the one in \hyperlink{SM26}{$\mathsf{SM26}$}.

\providecommand{\href}[2]{#2}\begingroup\raggedright\endgroup


\end{document}